\documentclass[12pt]{article}
\RequirePackage{amsthm,amsmath,amsfonts,amssymb}
\usepackage{mathrsfs,braket,enumitem,subcaption,float,xcolor,bbm,setspace}
\usepackage{caption}
\usepackage{subcaption}
\usepackage[flushleft]{threeparttable}
\usepackage{algorithm,algpseudocode}
\RequirePackage[colorlinks,linkcolor=blue,citecolor=blue,urlcolor=blue]{hyperref}
\usepackage{graphicx}
\usepackage{natbib}
\usepackage{url} 

\newcommand{\R}{\mathbb{R}}
\newcommand{\C}{\mathbb{C}}
\newcommand{\N}{\mathbb{N}}
\newcommand{\Z}{\mathbb{Z}}
\newcommand{\T}{\mathbf{T}}
\renewcommand{\Re}{\operatorname{Re}}
\renewcommand{\Im}{\operatorname{Im}}

\theoremstyle{plain}
\newtheorem{lem}{Lemma}
\newtheorem{thm}{Theorem}
\newtheorem{prop}{Proposition}
\newtheorem{cor}{Corollary}

\theoremstyle{definition}
\newtheorem{definition}{Definition}
\newtheorem{remark}{Remark}
\newtheorem{assumption}{Assumption}

\usepackage[sectionbib]{bibunits}
\defaultbibliographystyle{plainnat}
\defaultbibliography{bib}

\def\pageoption{3}
\def\Title{Weighted shape-constrained estimation for the autocovariance sequence from a reversible Markov chain}
\begin{document}

\def\spacingset#1{\renewcommand{\baselinestretch}%
{#1}\small\normalsize} \spacingset{1}

\ifnum\pageoption=2\begin{bibunit}\fi

\title{\bf \Title}
  \author{Hyebin Song\thanks{Both authors contributed equally.} \, and Stephen Berg\footnotemark[1]  \thanks{Corresponding author: sqb6128@psu.edu} \\
Department of Statistics, Pennsylvania State University}
\maketitle
\bigskip
\begin{abstract}
We present a novel weighted $\ell_2$ projection method for estimating autocovariance sequences and spectral density functions from reversible Markov chains. \citet{berg2023efficient} introduced a least-squares shape-constrained estimation approach for the autocovariance function by projecting an initial estimate onto a shape-constrained space using an $\ell_2$ projection. While the least-squares objective is commonly used in shape-constrained regression, it can be suboptimal due to correlation and unequal variances in the input function. To address this, we propose a weighted least-squares method that defines a weighted norm on transformed data. Specifically, we transform an input autocovariance sequence into the Fourier domain and apply weights based on the asymptotic variance of the sample periodogram, leveraging the asymptotic independence of periodogram ordinates. Our proposal can equivalently be viewed as estimating a spectral density function by applying shape constraints to its Fourier series. We demonstrate that our weighted approach yields strongly consistent estimates for both the spectral density and the autocovariance sequence. Empirical studies show its effectiveness in uncertainty quantification for Markov chain Monte Carlo estimation, outperforming the unweighted moment LS estimator and other state-of-the-art methods.
\end{abstract}

\noindent%
\vfill

\newpage
\spacingset{1.5}  
\section{Introduction}\label{sec: 1_intro}
Markov chain Monte Carlo (MCMC) methods have been essential tools for statistical inference in complex models where analytical solutions or direct sampling from a target distribution $\pi$ are intractable. By constructing a suitable Markov chain $X=X_0,X_1,X_2,\dots$, which converges to the target distribution, we can estimate quantities of interests which are functions of the target distribution. For instance, suppose one is interested in estimating the expectation of $g$ with respect to the target distribution, i.e., $\mu_g = \int g(x) \pi(dx)$ for some $g:\mathsf{X} \to \R$. The Markov chain sample mean $Y_M = M^{-1}\sum_{t=0}^{M-1} g(X_t)$ estimates $\mu_g$.  

In general, the $g(X_t)$ are not independent, and it is of interest to estimate the covariance structure between the draws. Estimating this allows one to assess the degree of correlation between the draws and evaluate the quality of estimates based on MCMC simulations~\citep{geyer1992practical,jones2006fixed}. The autocovariance function $\gamma:\Z \to \R$, defined as $\gamma(k) = \operatorname{Cov}_\pi(g(X_0), g(X_k))$ for $k \in \Z$,  measures dependence between observations $g(X_t)$ at different time lags.  For instance, the estimated autocorrelation function (ACF) plot is often used as a visual tool to assess the degree of dependence in draws from an MCMC sampler. Typically, the empirical autocovariance sequence $\tilde{r}_M$, defined as 
\begin{align}\label{def:empirical_autocov}
    \tilde{r}_M(k)=  \frac{1}{M} \sum_{t=0}^{M-1-|k|}\left\{g(X_t)-Y_M\right\}\left\{g(X_{t+|k|})-Y_M\right\} 
\end{align}
for $k$  such that $|k| < M$ and $\tilde{r}_M(k)=0$ otherwise, is used as an estimate for $\gamma$, and is pointwise consistent~(e.g., \citealt{anderson1971statistical}). 

Moreover, the autocovariance sequence $\gamma$ is closely related to the asymptotic variance $\sigma^2$ of the sample mean, defined as
\begin{align}
    \sigma^2:=\underset{M\to\infty}{\lim}\,M\,\text{Var}(M^{-1}\sum_{t=0}^{M-1}g(X_t)).\label{eq:mcmcavar}
\end{align}
The square root $\sigma$ of $\sigma^2$ in~\eqref{eq:mcmcavar} is also called the Monte Carlo Standard Error (MCSE). Under mild regularity conditions (e.g., ~\citealp{Haggstrom2007-sy}), the limit in the definition of $\sigma^2$ exists, and it turns out that $\sigma^2$ is the infinite sum of $\gamma$ over all lags, i.e., $\sigma^2=\sum_{k\in\mathbb{Z}}\gamma(k)$.  Estimating $\sigma^2$ has an important practical application, as $\sigma^2$ can be used to quantify the precision of the Monte Carlo average $Y_M$ relative to the posterior mean $\mu_g$. Note that since $\sigma^2$ depends on $\gamma$ from all lags, pointwise consistency is not sufficient to guarantee asymptotic variance convergence. For instance, while each $\tilde{r}_M(k)$ is consistent for $\gamma(k)$, the sum $\sum_{k\in\Z} \tilde{r}_M(k)$ is not a consistent estimate for $\sigma^2$. Therefore, alternative estimators have been proposed, including batch-mean and overlapping batch mean estimators, spectral variance estimators, and initial sequence estimators \citep{meketon1984overlapping, glynn1990simulation, glynn1991estimating, Damerdji1991-oj,geyer1992practical,damerdji1994strong,damerdji1995mean,flegal2010batch,vats2022lugsail}.

In the time series literature, it is also common to estimate the closely related quantity, the spectral density \begin{align}
    \phi_{\gamma}(\omega)=\sum_{k\in\mathbb{Z}}e^{-ik\omega}\gamma(k)=\sum_{k\in\mathbb{Z}}\gamma(k)\cos(\omega k) + i\sum_{k\in\mathbb{Z}}\gamma(k)\sin(\omega k),\label{eq:spectralDensity}
\end{align} which is the Fourier transform of the autocovariance sequence.  The spectral density characterizes how the power of a time series signal is distributed across different frequencies, making it a crucial quantity to estimate for understanding the cyclic properties of a signal in time series analysis. Spectral properties of stationary stochastic processes have been studied since~\citet{wiener1930generalized} and~\citet{khintchine1934korrelationstheorie}. 
Non-parametric estimation of the spectral density function~\eqref{eq:spectralDensity} using the periodogram, which is the Fourier transform of the empirical autocovariance function $\tilde{r}_M$, dates back to~\citet{kendall1945analysis},~\citet{bartlett1946theoreticalff}, and~\citet{bartlett1948smoothing}, and has since been extensively researched, with a focus on smoothing techniques using suitable kernel and bandwith choices \citep{wahba1980automatic, pawitan1994nonparametric,kooperberg1995rate, fan1998automatic,mcmurry2004nonparametric}.

There is a notable connection between asymptotic variance estimation and spectral density estimation. From the definition of the spectral density \eqref{eq:spectralDensity}, evaluating the spectral density function at $\omega=0$ yields the asymptotic variance $\sum_{k\in \Z}\gamma(k)$, also known as the long run variance in time series literature. This connection has been leveraged to devise consistent estimators for the long run variance or asymptotic variance in previous literature \citep{heidelberger1981spectral, Damerdji1991-oj,  flegal2010batch, mcelroy2024estimating}.

In this paper, we study the spectral density and asymptotic variance estimation problems in the case that $\{X_{t}\}$ is a reversible Markov chain. As in~\citet{berg2023efficient}, our approach leverages the fact that the autocovariance sequence $\gamma$ from a reversible Markov chain can be expressed as a mixture of exponentials:
\begin{align}\gamma(k)=\int_{[-1,1]}\alpha^{|k|}\,\mu(d\alpha),\label{eq:mixture}
\end{align} $\forall k\in\mathbb{Z}$, with respect to some nonnegative measure $\mu$. The representation~\eqref{eq:mixture} imposes constraints on the true autocovariance sequence, similar to settings of estimating functions on discrete supports with various shape constraints \citep{durot2013least, Giguelay2017-sy,Balabdaoui2020-nn}. \citet{berg2023efficient} used the least-squares approach, which minimizes the squared $\ell_2$ distance between an input $r_M$, such as the empirical autocovariance, and the function $m$, i.e.,
\begin{equation}\label{eq: objective_LS}
\sum_{k=-\infty}^\infty (r_M(k)-m(k))^2
\end{equation}
subject to $m$ admitting a mixture representation as in \eqref{eq:mixture}.  While least-squares objectives are commonly used in shape-constrained regression (e.g., \citealp{Barlow1972-ej, kuosmanen2008representation, guntuboyina2018nonparametric, Balabdaoui2020-nn}), the objective \eqref{eq: objective_LS} essentially gives equal weights among the input data $r_M(k)$, and does not take the correlation between $r_M(k)$ into consideration, which can be a source of estimation inefficiency.

We propose a novel weighted least-squares approach to better account for unequal variances and correlations in the input data. Our main idea lies in defining a weighted norm based on distances in the Fourier-transformed space, where correlations between (transformed) samples tend to be reduced. We assign weights to these transformed samples based on their asymptotic variances. Since the minimization is performed in the Fourier domain, our method can be also viewed as estimating a function by applying shape constraints to its Fourier series, which to our knowledge has not been previously explored. 

We provide characterizations of the estimator and study the asymptotic properties of the resulting spectral density and asymptotic variance estimators. While our approach builds upon the shape-constrained estimation framework of \citet{berg2023efficient}, the transformation into the Fourier domain and the introduction of a weight function necessitate non-trivial modifications to their analysis. Notably, leveraging the variation diminishing property of totally positive kernels~\citep{karlin1968total}, we establish that while the objective is minimized over the space of infinite mixtures of scaled exponentials, the minimizer of the weighted objective has indeed a finite number of components. We also derive a tighter bound on the number of support points in the mixing measure, i.e., number of mixing components, than in~\citet{berg2023efficient}. We further establish the strong consistency of our proposed spectral density and asymptotic variance estimators under mild smoothness conditions on the estimated weight function in the Fourier domain, which are related to the rate of decay of its Fourier series in the time domain.

The rest of this paper is structured as follows. Section \ref{sec: 2_problem-setup} provides background on Markov chains and on the unweighted shape-constrained autocovariance sequence estimator from \citet{berg2023efficient}. Section \ref{sec: 3_weighted-momentLS} introduces our proposed estimator, the weighted moment least squares estimator, and provides some characterizations of the proposed estimator. Section \ref{sec: 4_statistical_analysis} presents statistical convergence results for the proposed estimators, including strong convergence of the autocovariance sequence, spectral density, and asymptotic variance estimates. Section \ref{sec: 5_empirical} demonstrates empirical improvements of the proposed method compared to the unweighted moment least squares estimator and other state-of-the-art estimators on some simulated and real data examples, including an AR(1) chain and a Bayesian LASSO Gibbs sampling chain~\citep{park2008bayesian, rajaratnam2019uncertainty}.

\section{Problem set-up and background}\label{sec: 2_problem-setup}
We first fix a few notations. For a measure space $(\mathsf{X}, \mathcal{X}, \pi)$ and measurable function $f:\mathsf{X} \to \C$, we define the $L^p$ norm $\|f\|_{L^p(\pi)} = \{\int |f(x)|^p \pi(dx)\}^{1/p}$, and define the $L^p$ space as 
$L^p(\mathsf{X}, \mathcal{X}, \pi):=\{f: \mathcal{X} \to \C; \|f\|_{L^p(\pi)}  <\infty\}$ for $p \ge 1$. We abbreviate $L^p(\mathsf{X}, \mathcal{X}, \pi)$ by $L^p(\pi)$ if there is no danger of confusion. When $p=2$, $L^2(\mathsf{X}, \mathcal{X}, \pi)$ is a Hilbert space with the inner product $[f,g]_\pi = \int f(x) \overline{g(x)} \pi(dx)$.

For a sequence $x: \Z \to \C$, we define $\|x\|_p = (\sum_{k\in\mathbb{Z}}|x(k)|^p)^{1/p}$, and we let $\ell_p(\Z,\C) = \{x:\Z\to\C; \|x\|_p <\infty\}$ for $p\ge 1$. We use $\Re:\C\to \R$ and $\Im:\C\to\R$ to denote functions which return the real part or imaginary part of a complex number. We let $\ell_p(\Z,\R)$ be the set of real-valued sequences, which we define as the subspace of $\ell_p(\Z,\C)$, such that $\ell_p(\Z, \mathbb{R}) = \{x\in \ell_p(\Z,\C);  \Im(x)=0\}$. When $p=2$, $\ell_2(\Z,\C)$ an inner product space, equipped with the inner product
$\braket{x,y}=\sum_{k\in\mathbb{Z}}x(k)\overline{y(k)} 
$ for functions $x,y:\Z \to \C$.

Throughout the paper, we consider a discrete time, time-homogeneous Markov chain $X=\{X_t\}_{t=0}^{\infty}$ on a general state space $\mathsf{X}$ with a $\sigma$-algebra $\mathcal{X}$. 
We let $Q$ be the Markov transition kernel so that $Q(x,A) = P(X_{t+1} \in A|X_t = x) $ for all $x \in \mathsf{X}$ , $A \in \mathcal{X}$, and $t =1,2,\dots$. An initial measure $\nu$ for $X_0$ and a transition kernel $Q$ define a Markov chain probability measure $P_\nu$ for $X = (X_0, X_1, X_2,\dots )$ on the canonical sequence space $(\Omega, \mathcal{F})$. We write $E_\nu$ to denote expectation with respect to $P_\nu$. We say the chain $X$ (or the kernel $Q$) has a stationary probability distribution $\pi$ if there exists a probability distribution such that $\int Q(x,A) \pi(dx) = \pi(dx)$. We say the chain $X$ (or kernel $Q$) is reversible with respect to $\pi$, if for any $A,B \in \mathcal{X}$, we have $\int_A  \pi(dx) Q(x,B) = \int_B \pi(dy) Q(y,A)$.  

For a function $f:\mathsf{X} \to \R$ and transition kernel $Q$, we define the associated linear operator $Q$ by $(Qf)(x) = \int f(y) Q(x,dy)$. We refer to this operator as a Markov operator. Also, we define $Q^k$ recursively for $k \ge 2$ by $Q^{k}f(x) = Q(Q^{k-1}f)(x)$. For an operator $T$ on $L^2(\pi)$, we define its spectrum as $\sigma(T)=\{\lambda \in \mathbb{C} ;(T-\lambda I)^{-1} \text { does not exist or is unbounded}\}$. For Markov operators $Q$, we define the spectral gap $\delta(Q)$ of $Q$ as $\delta(Q)=1-\sup \{|\lambda| ; \lambda \in \sigma(Q_0)\}$ where $I$ is the identity operator, $Q_0$ is defined as
$Q_0 f=Q f-E_\pi\left[f\left(X_0\right)\right] f_0$ and $f_0 \in L^2(\pi)$ is the constant function such that $f_0(x)=1$ for all $x \in \mathsf{X}$.  

Throughout the paper, we assume a Harris ergodic, $\pi$-reversible Markov chain $X = (X_0,X_1,\dots)$ with a positive spectral gap. More specifically, we assume the following:
\begin{assumption}[Assumptions on Markov chain]\label{assumption1}
\, 
\begin{enumerate}[label={1\alph*}, itemsep=0em]
\item \label{cond:harris_ergodicity} (Harris ergodicity) $X$ is $\psi$-irreducible, aperiodic, and positive Harris recurrent.
\item \label{cond:piReversible}
(Reversibility) $Q$ is $\pi$-reversible for a probability measure $\pi$ on $(\mathsf{X},\mathscr{X})$.
\item \label{cond:geometric_ergodicity}
(Geometric ergodicity) There exists a real number $\rho<1$ and a non-negative function $M$ on the state space $\mathsf{X}$ such that $\|Q^n(x,\cdot) - \pi(\cdot)\|_{\rm TV} \le M(x)\rho^n, \mbox{  for all } x\in \mathsf{X},$ where $\|\cdot\|_{\rm TV}$ is the total variation norm.
\end{enumerate}
\end{assumption}

\begin{remark}
For the definitions of $\psi$-irreducibility, aperiodicity, and positive Harris recurrence, see e.g., \citet{meyn2009markov}. Reversibility and geometric ergodicity are two key assumptions for the class of Markov chains we consider in this paper. Since $\pi$-reversibility is a sufficient condition for $\pi$-stationarity, reversible Markov chains have been widely used in the MCMC literature, as checking reversibility is often an easy way to verify that the constructed Markov chain converges to the correct target distribution.
Examples of reversible Markov chains include Metropolis~\citep{metropolis1953equation}, Metropolis-Hastings chains~\citep{hastings1970monte}, random-scan Gibbs samplers \citep{geman1984stochastic,liu1995covariance,greenwood1998}, and marginal chains from data augmentation Gibbs samplers \citep{wongKongLiu}.
The geometric ergodicity condition~\ref{cond:geometric_ergodicity} implies exponential convergence of the Markov chain $X$ to its target distribution $\pi$. When the state space $\mathsf{X}$ is finite, all irreducible and aperiodic Markov chains are geometrically ergodic, though the constant $\rho$ in~\ref{cond:geometric_ergodicity} may be impractically large. When $Q$ is reversible, $Q$ has a positive spectral gap if and only if the chain is geometrically ergodic~\citep{roberts1997geometric,kontoyiannis2012geometric}. As $|\gamma(k)| \le (1-\delta(Q))^k \gamma(0)$ by \eqref{eq:mixture} for any $k \in \Z$, the reversibility and geometric ergodicity conditions ensure that the covariances between the MCMC draws decay exponentially.
\end{remark}

We also assume that the function $g:\mathsf{X}\to\R$ is in $L^2(\pi)$, i.e,
\begin{assumption}[Square integrability]\label{cond:integrability}$\int g(x)^2\pi(dx)<\infty$.
\end{assumption}

It is well known that the autocovariance sequence $\gamma$ from a reversible Markov chain has the following spectral representation, due to the spectral theorem for self-adjoint operators (e.g., \citealp{Rudin1973-ry}).
\begin{prop}\label{prop:gamma_momentLSrep} Assume \ref{cond:piReversible} and \ref{cond:integrability}. 
The true autocovariance sequence $\gamma(k)= [ g, Q_0^k g ]_\pi, k \in \mathbb{Z}$, has the following representation for some positive measure $\mu_\gamma$
\begin{equation}\label{eq: mixture2}
\gamma(k)=\int_{\sigma\left(Q_0\right)} \alpha^{|k|} \mu_\gamma(d\alpha)    
\end{equation}
where $\sigma\left(Q_0\right)$ is the spectrum of the linear operator $Q_0$. Moreover, $\sigma\left(Q_0\right)$ lies on the real axis, and $\sigma\left(Q_0\right) \subseteq[-1,1]$.
\end{prop} 
\begin{remark}
Mixture representations for a sequence, and the corresponding support of the mixing measure, are closely linked to the \textit{shape} of the sequence. This connection has been leveraged in various shape-constrained regression problems, as it allows for the transformation of these problems into non-parametric mixture modeling problems. For instance, convex functions on $\N$ can be represented as mixtures of ``trifunctions", and completely monotone functions on $\N$ admit a representation as a scale mixture of exponentials with a mixing measure supported on $[0,1]$~\citep{durot2013least, Balabdaoui2020-nn}. 
\end{remark}

\citet{berg2023efficient} utilized this connection and proposed a momentLS estimator for $\gamma$ from Markov chain, in which they perform least-squares estimation of $\gamma$ subject to the mixture constraint \eqref{eq: mixture2}. 
For a given initial autocovariance sequence estimate $r_{M} \in \ell_2(\Z,\R)$, \citet{berg2023efficient} proposed the (unweighted) moment least squares estimator (moment LSE) as the projection of $r_{M}$ onto a $C$-\textit{moment} sequence space 
\begin{equation}\label{def: C_moment_space}
\mathscr{M}_{\infty}(C)=\{x\in \mathbb{R}^{\mathbb{Z}}:\,x(k)=\int\alpha^{|k|}\,\mu(d\alpha),\,\forall k\in\mathbb{Z},\,\text{for some $\mu \in \mathcal{M}_\R$ with $\operatorname{Supp}(\mu)\subseteq C$}\}    
\end{equation}
which is the space of functions $x$ on $\Z$ such that $x$ can be represented as an infinite scale mixture of exponentials for a measure $\mu$ supported on $C$.  Here, $\mathcal{M}_\R$ denotes the set of Radon measures, i.e., positive regular measures which are finite on all compact sets (see, e.g., \citealp{folland1999real}).

More concretely, for a given $C\subseteq [-1,1]$, the momentLS $\Pi(r_{M};C)$ is given by
\begin{align}\label{def: unweighted_momentLS}
\Pi(r_{M};C)=\underset{f\in\mathscr{M}_{\infty}(C)\cap \ell_2(\mathbb{Z},\R)}{\arg\min}\;\|r_{M}-f\|_2^2.
\end{align} 
When $C = [-1+\delta,1-\delta]$ for a nonnegative constant $0\leq \delta\leq 1$, for notational simplicity we let $\mathscr{M}_{\infty}(\delta) = \mathscr{M}_{\infty}([-1+\delta,1-\delta])$ and $\Pi(r_M;\delta) = \Pi(r_M;[-1+\delta,1-\delta])$.
We note that optimizing \eqref{def: unweighted_momentLS} over $f \in \mathscr{M}_\infty(C)$ is equivalent to optimizing over the corresponding mixing measures $\mu_f$ such that $f(k) = \int\alpha^{|k|} \mu_f(d\alpha)$, as there is a one-to-one correspondence between the mixing measure $\mu_f$ and the moment sequence $f$ (ref. Proposition 2 in \citealp{berg2023efficient}).

\citet{berg2023efficient} showed that the momentLS estimator $\Pi(r_M;\delta)$ is a strongly consistent estimator for $\gamma$ in the $\ell_2$ sense when the hyperparameter $\delta$ is chosen sufficiently small so that $\operatorname{Supp}(\mu_\gamma) \subseteq [-1+\delta,1-\delta]$, and the initial sequence estimator $r_M$ satisfies some mild conditions. 
Since we will utilize the same set of conditions in this paper, we summarize conditions on the initial autocovariance estimator $r_M$ as follows:
\begin{assumption} [conditions on initial autocovariance estimator $r_M$]
\, 
\begin{enumerate}[label={3\alph*}, itemsep=0em]
\item \label{cond:rM_pwConvergence} (a.s. elementwise convergence)
$r_M(k) \underset{M \rightarrow \infty}{\rightarrow} \gamma(k)$ for each $k \in \mathbb{Z}, P_x$-almost surely, for any initial condition $x \in \mathrm{X}$,    
\item \label{cond:rM_finiteSupport} (finite support)
    $r_M(k)=0$ for $k \geq n(M)$ for some $n(M)<\infty$, and
\item \label{cond:rM:even}(even function with a peak at 0)
    $r_M(k)=r_M(-k)$ and $r_M(0) \geq\left|r_M(k)\right|$ for each $k \in \mathbb{Z}$.
\end{enumerate}
\end{assumption}

\begin{remark}\label{remark:emp_autocov} 
We note that conditions \ref{cond:rM_pwConvergence}--\ref{cond:rM:even} are quite mild, particularly because we do not require the initial sequence estimator to converge to $\gamma$ in the $\ell_2$ sense. Notably, common estimators such as the empirical autocovariance sequence $\tilde{r}_M$ and windowed autocovariance sequences with increasing window sizes satisfy \ref{cond:rM_pwConvergence}--\ref{cond:rM:even} (ref. Proposition 7 in \citealp{berg2023efficient}).
\end{remark}

We introduce a few notations regarding moment sequences, i.e., sequences in $\mathscr{M}_\infty(C)$ for some $C \subseteq [-1,1]$. For $\alpha \in [-1,1]$,  we define $x_\alpha: \Z \to \R$ by $x_\alpha(k) = \alpha^{|k|}$ for each $k\in \Z$. We write $m=\int x_\alpha \,\mu(d\alpha)$ to refer to a moment sequence $m:\Z\to\R$ which satisfies
\begin{equation}\label{eq: moment_seq_rep}
    m(k) = \int x_\alpha(k)\, \mu(d\alpha) = \int \alpha^{|k|} \,\mu(d\alpha),\,\,\forall k\in\Z.
\end{equation}
If \eqref{eq: moment_seq_rep} holds, we say that $\mu$ is a \textit{representing measure} for $m$ since $\mu$ uniquely determines $m$.
Finally, for future reference, we summarize the previous main results for the (unweighted) moment LS in Theorem \ref{thm:unweighted}.
\begin{thm}[\citet{berg2023efficient}]\label{thm:unweighted}
Suppose $X$ is a Markov chain with transition kernel $Q$ satisfying~\ref{cond:harris_ergodicity}--\ref{cond:geometric_ergodicity}, and suppose $g:\mathsf{X}\to\mathbb{R}$ satisfies~\ref{cond:integrability}. Let $\gamma:\mathbb{Z}\to\mathbb{R}$ with $\gamma(k)=\operatorname{Cov}_{\pi}(g(X_0),g(X_k))$, $ k\in\mathbb{Z}$ denote the autocovariance sequence for $g$. Let $\delta>0$ such that $\operatorname{Supp}(\mu_{\gamma})\subseteq [-1+\delta,1-\delta]$, where $\mu_{\gamma}$ denotes the representing measure for $\gamma$. Suppose $r_M$ is an initial autocovariance sequence estimator satisfying~\ref{cond:rM_pwConvergence}--\ref{cond:rM:even}. Then for each initial condition $x\in\mathsf{X}$, 
\begin{enumerate}\setlength{\itemsep}{0em}
    \item $\|\Pi(r_M;\delta)-\gamma\|_2\underset{M\to\infty}{\to} 0$, $P_x$-a.s.
    \item $P_x(\hat{\mu}_{\delta M}\to \mu_{\gamma}\,\text{vaguely, as $M\to\infty$})=1$, where $\hat{\mu}_{\delta M}$ denotes the representing measure for $\Pi(r_M;\delta)$, and
    \item $\sigma^2(\Pi(r_M;\delta))\to\sigma^2(\gamma)$ $P_x$-a.s. as $M\to\infty$,
\end{enumerate} where we define $\sigma^2(m)=\sum_{k\in\mathbb{Z}}m(k)$ for a sequence $m$ on $\mathbb{Z}$.
\end{thm}

We extend this result in Theorem \ref{thm:weightedl1} in this paper to show that the convergence in 1. also holds with the $\|\cdot\|_2$ norm replaced with the $\|\cdot\|_1$ norm, and additionally that the spectral density estimates $\phi_{\delta M}$ based on $\Pi(r_M;\delta)$ converge uniformly to the true spectral density $\phi_{\gamma}$.

\section{Weighted moment least squares estimator}\label{sec: 3_weighted-momentLS}
\subsection{Motivation}
The least-squares objective $\sum_{k=-\infty}^\infty (r_M(k) - f(k))^2$ uses equal weights for each $r_M(k)$, and thus does not account for the unequal variances or dependencies between $r_M(k)$. However, there are non-zero covariances between $r_M(k)$ values, and the variances of $r_M(k)$ often differ significantly~\citep[e.g., Theorem 7.2.1][]{brockwell2009time}. This can lead to sub-optimal estimation of $\gamma$. 

To address this, we consider the transformation of an input sequence $r_M$ into the Fourier domain. It is well known that the Fourier transform of the periodogram ordinates $\tilde{r}_M(\omega_j),$ $j=0,1,...,\lfloor M/2\rfloor$, where $\omega_j=2\pi j/M, $ $j=0,1,...,M-1$ are the Fourier frequencies, behave approximately like independent exponential random variables with limiting variances equal to the squared true spectral density $\phi_\gamma$ at those Fourier frequencies~\citep[e.g.,][]{kokoszka2000periodogram,brockwell2009time,wu2011asymptotic}. 
On the other hand, the Fourier transform is an isometry. For any $f \in \ell_2(\Z,\C)$, we have 
\begin{align}\label{eq: isometry}
   \sum_{k=-\infty}^\infty (\tilde{r}_M(k) - f(k))^2=\frac{1}{2\pi}\int_{[-\pi,\pi]}\, |\widehat{\tilde{r}_{M}}(\omega)-\hat{f}(\omega)|^2\,d\omega
\end{align}
where $\widehat{\tilde{r}_{M}}$ and $\hat{f}$ refer to the Fourier transforms of $\tilde{r}_M$ and $f$, respectively (a precise definition of the Fourier transform for $\ell_2$ sequences will be provided shortly). These observations motivated us to work with the right-hand side of equation \eqref{eq: isometry} and propose a weighted norm by suitably weighting each $\widehat{\tilde{r}_{M}}(\omega)$ instead. Specifically, we propose to minimize a weighted objective
\begin{align}\label{eq: weighted_norm}
    \frac{1}{2\pi}\int_{[-\pi,\pi]}\, W(\omega)| \widehat{\tilde{r}_{M}}(\omega)-\hat{f}(\omega)|^2\,d\omega
\end{align}
with a weight function $W: [-\pi,\pi] \to \R^+$.  As for the weight function, we propose utilizing a suitable initial estimate of the spectral density so that each $\widehat{\tilde{r}_{M}}(\omega)$ is weighted based on the inverse of its limiting variance.

\subsection{Weighted inner product and weighted moment LSE}\label{subsec: 3_2_weighted_mLS}
In this subsection, we define a weighted inner product on $\ell_2(\mathbb{Z},\mathbb{C})$ such that the resulting induced norm corresponds to the weighted objective \eqref{eq: weighted_norm}.
To do so, we first introduce a few relevant definitions and notations in Fourier analysis. 
It is well known that there is a close connection between $2\pi$-periodic functions on $\R$ and sequences on $\Z$. 
Also, since all information on a $2\pi$-periodic function $f$ is contained in the values of $f$ on  $[-\pi,\pi]$, functions with period $2 \pi$ are functions on the quotient space $\T = \R\setminus 2\pi \Z = \{x+2\pi\Z; x\in \R\}$. Therefore, throughout the paper, we refer to $2\pi$-periodic functions and functions on $\T$ interchangeably. For functions $f$ on $\T$, we define the $L^p$ norm as $\|f\|_{L^p(\T)}=\left\{(2\pi)^{-1}\int_{[-\pi,\pi]} |f(x)|^p\,dx \right\}^{1/p}$, and denote the $L^p$ space as $L^p(\T)$. For $p=2$, we also define an inner product $(\cdot,\cdot)$ such that for $\hat{x},\hat{y} \in L^2(\T)$, $( \hat{x},\hat{y} ) = (2\pi)^{-1}\int_{[-\pi,\pi]}\hat{x}(\omega)\overline{\hat{y}(\omega)}d\omega.$ In addition, we use $C^k(\T)$ to denote the space of $2\pi$-periodic functions with $k$ continuous derivatives for $k\ge 0$. 

For $f\in \ell_2(\Z,\C)$, we define the Fourier transform $\hat{f}:\R \to \C$ of $f$ as
\begin{align}\label{def: Fourier_transform}
    \hat{f}(\omega)=\sum_{k\in\mathbb{Z}}e^{-ik\omega}f(k)
\end{align}    
where $\hat{f}(\omega) = \sum_{k\in\mathbb{Z}}e^{-ik\omega}f(k)$ denotes the $L^2(\T)$ limit of $S_N(\hat{f})(\omega)=\sum_{|k|\le N} e^{-i\omega k} f(k)$ as $N\to \infty$.
\begin{remark}
Since $e^{-ik(\omega+2\pi l)} = e^{-ik\omega}$ for any $l \in \Z$, $\hat{f}$ is a $2\pi$-periodic function on $\R$, i.e., $\hat{f} \in L^2(\T)$.
While the $L^2(\T)$-limit of $S_N(\hat{f})$ is well defined since $\{S_N(\hat{f})\}_N$ is a Cauchy sequence and $ L^2(\T)$ is a complete space, the pointwise limit of $S_N(\hat{f})(\omega)$ is not necessarily well-defined when $f$ is not absolutely summable. When $f\in \ell_1(\Z,\C)$, $S_N(\hat{f})$ converges for each $\omega$, and $\hat{f}$ is also the uniform limit of $S_N(\hat{f})$. See Section\ifnum\pageoption>1~\ref{supp_sec: 1_Fourier} \else ~S1 \fi in Supplementary Material for more details.     
\end{remark}

We define the subspace of $\ell_2(\Z,\mathbb{K})$ for $\mathbb{K}=\C$ or $\R$ consisting of even sequences, i.e., ${\ell}_2^{even}(\Z,\mathbb{K}) = \{x \in \ell_2(\Z,\mathbb{K}); x(k)= x(-k),\,\forall k \in \Z\}.$ While the Fourier transform of an $\ell_2$ sequence is complex-valued in general, the Fourier transform of an even, real-valued sequence is a real-valued, even function. For example, the spectral density function, which is the Fourier transform of the autocovariance sequence $\gamma$, is a real-valued, even function, as $\gamma(k) = \gamma(-k)$ for all $k \in \Z$.
\begin{lem}\label{lem: Fourier_even}
For $x \in {\ell}_2^{even}(\Z,\R)$, $\hat{x}$ is also a real-valued even function, i.e., $\Im(\hat{x}(\omega))=0$ and $\hat{x}(\omega) = \hat{x}(-\omega)$ for all $\omega$. 
\end{lem}
\begin{proof} 
Note that 
\begin{align*}
    \hat{x}(\omega) = \sum_{k=-\infty}^\infty x(k)e^{-ik\omega} = \sum_{k=-\infty}^\infty x(k)\cos(wk) - i\sum_{k=-\infty}^\infty x(k) \sin(wk)=\sum_{k=-\infty}^\infty x(k)\cos(wk) 
\end{align*}
where we have $\sum_{k=-\infty}^\infty x(k) \sin(\omega k) = \sum_{k=1}^{\infty} x(-k) \sin(-\omega k)+\sum_{k=1}^{\infty} x(k) \sin(\omega  k)  = 0$ and we use the fact $\sin(\omega k) = -\sin(\omega k)$ and $x(k) = x(-k)$ for all $k \in \N$. 
\end{proof}

For $x,y \in \ell_2(\Z,\C)$,  we have from Parseval's identity (ref.\ifnum \pageoption>1~\ref{eq:parseval-2}\else~S-2\fi~in Supplementary Material) that the inner product of $x$ and $y$ coincides with the inner product of the Fourier transforms of $x$ and $y$, i.e., 
\begin{align}
\langle x, y \rangle = \sum_{k = -\infty}^\infty x(k) \overline{y(k)}=(2\pi)^{-1}\int_{[-\pi,\pi]}\hat{x}(\omega)\overline{\hat{y}(\omega)}d\omega = (\hat{x},\hat{y}),\label{eq:innerProduct}
\end{align} 
where $\hat{x}$ and $\hat{y}$ are the Fourier transforms of $x$ and $y$.  
Now, for a $2\pi$-periodic and continuous $\phi:\R\to\mathbb{R}^+$ satisfying $0< c_0\le \phi(\omega) \le c_1 <\infty$ for all $\omega\in[-\pi,\pi]$ and $x,y \in \ell^2(\Z,\C)$, we define the \textit{$\phi$-weighted} inner product between $x$ and $y$ as
\begin{align*}
\braket{x,y}_{\phi}&=(\hat{x}/\phi,\hat{y}/\phi)=\frac{1}{2\pi}\int_{[-\pi,\pi]}\frac{\hat{x}(\omega)\overline{\hat{y}(\omega)}}{|\phi(\omega)|^2}d\omega
\end{align*} 
and denote the induced norm by $\|x\|_{\phi}=\braket{x,x}_{\phi}^{1/2}$.

For reference later, we summarize our requirements on a valid weight function $\phi$ as below:
\begin{assumption}[Assumption on weight function]\label{cond:phi}
$\phi:\R\to\mathbb{R}^+$ is an even, continuous, and $2\pi$-periodic function satisfying $0< c_0\le \phi(\omega) \le c_1 <\infty$ for all $\omega\in[-\pi,\pi]$
\end{assumption}

\begin{remark}[weighted inner product]\label{rmk: weighted_inner_product}
For a weight function $\phi$ satisfying Assumption~\ref{cond:phi} below, it is easy to check that $\braket{x,y}_{\phi}$ is a valid inner product on $\ell_2(\Z,\C)$. Also, $\|x\|_{\phi}$ and $\|x\|_2$ are equivalent norms on $\ell_2(\Z,\C)$ since $c_1^{-2}\|x\|_2 \le \|x\|_{\phi} \le c_0^{-2}\|x\|_2$.
\begin{enumerate}
 \setlength{\itemsep}{0pt}
    \item conjugate symmetry: $\langle x, y\rangle_\phi= (\hat{x}/\phi, \hat{y}/\phi) = \overline{(\hat{y}/\phi, \hat{x}/\phi)} = \overline{\langle y, x\rangle}_\phi$
    \item linearity in the first term: $\langle a x+b y, z\rangle_\phi =( a \hat{x}/\phi+b \hat{y}/\phi, \hat{z}/\phi)  = a( \hat{x}/\phi, \hat{z}/\phi)+b( \hat{y}/\phi, \hat{z}/\phi) = a\langle x, z\rangle_\phi+b\langle y, z\rangle_\phi$
    \item  non-negativity: $\langle x, x\rangle_\phi =( \hat{x}/\phi, \hat{x}/\phi) =  \frac{1}{2\pi} \int_{[-\pi,\pi]}\frac{|\hat{x}(\omega)|^2}{|\phi(\omega)|^2} d\omega \ge 0$.
    \item non-degeneracy:  $0 = \langle x, x\rangle_\phi =  \frac{1}{2\pi} \int_{[-\pi,\pi]}\frac{|\hat{x}(\omega)|^2}{|\phi(\omega)|^2} d\omega$ iff $x= 0$ as $|\phi(\omega)|^2 \ge c_0^2 >0 $ by assumption (ref. Lemma\ifnum \pageoption>1~\ref{lem:a_ahat}\else~S-3\fi~in Supplementary Material\ifnum \pageoption>1~\ref{supp_sec: 1_Fourier}\else~S1\fi). 
\end{enumerate}
\end{remark}
\begin{remark}[weighted inner product for even sequences]
As an easy consequence of Lemma \ref{lem: Fourier_even}, for two even sequences $x,y \in {\ell}_2^{even}(\Z,\R)$ and a $2\pi$-periodic $\phi:\R\to\mathbb{R}^+$, we have 
\begin{align}\label{eq:l2even_innerProduct_sym}
\langle x, y \rangle_\phi  = \langle y, x \rangle_\phi    
\end{align}
 since $ \langle x, y \rangle_\phi = (2\pi)^{-1} \int \hat{x}(\omega)\overline{\hat{y}(\omega)} / |\phi(\omega)|^2 d\omega = \int \overline{\hat{x}(\omega)}\hat{y}(\omega) / |\phi(\omega)|^2 d\omega = \langle y, x \rangle_\phi$
due to $\Im(\hat{x}(\omega)) = 0 $ and $\Im(\hat{y}(\omega)) = 0$.  
\end{remark}


In analogy to weighted least squares regression, for an input autocovariance sequence $r_M \in \ell_2(\Z,\C)$, weight function $\phi$ satisfying~\ref{cond:phi}, and a closed set $C\subseteq [-1,1]$, we define a weighted moment LSE as\begin{align}
\Pi^{\phi}(r_{M};C)=\underset{f\in \mathscr{M}_\infty(C) \cap \ell_2(\Z,\R)}{\arg\min}\|r_{M}-f\|_{\phi}^{2}~\label{eq: weightedMLSE}
\end{align} which corresponds to minimizing $\|r_{M}-f\|_{\phi}^{2}$
over $f$ in the square-summable $C$-\textit{moment} sequence space, where we recall that the $C$-moment sequence space is the set of moment sequences $f$ such that $f = \int x_\alpha\, \mu_f(d\alpha) $ for some $\mu_f$, as defined in \eqref{def: C_moment_space}.
Given the one-to-one correspondence between the moment sequence $f$ and representing measure $\mu_f$, we solve \eqref{eq: weightedMLSE} by optimizing over measures $\mu_f$. We provide a detailed discussion of the optimization procedure for \eqref{eq: weightedMLSE} in Section \ref{subsec: 3_4 computation}.

\begin{remark}[weighted least squares objective]
By definition, $\|r_{M}-f\|_{\phi}^{2} = \langle r_M- f, r_M- f\rangle_\phi = ((\hat{r}_M- \hat{f})/\phi, (\hat{r}_M- \hat{f})/\phi) = \frac{1}{2\pi}\int |\frac{\hat{r}_M(\omega)-\hat{f}(\omega)}{\phi(\omega)}|^2 d\omega$. Therefore, the objective \eqref{eq: weightedMLSE} essentially corresponds to estimating $\hat{f}$ under the constraint that the Fourier series $f$ of $\hat{f}$ obeys a certain shape constraint or mixture representation. Moreover, when the input sequence $r_M$ is even, all of $\hat{r}_M$, $\hat{f}$, and $\phi$ are real-valued, and the objective becomes
\begin{align}\label{def: weighted_obj}
\|r_{M}-f\|_{\phi}^2=\frac{1}{2\pi}\int_{[-\pi,\pi]}\left(\frac{\hat{r}_{M}(\omega)-\hat{f}(\omega)}{\phi(\omega)}\right)^2d\omega.
\end{align} 
This is convenient because the optimization can be performed entirely in the real plane. As both $r_M(k)$ and $r_M(-k)$ are estimates for the $k$th lag autocovariance, and $\gamma(k)=\gamma(-k)$, there seems to be no practical advantage to considering non-even input sequences. Therefore, for the remainder of paper, we will will focus on even input sequences, i.e., sequences $r_M$ such that $r_M(k) = r_M(-k)$ for all $k\in \Z$. Similar objectives to \eqref{def: weighted_obj} have previously been considered in the time series literature for estimating the spectral density of autoregressive (AR) or parametric models as in ~\citet{shibata1981optimal} and~\citet{chiu1988weighted}. Additionally,~\citet{stoica2000ma} considers estimation of moving average (MA) models based on weighted least squares fitting to the empirical autocovariance sequence.
\end{remark}

The weighted optimization problem \eqref{eq: weightedMLSE} requires specifying both $\phi$ and $C$.  For the weight function $\phi$, we propose using an estimate of $\phi(\omega)$ such that $\phi(\omega)^2$ approximates the asymptotic variance of $\hat{r}_M(\omega)$. For instance, when $r_M=\tilde{r}_M$ (empirical autocovariance), we use an estimate of the spectral density function $\phi_\gamma$, such as the one obtained from the unweighted momentLS estimator. 

The set $C$ is related to the size of the moment sequence space we work with, as for $C_1 \subseteq C_2$, we have $\mathscr{M}_{\infty}(C_1)\subseteq \mathscr{M}_{\infty}(C_2)$. Intuitively, using a smaller search space is preferable, as long as the true $\gamma$ is a feasible point in the search space. In this paper, we will mainly consider a closed interval in $(-1,1)$ for $C$, e.g., $C = [-1+\delta,1-\delta]$ for some $\delta>0$. 
For $C=[-1+\delta,1-\delta]$, the objective \eqref{eq: weightedMLSE} includes a hyperparameter $\delta$ that determines the endpoints of $C$. 
While the ``optimal" choice of $\delta$ depends on the representing measure of $\gamma$,  we do not require knowledge of the optimal $\delta$ in both our theoretical analysis and practical use of the proposed estimator. Our theoretical guarantees hold for any sufficiently small $\delta>0$ that ensures $\gamma$ is a feasible point of the objective \eqref{eq: weightedMLSE}. In practice, we adopt a data-driven method for tuning $\delta$, as proposed by \citet{berg2023efficient} (see Section \ref{sec: 5_empirical} for more details). 

\paragraph{Asymptotic variance and spectral density estimates from momentLS estimates}
For a (weighted or unweighted) moment LS estimator $\Pi(r_M) =\Pi^\phi(r_M;C)$ or $\Pi(r_M)=\Pi(r_M;C)$ for $\gamma$, we propose plug-in estimators for the asymptotic variance $\sigma^2 = \sum_{k \in \Z} \gamma(k)$ and spectral density $\phi_\gamma(\omega) = \sum_{k\in\Z} \gamma(k) e^{-i\omega k}$ based on $\Pi(r_M)$. 

Specifically, given an input sequence $r_M \in \ell_2(\Z,\R)$, we define the estimated spectral density $\phi_M$ based on the moment LS estimator $\Pi(r_M)$ as the Fourier transform of $\Pi(r_M)$, i.e.,
\begin{align}\label{def: specden_mLS}
    \phi_M(\omega) = \sum_{k\in\Z} \Pi(r_M)(k) e^{-i\omega k}
\end{align}
where we note that the Fourier transform of $\Pi(r_M)$ is well-defined since $\Pi(r_M) \in \ell_2(\Z,\R)$ by the definition of moment LS estimators for $r_M \in \ell_2(\Z,\R)$.
Furthermore, when $r_M \in \ell_1(\Z,\R)$, we define the asymptotic variance estimator $\sigma^2(\Pi(r_M))$ based on $\Pi(r_M)$ as
\begin{align}\label{def: asymp_var_mLS}
   \sigma^2(\Pi(r_M))= \sum_{k\in \Z}\Pi(r_{M}) (k) 
\end{align}
which is well-defined since it can be shown that $\Pi(r_M)$ is absolutely summable when $r_M \in \ell_1(\Z,\R)$ (ref. Lemma\ifnum \pageoption>1~\ref{lem: l1projection}\else~S-8\fi~in Supplementary Material Section\ifnum \pageoption>1~\ref{supp_sec: 2_lemmas}\else~S2\fi). We note that any finitely supported sequence $r_M$ is automatically in $\ell_1(\Z,\R)$. Therefore both the spectral density and asymptotic variance estimators are well-defined for momentLS estimators with $r_M$ satisfying \ref{cond:rM_pwConvergence}--\ref{cond:rM:even}.

We now comment on the practical computation of \eqref{def: specden_mLS} and \eqref{def: asymp_var_mLS}, which involve infinite sums. From performing the minimization of the objective \eqref{eq: weightedMLSE}, we obtain the representing measure of $\hat{\mu}$ of $\Pi(r_M)$. Disregarding technicalities related to exchanging integration order or handling infinite sums in Fourier transforms, we have
$\sum_{k\in\Z} \Pi(r_M) (k) e^{-i\omega k} = \sum_{k\in\Z}  \int \alpha^{|k|}\hat{\mu}(d\alpha)e^{-i\omega k} =   \int \sum_{k\in\Z} \alpha^{|k|}e^{-i\omega k} \hat{\mu}(d\alpha) =   \int K(\alpha,\omega) \hat{\mu}(d\alpha)$
where $K(\alpha,\omega)=\frac{1-\alpha^2}{1-2\alpha\cos\omega+\alpha^2}$ is the Fourier transform of $x_{\alpha}: k \to \alpha^{|k|}$ (ref. Lemma\ifnum \pageoption>1~\ref{lem: PoissonKernel}\else~S-5\fi~in Supplementary Material Section\ifnum \pageoption>1~\ref{supp_sec: 2_lemmas}\else~S2\fi). Similarly,
$
    \sum_{k\in \Z}\Pi(r_M)(k) = \sum_{k\in \Z} \int \alpha^{|k|}\,\hat{\mu}(d\alpha) = \int \sum_{k\in \Z} \alpha^{|k|} \hat{\mu}(d\alpha) =  \int \frac{1+\alpha}{1-\alpha} \,\hat{\mu}(d\alpha).
$
Thus, we can compute $\phi_M$ and $\sigma^2(\Pi(r_M))$ by 
\begin{equation}\label{def: momentLS_est2}
     \phi_M(\omega ) = \int K(\alpha,\omega) \hat{\mu}(d\alpha) \quad \mbox{ and } \quad \sigma^2(\Pi(r_M))  = \int \frac{1+\alpha}{1-\alpha} \hat{\mu}(d\alpha)
\end{equation}
A formal justification of \eqref{def: momentLS_est2} is provided in Lemma\ifnum \pageoption>1~\ref{lem: fourier_xa}\else~S-6\fi~in Supplementary Material\ifnum \pageoption>1~\ref{supp_sec: 2_lemmas}\else~S2\fi~and Lemma 5 from \citet{berg2023efficient}.

\subsection{Characterization of weighted Moment LSE}
In this subsection, we provide some characterizations of the weighted moment least squares (LS) estimator. Specifically, we demonstrate that the minimizer of \eqref{eq: weightedMLSE} is unique and is characterized by two necessary and sufficient optimality conditions presented in Proposition \ref{prop:weighted_opt_ineq}. Furthermore, while the minimization is performed over the function space of (infinite) scale mixtures of exponentials $\int \alpha^{|k|} \mu (d\alpha)$, we show that the mixing measure $\hat{\mu}_C$ of the minimizer $\Pi^{\phi}(r_{M};C)$ is finite, with a support size bounded by the number of non-zero points in the input sequence.

The finite support property of the estimated mixing measure, as shown in Proposition \ref{prop:finiteSupport}, is a common feature in non-parametric mixture modeling and has been established in various settings \citep{jewell1982mixtures,lindsay1983geometry, lindsay1993uniqueness, Balabdaoui2020-nn}. Our approach is related to that in \citet{lindsay1993uniqueness}, who studied the properties of non-parametric maximum likelihood estimators of mixing distributions by leveraging the properties of totally positive kernels. In particular, we show that the Poisson kernel, defined as $K(\alpha,\omega) = \frac{1-\alpha^2}{1-2\alpha \cos(\omega) + \alpha^2}$, is a strictly totally positive kernel on $(-1,1) \times [-\pi, 0]$. This allows us to control the number of support points in the mixing measure of the minimizer based on the number of sign changes in $\widehat{\Pi^{\phi}(r;C)}(\omega)-\hat{r}(\omega)$. We believe that the total positivity of the Poisson kernel (ref. Corollary \ref{cor:tp_K}) has not been previously shown in the literature and may be of independent interest. 

First of all, we present the following Proposition \ref{prop: existence_and_uniqueness} which addresses the uniqueness of the weighted moment LSE, as well as its representing measure. Proposition \ref{prop: existence_and_uniqueness} essentially follows from the Hilbert space projection theorem and the equivalence between the weighted norm $\|\cdot\|_\phi$ and the standard $\ell_2$ norm $\|\cdot\|_2$ under Assumption  \ref{cond:phi}. We defer the proof to Section\ifnum \pageoption>1~\ref{supp_sec: 3_1_unique}\else~S3.1\fi~in the Supplementary Material.
\begin{prop}\label{prop: existence_and_uniqueness}
Suppose $r \in {\ell}_2^{even}(\mathbb{Z},\R)$ and $\phi$ satisfies Assumption \ref{cond:phi}. Let $C$ be a closed subset of $[-1,1]$. The minimizer $\Pi^{\phi}(r;C)$ of \eqref{eq: weightedMLSE} exists and is unique. Moreover, the representing measure $\mu \in \mathcal{M}_\R$ for $\Pi^{\phi}(r;C)$ such that $\Pi^{\phi}(r;C) = \int x_\alpha \,\mu(d\alpha)$ is unique.
\end{prop}

We start with a characterization of the weighted moment LSE $\Pi^{\phi}(r;C)$, in Proposition \ref{prop:weighted_opt_ineq}. 
This characterization essentially follows from the fact that $\Pi^{\phi}(r;C)$ minimizes the weighted inner product norm. In Proposition~\ref{prop:weighted_opt_eq}, we show that $\Pi^{\phi}(r;C)$ and $r$ give the same inner product with $x_{\alpha}$, for $\alpha\in \text{Supp}(\mu_{C}^{\phi})$.
We omit the proofs of Propositions \ref{prop:weighted_opt_ineq} and \ref{prop:weighted_opt_eq}, as they follow similar approaches to Proposition 2.2 in~\citet{Balabdaoui2020-nn} and Proposition 5 in~\citet{berg2023efficient}, respectively.

\begin{prop}\label{prop:weighted_opt_ineq}
Suppose $r \in {\ell}_2^{even}(\mathbb{Z},\R)$ and $\phi$ satisfies Assumption \ref{cond:phi}. Let $C$ be a closed subset of $[-1,1]$. Then for $f \in \mathscr{M}_{\infty}(C) \cap \ell_2(\mathbb{Z})$, we have $f=\Pi^\phi(r ; C)$ if and only if
\begin{enumerate}
    \item for all $\alpha \in C \cap(-1,1),\left\langle f, x_\alpha\right\rangle_\phi \geq\left\langle r, x_\alpha\right\rangle_\phi$,
    \item $\langle f,f\rangle_\phi = \langle r,f\rangle_\phi$.
\end{enumerate}
\end{prop}

\begin{prop}\label{prop:weighted_opt_eq}
Suppose $r \in {\ell}_2^{even}(\mathbb{Z},\R)$ and $\phi$ satisfies Assumption \ref{cond:phi}. Let $C$ be a closed subset of $[-1,1]$. Let $\mu^\phi_{ C}$ denote the representing measure for $\Pi^\phi(r ; C)$. 
Then for each $\alpha \in \operatorname{Supp}(\mu^\phi_{ C}) \cap(-1,1)$, we have
$$
\left\langle\Pi^\phi (r ; C), x_\alpha\right\rangle_\phi=\left\langle r, x_\alpha\right\rangle_\phi.
$$
\end{prop}

In Proposition~\ref{prop:finiteSupport} below, we show that for an input sequence $r$ with $|r(k)|=0$ for $|k|>M-1$, the support of the representing measure $\mu_{C}^{\phi}$ for $\Pi^{\phi}(r;C)$ is discrete and finite, with the maximum possible number of support points depending on $M$. We use the variation diminishing property of totally positive kernels, which essentially states that if a kernel $K(x,\theta): \mathcal{X}\times \Theta \to \R$ is totally positive (i.e., the determinant of the matrix $(K(x_i, \theta_j))_{i,j=1}^n$ is positive for any choice of $n$, $x_1<\dots<x_n$, $x_i\in \mathcal{X},$ $i=1,...,n$ and $\theta_1<\dots<\theta_n$, $\theta_i\in\Theta$, $i=1,...,n$),  then the number of sign changes in the function $g(\theta) = \int K(x,\theta) f(\theta) d\theta$ is controlled by the number of sign changes in the function $f$ \citep{schoenberg1930variationsvermindernde, karlin1968total}.

We give a sketch of the proof approach here. Let $g:(-1,1)\to \mathbb{R}$ be defined by \begin{align}
g(\alpha)=\braket{\Pi^{\phi}(r;C)-r,x_{\alpha}}_{\phi}=(2\pi)^{-1}\int_{[-\pi,\pi]}K(\alpha,\omega)\frac{\widehat{\Pi^{\phi}(r;C)}(\omega)-\hat{r}(\omega)}{\phi(\omega)^2}\,d\omega,\label{eq:gradient}
\end{align} where $K(\alpha,\omega)=\frac{1-\alpha^2}{1-2\alpha\cos\omega+\alpha^2}$ denotes the Fourier transform of $x_{\alpha}$. From Proposition~\ref{prop:weighted_opt_eq}, the number of support points of $\mu_{C}^{\phi}$ in $(-1,1)$ is smaller than or equal to the number of zeroes of $g(\alpha)$ on $(-1,1)$. Direct analysis of~\eqref{eq:gradient} appears to be difficult, since we assume little about the weight function $\phi$ apart from positivity and continuity. Instead, we first show that $K(\alpha,\omega)$ is a strictly totally positive kernel, and then we use the resulting variation diminishing property of $K(\alpha,\omega)$ to control the number of zeroes of $g$ by the number of sign changes of the function $\frac{\widehat{\Pi^{\phi}(r;C)}(\omega)-\hat{r}(\omega)}{\phi(\omega)^2}$, which is shown to be a continuous function of $\omega$. 

The following Lemma \ref{lem:tp_K} and Corollary \ref{cor:tp_K} concern total positivity of the kernels $\tilde{K}: (-1,1) \times [-1,1] \to \R$ such that $\tilde{K}(\alpha,x) = \frac{1-\alpha^2}{1+\alpha^2 -2\alpha x}$ and $K:(-1,1)\times[-\pi,0] \to \R$ such that $K(\alpha,\omega) = \tilde{K}(\alpha,\cos(\omega))=\frac{1-\alpha^2}{1-2\alpha\cos\omega+\alpha^2}$. The proof of Lemma \ref{lem:tp_K} is based on the direct computation of the determinant of the matrix $\mathbf{M} = (\tilde{K}(\alpha_i,x_j))_{i,j=1}^n$, using the fact that $\mathbf{M}$ can be written as a transformation of a Cauchy matrix. Corollary \ref{cor:tp_K} is an easy consequence of Lemma \ref{lem:tp_K} as cosine is a strictly increasing function on $[-\pi,0]$. The proofs are provided in Section\ifnum \pageoption>1~\ref{supp_sec: 3_2_tp}\else~S3.2\fi~in Supplementary Material.
\begin{lem}\label{lem:tp_K}
Define $\tilde{K}(\alpha,x) = \frac{1-\alpha^2}{1+\alpha^2 - 2\alpha x}$. Fix $n\in \mathbb{N}$. For any  $-1<\alpha_1<\dots<\alpha_n<1$ and $-1\leq x_1<\dots<x_n \leq 1$, define $\mathbf{M}\in \R^{n\times n}$ such that $\mathbf{M}_{ij} = \tilde{K}(\alpha_i,x_j)$. We have $|\mathbf{M}| > 0$. That is, $\tilde{K}$ is a strictly totally positive kernel.    
\end{lem}

\begin{cor}\label{cor:tp_K}
Define $K(\alpha,\omega)=\frac{1-\alpha^2}{1-2\alpha\cos\omega+\alpha^2}$. Then $K$ is strictly totally positive on $(-1,1)\times[-\pi,0]$.
\end{cor} 
We finally provide the bounds in the number of support points in the representing measure $\mu_{C}^{\phi}$ for $\Pi^{\phi}(r;C)$ in Proposition \ref{prop:finiteSupport}. The proof can be found in Supplementary Material Section\ifnum \pageoption>1~\ref{supp_sec: 3_3_finiteSupport}\else~S3.3\fi.

\begin{prop}\label{prop:finiteSupport} Let $C$ be a closed subset of $[-1,1]$, and suppose $r\in{\ell}_2^{even}(\mathbb{Z},\R)$ satisfies $r(k)=0$ for all $k$ with $|k|>M-1$ for $1\leq M<\infty$. Suppose $\phi$ is a weight function satisfying \ref{cond:phi}. Let $\Pi^\phi (r;C)$ denote the projection of $r$ onto $\mathscr{M}_{\infty}(C)\cap \ell_2(\mathbb{Z})$, and let $\mu^\phi_{C}$ denote the representing measure for $\Pi^\phi(r;C)$. Then $\operatorname{Supp}(\mu^\phi_{C})$ contains at most $n$ points, where $n$ is the smallest even number such that $n>(M-1)$, and the support of $\mu^\phi_{C}$ is contained in $(-1,1)$, that is, $\operatorname{Supp}(\mu^\phi_{C})\cap \{-1,1\}=\emptyset$. Furthermore, in the case $C=[L,U]$ for some $-1\leq L\leq U\leq  1$, we have $|\operatorname{Supp}(\mu_{C}^{\phi})|\leq  \frac{n}{2}+1$. If in addition either one of $L=-1$ or $U=1$ holds, then we have $|\operatorname{Supp}(\mu_{C}^{\phi})|\leq \frac{n}{2}$.
\end{prop}

We remark that since the unweighted estimate is the special case of the weighted estimate with $\phi\equiv 1$, the bound in Proposition~\ref{prop:finiteSupport} in the case $C$ is an interval is stronger than the bound in~\citet{berg2023efficient} Proposition 6, essentially by a factor of 2.

\subsection{Computation}\label{subsec: 3_4 computation}

We now provide some details on computing the solution to the optimization problem in~\eqref{eq: weightedMLSE}. For a general set $C$, we use a grid approximation $\tilde{C}$ of $C$ containing a length $s$ grid of closely spaced points in $C$, and we approximate the estimator $\Pi^\phi(r_M;C)$ with $\Pi^\phi(r_M;\tilde{C})$. 
For instance, when $C = [-1+\delta, 1-\delta]$, we create a length $s$ grid $\tilde{C} = \{\alpha_1,\dots,\alpha_s\}$ from $\alpha_1 = -1+\delta$ to $\alpha_s = 1-\delta$, and minimize $\|r_M-f\|_{\phi}^2$ over $f$ such that $f=\int x_\alpha \,\mu_f(d\alpha)$ with $\operatorname{Supp}(\mu_f) \subseteq \{\alpha_1,\dots,\alpha_s\}$. 
We note since the representing measure $\hat{\mu}_C$ of $\Pi^\phi(r_M;C)$ is a discrete measure with finitely many support points (Proposition \ref{prop:finiteSupport}), the approximation $\Pi^\phi(r_M;\tilde{C})$ is exact when $\tilde{C}$ contains the support of $\hat{\mu}_C$. In practice, we observed that the estimated autocovariance function, spectral density, and asymptotic variance were robust to the choice of grid points.

Now we discuss how to solve $\Pi^\phi(r_M;\tilde{C})$ where $\tilde{C} = \{\alpha_1,\dots,\alpha_s\}$. First, note that for $f \in \mathscr{M}_\infty(\tilde{C})$, we have $f = \int x_\alpha \,\mu_f(d\alpha)$ for $\mu_f = \sum_{i=1}^s w_{\alpha_i} \delta_{\alpha_i}$, where $\delta_{\alpha_i}$ represents a point mass at $\alpha_i$ and $w_{\alpha_i} \in \R_{+}$ is a weight for $\alpha_i$. Then $f = \sum_{i=1}^s w_{\alpha_i}x_{\alpha_i} $, and therefore for $r_M \in {\ell}_2^{even}(\Z,\C)$,
\begin{align}
    \|r_M-f\|_{\phi}^2
    &=\|r_M\|_{\phi}^2-\braket{r_M,f}_{\phi}-\braket{f,r_M}_{\phi}+\|f\|_\phi^2\nonumber\\
    &=\|r_M\|_{\phi}^2-2\sum_{i=1}^s w_{\alpha_i}\braket{r_M,x_{\alpha_i}}_{\phi}+\sum_{i=1}^s \sum_{j=1}^s w_{\alpha_i}w_{\alpha_j}\braket{x_{\alpha_i}, x_{\alpha_j}}_\phi,\label{eq: obj_comp1}
\end{align}
where for the second equality we use Lemma \ref{lem: Fourier_even}. Write $\mathbf{w}=[w_{\alpha_1},...,w_{\alpha_s}]^\top\in\mathbb{R}^{s}$, $\mathbf{a}=[a_1,...,a_s]^\top\in\mathbb{R}^{s}$ such that $a_i=\braket{x_{\alpha_i},r_M}_{\phi}$, and $\mathbf{B}\in\mathbb{R}^{s\times s}$ such that $B_{ij}=\braket{x_{\alpha_i},x_{\alpha_j}}_{\phi}$, for $i,j=1,...,s$. Then we can write \eqref{eq: obj_comp1} as $\|r_M-m\|_{\phi}^2=\|r_M\|_{\phi}^2-2\mathbf{a}^\top \mathbf{w}+\mathbf{w}^\top \mathbf{B}\mathbf{w}$, and the optimization problem for obtaining $\Pi^{\phi}(r_M,\tilde{C})$ becomes \begin{align}
    &\underset{\mathbf{w}}{\min}\quad \|r_M\|_{\phi}^2-2\mathbf{a}^\top\mathbf{w}+\mathbf{w}^\top \mathbf{B}\mathbf{w}\nonumber\\
    &\text{subject to }\mathbf{w}\geq 0.\label{eq:quadprog}
\end{align} 
Given $\mathbf{a}$ and $\mathbf{B}$, the optimization problem \eqref{eq:quadprog} is identical to the quadratic programming formulation of a nonnegative least squares problem, and can be solved using standard software. 

It remains to discuss how to compute $\mathbf{a}$ and $\mathbf{B}$. The elements of $\mathbf{a}$ and $\mathbf{B}$ are given as 
\begin{align*}
a_i&=\braket{x_{\alpha_i},r_M}_{\phi}=\frac{1}{2\pi}\int_{[-\pi,\pi]}\frac{K(\alpha_i,\omega)\hat{r}_{M}(\omega)}{\phi(\omega)^2}\,d\omega,\\
B_{ij} &= \braket{x_{\alpha_i},x_{\alpha_j}}_\phi = \int\frac{K(\alpha_i,\omega)K(\alpha_j,\omega)}{\phi(\omega)^2}d\omega
\end{align*}
for $i,j=1,\dots,s$.

For $i=1,...,s$, we use $a_i\approx M_0^{-1}\sum_{j=0}^{M_0-1}\frac{K(\alpha_i,\omega_j)\hat{r}_{M}(\omega_j)}{\phi(\omega_j)^2}$, where $M_0$ is a large integer, and $\omega_j=\frac{2\pi j}{M_0}$, $j=0,...,M_0-1$ which are the Fourier frequencies for $M_0$. For an input autocovariance sequence $r_M$ satisfying~\ref{cond:rM_finiteSupport}, the values $\hat{r}_{M}(\omega_j)$ can be computed using a fast Fourier transform, possibly after zero-padding of $r_M$. In our numerical experiments, we have chosen $M_0 =2M$, in order to sufficiently capture variability in $\hat{r}_{M}(\omega)=\sum_{k=-(n(M)-1)}^{n(M)-1}r_M(k)\cos(k\omega)$. 

Similarly, for the $B_{ij}$ integrals, we use 
 $   B_{ij}=\braket{x_{\alpha_i},x_{\alpha_j}}_\phi\approx M_1^{-1}\sum_{k=0}^{M_1-1}\frac{K(\alpha_i,\omega_k)K(\alpha_j,\omega_k)}{\phi(\omega_k)^2}$, where now $\omega_k=\frac{2\pi k}{M_1}$, $k=0,...,M_1-1$, and $M_1$ is a large integer. In our experiments, we used $M_1=M_0$. The accuracy of these $a_i$ and $B_{ij}$ approximations will depend on the smoothness of the particular $K(\alpha,\cdot)$ and $\phi(\cdot)$ involved in the integrals, as well as the fineness of the Fourier frequency grid. We expect large values for $M_0$ and $M_1$ to lead to higher accuracy, but such choices also lead to a higher computational cost for evaluating the sums in the approximation. In our examples, the choice $M_0=M_1=2M$ appears to perform satisfactorily. On the other hand, further development and formal analysis of methods for evaluating the $a_i$ and $B_{ij}$ integrals seems to be an important and interesting problem. We leave these extensions for future work.

\paragraph{Asymptotic variance and spectral density estimates} Given $\Pi^{\phi}(r_M;\tilde{C})$, i.e., the solution $\hat{\mathbf{w}} = [\hat{w}_{\alpha_1},\dots,\hat{w}_{\alpha_s}]$ of the quadratic programming program \eqref{eq:quadprog}, the spectral density estimate $\phi_M$ and asymptotic variance estimate $\sigma^2(\Pi^{\phi}(r_M;\tilde{C}))$ in \eqref{def: momentLS_est2} are given by $\phi_M(\omega) =  \sum_{i=1}^s \hat{w}_{\alpha_i}K(\alpha_i, \omega) $ and $\sigma^2(\Pi^{\phi}(r_M;C) )= \sum_{i=1}^s\hat{w}_{\alpha_i}\frac{1+\alpha_i}{1-\alpha_i}$.

\section{Statistical analysis}\label{sec: 4_statistical_analysis}
In this section, we study the consistency properties of weighted moment LS estimators. Recall that our weighted objective \eqref{eq: weightedMLSE} is based on a weight function $\phi$, which can be any function which satisfies Assumption \ref{cond:phi}. Practically, for a given sample $(X_0,\dots,X_{M-1})$, we would use a spectral density function estimate as the weight function. To emphasize that the weight function can be random and depend on $M$, we write $\phi=\phi_M$. We analyze the asymptotic properties of $\Pi^{\phi_M}(r_M;C)$, i.e., the minimizer of the following objective  \begin{align*}
  \Pi^{\phi_M} (r_{M}; \delta)=\underset{m\in\mathscr{M}_\infty(\delta)}{\arg\min}\,\|r_M-m\|_{\phi_M}^2.
\end{align*} 

Our starting point for the analysis is the following inequality in Lemma~\ref{lem:weightedl2_ineq}, which bounds the $\|\cdot\|_{\phi}$ distance between the projection $\Pi^{\phi}(r;\delta)$ and an arbitrary vector $f\in\mathscr{M}_\infty(\delta)\cap\ell_2(\mathbb{Z},\mathbb{R})$. Lemma~\ref{lem:weightedl2_ineq} will be used later with $f=\gamma$ and $r=r_{M}$.

\begin{lem}\label{lem:weightedl2_ineq}
Let $r \in {\ell}_2^{even}(\mathbb{Z},\R)$ and a weight function $\phi$ which satisfies \ref{cond:phi} be given. Suppose $\delta \in [0,1]$. We let $\Pi^\phi(r;\delta)=\arg\min_{m \in \mathscr{M}_{\infty}(\delta)\cap \ell_2(\Z,\R)} \|r-m\|_\phi^2$. For any $f \in \mathscr{M}_{\infty}(\delta) \cap \ell_2(\mathbb{Z},\R)$, we have
\begin{align*}
0 \leq \|\Pi^\phi(r; \delta)-f\|^2_\phi \leq-\int\left\langle x_\alpha, r-f\right\rangle_\phi \mu_f(d \alpha)+\int\left\langle x_\alpha, r-f\right\rangle_\phi \hat{\mu}^\phi_\delta(d \alpha)  
\end{align*}
where $\mu_f$ and $\hat{\mu}^\phi_\delta$ are the representing measures for $f$ and 
$\Pi^\phi(r; \delta)$, respectively.
\end{lem}
\begin{proof} The proof follows similar lines as in Lemma 1 of~\citet{berg2023efficient} for the unweighted case, and is therefore omitted here.
\end{proof}

Taking $f=\gamma$ and $r=r_M$ in Lemma~\ref{lem:weightedl2_ineq}, we obtain bounds $\|\Pi^{\phi}(r_M;\delta)-\gamma\|_{\phi_M}^{2}$ in terms of integrals involving inner products $\braket{x_{\alpha},r_M-\gamma}_{\phi_M}$. We note that the almost sure pointwise convergence $r_M(k)\to\gamma(k)$ as $M\to\infty$ in~\ref{cond:rM_pwConvergence} does not imply $\|r_M-\gamma\|_{\phi_M}\to 0$, even in the unweighted norm case where $\|\cdot\|_{\phi_M}=\|\cdot\|_{2}$ (i.e., $\phi_M\equiv 1)$. Thus, although we have the bound $\|\Pi^{\phi}(r_M;\delta)-\gamma\|_{\phi_M}\leq \|r_M-\gamma\|_{\phi_M}$ for each $M$, it is not directly useful, and we pursue an alternative approach based on Lemma~\ref{lem:weightedl2_ineq}.
In particular, we show in Proposition~\ref{prop:prop_c} that the $\braket{x_{\alpha},r_{M}-\gamma}_{\phi_{M}}$ inner products are asymptotically well-behaved under suitable assumptions on the weight function $\phi_{M}$. 

We first state the assumptions for $\phi_M$ in Assumption~\ref{cond:phi_M_all}. In addition to assuming that the weight function $\phi_M$ is valid for each $M$, we need asymptotic bounds on $\phi_M$ to ensure that it remains bounded above $0$ and finite as $M\to \infty$. Additionally, we impose smoothness conditions on $\phi_M$ to guarantee that the inverse Fourier transform of $\phi_M$ decays sufficiently fast.
\begin{assumption}\label{cond:phi_M_all}
\,
    \begin{enumerate}[label={5\alph*}, itemsep=0em]
\item (condition on $\phi_M$ for each $M$)\label{cond:phi_M}
For each $M$, $\phi_M$ is a real-valued, even function in $C^2(\T)$, and there exist $0<c_{0M}, c_{1M}<\infty$ such that $0<c_{0M} \le \phi_M(\omega) \le c_{1M} <\infty$. Here, $c_{0M}$ and $c_{1M}$ can be possibly random.
\item (asymptotic bounds on $\phi_M$)\label{cond:phi_M_asymp}
There exist constants $0 < c_\phi, c_\phi'<\infty$ such that \,\,\,
$\liminf_{M\to\infty} \inf_{\omega \in [-\pi,\pi]} \phi_M(\omega) = c_\phi >0$ and 
    $\limsup_{M\to\infty} \max\{ \|\phi_M\|_\infty, \|\phi'_M\|_\infty,\|\phi''_M\|_\infty\} \le c'_\phi$, 
    $P_x$-almost surely for any initial condition $x \in \mathsf{X}$.
\end{enumerate}
\end{assumption}

One plausible choice for the weight function $\phi_M$ is the estimated spectral density function $\phi_{\delta M}$ from the unweighted moment least-squares estimator $\Pi(r_M;\delta)$. To handle the case that $\Pi(r_M;\delta)$ is a zero sequence, we introduce a modified weight function $\tilde{\phi}_{\delta M}$ defined by taking $\tilde{\phi}_{\delta M}(\omega)=\frac{\phi_{\delta M}(\omega)}{\Pi(r_M;\delta)(0)}$, $\forall \omega$ in the case $\Pi(r_M;\delta)(0)>0$, and defining $\tilde{\phi}_{\delta M}(\omega)=1$, $\forall \omega$ in the case $\Pi(r_M;\delta)(0)=0$. Note $\Pi(r_M;\delta)(0)=0$ implies that $\Pi(r_M;\delta)(k)=0$ for all $k$. We show in Lemma~\ref{lem:deterministicPhi_M} that the modified weights $\tilde{\phi}_{\delta M}$ satisfy~\ref{cond:phi_M_all}. The weighted LS estimate using the modified weight function $\tilde{\phi}_{\delta M}$ is identical to the weighted LS using the estimated spectral density function $\phi_{\delta M}$ when $\Pi(r_M;\delta)(0) >0$. On the other hand,  $\Pi(r_M;\delta)(0)=0$ happens only when the input sequence $r_M $ is a zero sequence (Lemma\ifnum \pageoption>1~\ref{lem:unweighted_spectr_bound}\else~S-12\fi~in Section\ifnum \pageoption>1~\ref{supp_sec: 4_statistical_analysis}\else~S4\fi~in Supplementary Material). In this case, we have $\Pi^{\tilde{\phi}_{\delta M}}(r_M;\delta)(k)=\Pi(r_M;\delta)(k)=0$, $\forall k\in\mathbb{Z}$.

\begin{lem}
    Suppose $0<\delta\leq 1$ and $r_M$ satisfies~\ref{cond:rM_finiteSupport} and~\ref{cond:rM:even}. Then the weight sequence $\tilde{\phi}_{\delta M}$ satisfies~\ref{cond:phi_M} and~\ref{cond:phi_M_asymp}. \label{lem:deterministicPhi_M}
\end{lem}

Under Assumption~\ref{cond:phi_M_all}, we show in Proposition~\ref{prop:prop_c} that $|\braket{x_{\alpha},r_M-\gamma}_{\phi_M}|\to 0$ almost surely, uniformly over compact sets contained in $(-1,1)$. This convergence result is useful in combination with Lemma~\ref{lem:weightedl2_ineq} for showing the convergence of $\|\Pi(r_M;\delta)-\gamma\|_{\phi_M}$ to 0 as $M\to\infty$. The proof requires non-trivial modification of Proposition 8 in \citet{berg2023efficient}, which showed Proposition~\ref{prop:prop_c} for the unweighted case where $\phi_{M}\equiv 1$. In the unweighted case we have $\braket{x_{\alpha},r_M-\gamma}=\sum_{k\in\mathbb{Z}}\alpha^{|k|}\{r_M(k)-\gamma(k)\}$, and the proof uses the rapid decay of $|x_{\alpha}(k)|=|\alpha|^{|k|}$ as $k\to\infty$. In the weighted case, we have \begin{align*}
    \braket{x_{\alpha},r_M-\gamma}_{\phi_{M}}&=\frac{1}{2\pi}\int_{[-\pi,\pi]}\frac{K(\alpha,\omega)}{\phi_{M}(\omega)^2}\{\hat{r}_{M}(\omega)-\hat{\gamma}(\omega)\}\,d\omega=\sum_{k\in\mathbb{Z}}g_{M\alpha}(k)\{r_M(k)-\gamma(k)\}
\end{align*} where we let $g_{M\alpha}$ denote the inverse Fourier transform of $\frac{K(\alpha,\omega)}{\phi_{M}(\omega)^2}$. The proof of Proposition~\ref{prop:prop_c} for general $\phi_{M}$ satisfying Assumption~\ref{cond:phi_M_all} uses the smoothness assumptions on $\phi_M$ from Assumption~\ref{cond:phi_M_all} to ensure $|g_{M\alpha}(k)|$ decay suitably quickly. Please refer to Supplementary Material Section\ifnum \pageoption>1~\ref{supp_sec: 4_1_prop_c}\else~S4.1\fi~for the complete proof.

\begin{prop}\label{prop:prop_c}
Let $r_M$ be an input autocovariance estimator satisfying \ref{cond:rM_pwConvergence}--\ref{cond:rM:even}.  Assume Assumption \ref{cond:harris_ergodicity}--\ref{cond:geometric_ergodicity} and \ref{cond:integrability}. Suppose the weight function $\phi_M$ satisfies \ref{cond:phi_M} and \ref{cond:phi_M_asymp}. Let $\mathcal{K}$ be a nonempty compact set with $\mathcal{K} \subseteq (-1,1)$. We have
    \begin{align}\label{prop_c:result}
       \sup_{\alpha \in \mathcal{K}}|\langle x_\alpha, r_M - \gamma\rangle_{\phi_M} | \to 0 
\end{align}
$P_x$-almost surely as $M\to\infty$, for any $x \in \mathsf{X}$.
\end{prop}

We now present our two main theorems, Theorems~\ref{thm:l2conv} and Theorem~\ref{thm:weightedl1}. Theorem~\ref{thm:l2conv} shows the strong $\ell_2$ consistency of the weighted momentLS estimator for autocovariance sequence estimation. We note that the unweighted momentLS estimator is a special case of the weighted momentLS estimator with $\phi_M\equiv 1$. Thus the convergence result for the unweighted momentLS estimator also follows from Theorem~\ref{thm:l2conv}. 
\begin{thm}\label{thm:l2conv}
Suppose a Markov chain $X=X_0, X_1, \dots$ satisfies \ref{cond:harris_ergodicity}--\ref{cond:geometric_ergodicity} and suppose $g:\mathsf{X}\to \R$ satisfies \ref{cond:integrability}.
Let $r_M$ be an input autocovariance estimator satisfying \ref{cond:rM_pwConvergence}--\ref{cond:rM:even}.
Let $\Pi^{\phi_M}(r_M;\delta)$ be the weighted momentLS estimator, where the weight function $\phi_M(\omega)$ satisfies \ref{cond:phi_M} and \ref{cond:phi_M_asymp} and $\delta >0$ is chosen so that $\operatorname{Supp}(\mu_\gamma) \subseteq [-1+\delta, 1-\delta]$ where $\mu_\gamma$ is the representing measure of $\gamma$. Then 
\begin{enumerate}
\setlength{\itemsep}{0em}
    \item (convergence in weighted norm) $\|\gamma- \Pi^{\phi_M}(r_{M};\delta) \|_{\phi_M}^2 \to 0$
    \item (convergence in un-weighted $\ell_2$ norm) $\|\gamma- \Pi^{\phi_M}(r_{M};\delta) \|_2^2 \to 0$
\end{enumerate}
$P_x$ almost surely, as $M\to\infty$ for any initial condition $x\in\mathsf{X}$.
\end{thm} 

\begin{cor}[vague convergence of $\hat{\mu}_{\delta}^{\phi_M}$] \label{cor:meas_conv}Assume the same conditions as in Theorem
\ref{thm:l2conv}. For each initial condition $x \in \mathrm{X}$, we have $P_x(\hat{\mu}_{\delta}^{\phi_M} \to \mu_\gamma$ vaguely, as $M \to\infty)=1$, where $\hat{\mu}_{\delta}^{\phi_M}$ and $\mu_\gamma$ are the representing measures for $\Pi^{\phi_M}(r_M;\delta)$ and $\gamma$, respectively.
\end{cor}
\begin{proof}
    This is the direct consequence of strong $\ell_2$ convergence of the weighted momentLS estimator $\Pi^{\phi_M}(r_M;\delta)$ in Theorem \ref{thm:l2conv} as well as Lemma 7 in \citet{berg2023efficient}.
\end{proof}

Finally, we present strong consistency results for the spectral density function $\phi_{\delta M}^W$ and asymptotic variance $\sigma^2(\Pi^{\phi_M}(r_M;\delta))$ estimators based on the weighted moment LS estimator $\Pi^{\phi_M}(r_M;\delta)$, where we let $\phi_{\delta M}^W (\omega) =\sum_{k\in\Z}\Pi^{\phi_M}(r_M;\delta)(k)e^{-i\omega k}$ and $\sigma^2(\Pi^{\phi_M}(r_M;\delta)) = \sum_{k\in\Z} \Pi^{\phi_M}(r_M;\delta)(k)$, as defined in \eqref{def: specden_mLS} and \eqref{def: asymp_var_mLS}, respectively. Similarly, the strong convergence of the estimated spectral density function and asymptotic variance based on the unweighted moment LS estimator follow as a special case of Theorem \ref{thm:weightedl1} with $\phi_M\equiv 1$.

\begin{thm}[strong convergence of weighted spectral density and asymptotic variance momentLS estimators]\label{thm:weightedl1}
Assume the same conditions as in Theorem
\ref{thm:l2conv}. Let $\sigma^2(\gamma) = \sum_{k\in\Z} \gamma(k)$ be the true asymptotic variance.  Then for each $x\in\mathrm{X}$, we have
\begin{enumerate}
	\item ($\ell_1$-convergence of autocovariance sequence estimate) $\|\Pi^{\phi_{M}}(r_M; \delta) - \gamma\|_1 \to 0$,
	\item (uniform strong convergence of spectral density estimate) $\underset{\omega\in[-\pi,\pi]}{\sup}	|\phi_{\delta M}^W(\omega) - \phi_\gamma(\omega)| \to 0$,
    \item (strong convergence of asymptotic variance estimate) $\sigma^2(\Pi^{\phi_{M}}(r_M;\delta))\to \sigma^2(\gamma)$,
\end{enumerate} $P_x$-almost surely, as $M\to\infty$.
\end{thm}

\section{Empirical studies}\label{sec: 5_empirical}
In this section, we conduct numerical experiments to empirically evaluate the performance of the proposed weighted moment LS estimator in estimating the spectral density function and asymptotic variance. We consider two simulated settings: 1. an AR(1) chain where the true spectral density and asymptotic variance can be obtained analytically, and 2. Markov chains generated from a Bayesian LASSO regression model based on real-world data. In Section~\ref{sec:settings}, we describe our AR(1) and Bayesian LASSO simulation settings. In Section~\ref{subsec:ests}, we describe the estimators being compared. In Section~\ref{subsec:results}, we describe our comparison metrics and the simulation results.

\subsection{Settings\label{sec:settings}}

\paragraph*{Autoregressive chain} We consider the AR(1) autoregressive chain with $X_{t+1}=\rho X_t+\epsilon_{t+1}, \,\, t=0,1,2,\dots$
, where $\epsilon_t\overset{iid}{\sim}N(0,\tau^2)$ and $\rho\in (-1,1)$. The stationary measure $\pi$ for the $X_t$ chain is the measure corresponding to a $N(0,\tau^2/(1-\rho^2))$ random variable, and the $X_t$ chain can be shown to be reversible with respect to $\pi$. We consider the identity function $g(x)=x$.  It can be shown that for each $k\in\Z$, $\gamma(k) = \operatorname{Cov}_\pi(X_0, X_k) = \frac{\tau^2}{(1-\rho^2)}\rho^{|k|} $, and therefore $\gamma(k)$ can be represented as $\gamma(k) =\int x^{|k|}F(dx)$ for all $k \in \Z$ by letting $F=\frac{\tau^2}{1-\rho^2}\delta_{\rho}$, where $\delta_\rho$ denotes a unit point mass measure at $\rho$. The true spectral density is $\phi_{\gamma}(\omega)=\frac{\tau^2}{1-\rho^2}K(\rho,\omega)=\frac{\tau^2}{1-2\rho \cos(\omega)+\rho^2}$, and the true asymptotic variance is $\sigma^2(\gamma)=\frac{\tau^2}{(1-\rho)^2}$.

\paragraph*{Bayesian LASSO} 

We consider a linear model of the form 
\begin{align}
    Y=\mu\mathbf{1}_n+X\beta+\epsilon\label{eq:bayesianLM}
\end{align} where $Y$ is an $n\times 1$ vector of responses, $\mu\in\mathbb{R}$ is an unknown intercept, $\beta\in\mathbb{R}^{p}$ is a vector of unknown coefficients, $\mathbf{1}_{n}$ is the length $n$ vector of 1's, $X\in\mathbb{R}^{n\times p}$ is a matrix of standardized covariates, and $\epsilon\sim N(0,\sigma^2\mathbf{I}_{n})$ is a vector of iid Gaussian errors. The LASSO estimator of~\citet{tibshirani1996regression} solves the problem \begin{align}\label{eq: LASSO}
    \underset{\beta\in\mathbb{R}^p}{\min}\, (\tilde{Y}-X\beta)^\top(\tilde{Y}-X\beta)+\lambda\sum_{j=1}^{p}|\beta_j|
\end{align} where $\lambda>0$ is a regularization parameter, and $\tilde{Y}=Y-\bar{y}\mathbf{1}_n$, where $\bar{y}=\sum_{i=1}^{n}Y_i/n$. The LASSO estimator can be viewed as the posterior mode of a Bayesian model with independent Laplace priors on $\beta_i$. 
\citet{park2008bayesian} propose a Bayesian analysis for \eqref{eq: LASSO}, called Bayesian LASSO, using a conditional Laplace prior for $\beta|\sigma^2$ and an improper prior for $\sigma^2$. By representing the Laplace distribution as a scale mixture of Normals with an exponential mixing density, they propose an augmented prior specification with additional random variables $\tau=(\tau_1,...,\tau_p)$.  
Specifically, they assume an improper prior for $\sigma^2$, i.e., $p(\sigma^2) = 1/\sigma^2$, conditional iid exponential random variable $\tau_i$ with rate parameter $\lambda^2/2$, i.e., $\tau_i|\sigma^2\overset{iid}{\sim}\text{Exponential}(\lambda^2/2)$  for $i=1,\dots,p$, and conditionally independent normal random variables $\beta=(\beta_1,\dots,\beta_p)$ with mean zero and variance $\sigma^2 D_\tau$, i.e., $\beta|\tau,\sigma^2\sim N(\mathbf{0}_p,\sigma^2D_{\tau})$,  where $D_{\tau}=\text{diag}(\tau_1^2,...,\tau_p^2)$.

\citet{park2008bayesian} proposed a three-block Gibbs sampler for sampling from the posterior $(\beta,\sigma^2,\tau|Y)$, which cycles through updates of $\beta$, $\sigma^2$, and $\tau$. Later, \citet{rajaratnam2019uncertainty} developed a two-block Gibbs sampler for $(\beta,\sigma^2,\tau|Y)$ which alternates updates between $(\beta,\sigma^2)$ and $\tau$. They showed that the two-block Gibbs sampler has theoretical advantages over the original three-block sampler. They also showed the transition kernel for $(\beta,\sigma^2)$ marginal chain from the two-block sampler is positive and trace-class, and hence geometrically ergodic. We note that the $(\beta,\sigma^2)$ chain from the two-block sampler is reversible~\citep{rajaratnam2019uncertainty}. In our experiment, we will use the two-block Gibbs sampler by \citet{rajaratnam2019uncertainty}, with the updating scheme summarized in Algorithm~\ref{alg:BLsampler}.

\begin{algorithm}
\caption{\cite{rajaratnam2019uncertainty} two-block Bayesian LASSO sampler}\label{alg:BLsampler}
\begin{algorithmic}
\State $(\beta,\sigma^2)$ update: draw from $\beta,\sigma^2|\tau,Y$
\State \quad\quad $\sigma^2|\tau,Y\sim \text{InverseGamma}\left\{(n-1)/2,\;\tilde{Y}^\top(\mathbf{I}_{n}-XA_{\tau}^{-1}X^\top)\tilde{Y}/2\right\}$
\State \quad\quad $\beta|\sigma^2,\tau,Y\sim N(A_{\tau}^{-1}X^\top \tilde{Y},\;\sigma^2A_{\tau}^{-1})$ where $A_{\tau}=X^{\top}X+D_{\tau}^{-1}$\\
\State $\tau$ update: draw from $\tau|\beta,\sigma^2,Y$
\State \quad\quad $(1/\tau_i^2)|\beta,\sigma^2,Y\overset{ind.}{\sim} \text{InverseGaussian}(\mu'=\sqrt{\frac{\lambda^2\sigma^2}{\beta_j^2}},\;\lambda'=\lambda^2)$, \\
where the InverseGaussian density is given by $f_{\mu',\lambda'}(x)=\sqrt{\frac{\lambda'}{2\pi}}x^{-3/2}\exp\{-\frac{\lambda'(x-\mu')^2}{2(\mu')^2x}\}$.
\end{algorithmic}
\end{algorithm}

We consider a linear model of the form~\eqref{eq:bayesianLM} for predicting concrete compressive strength, using the compressive strength data in~\citet{yeh1998modeling} from~\citet{misc_concrete_compressive_strength_165}. The dataset contains $n=1030$ compressive strength measurements together with measurements of $p=8$ covariates, so that the design matrix $X\in\mathbb{R}^{n\times p}$ in~\eqref{eq:bayesianLM}. Thus, $p+1=9$ variables are sampled in the $(\beta,\sigma^2)$ block of the Gibbs sampler. We use the sampler implementation from the supplement S4 of~\citet{rajaratnam2019uncertainty}. We choose $\lambda=50$ for the amount of penalty. The posterior mean estimate for $\sigma^2$ with $\lambda=50$, based on the average of $100000$ chains, each with length $M=200000$, was $127.9$. For comparison, the residual variance estimate from fitting the full, frequentist unpenalized model~\eqref{eq:bayesianLM} is $108.16$, compared to a residual variance estimate  of $279.08$ from the intercept-only frequentist model. Thus, heuristically speaking, the posterior draws of $\beta$ seem to fit the data reasonably well. In particular, there does not seem to be an overwhelming amount of shrinkage due to the choice of $\lambda$.

\subsection{Descriptions of estimators\label{subsec:ests}}
We investigated the following estimators: 
\begin{enumerate}
\setlength\itemsep{0em}
    \item {\bfseries Weighted and unweighted MomentLS (mLS-w and mLS-uw)} our moment least squares estimators $\Pi^{\phi_{M}}(r_M;\delta)$ and $\Pi(r_M;\delta)$, with the empirical autocovariance sequence $\tilde{r}_{M}$ as the input sequence. For the weighted estimator, we take $\phi_M$ to be the spectral density estimate $\phi_{\delta M}$ resulting from the unweighted moment LS estimate. For the choice of $\delta$, we utilized the tuning method proposed in \citet{berg2023efficient} to select $\delta$, defined as
    $\tilde{\delta}_{M} = 0.8 \frac{1}{L}\sum_{l=1}^L \hat{\delta}_{M/L}^{(l)}.$ Here, $\hat{\delta}_{M/L}^{(l)}$ is an estimator for $\delta_{\gamma}$ computed from the $l$th batch of empirical autocovariances. We used $L=5$. 
\item {\bfseries Bartlett kernel (Bart)} the windowed empirical autocovariance sequence $\check{r}_M(k) = w_M(k) \tilde{r}_M(k)$ with a triangle lag-window $w_M(k) = (1-|k|/b_{M}) I(|k|<b_{M})$ with threshold $b_{M}$ chosen using the truncation point (batch size) tuning method implemented in the R package {\bfseries mcmcse}~\citep{liu2021batch}. 
    \item \textbf{Infinite order kernel (IO)} the windowed autocovariance sequence $\bar{r}_{M}(k)=w_M(k)\tilde{r}_{M}(k)$ with trapezoidal lag-window proposed by ~\citet{mcmurry2004nonparametric}, where $w_M(k) = I( \frac{|k|}{b_{M}}\leq 0.5) + \frac{2(b_{M}-|k|)}{b_{M}} I(0.5<\frac{|k|}{b_M}\leq 1)$, and $b_M$ is selected selected using the empirical rule of \citet{politis2003adaptive}. We use the R package {\bfseries iosmooth}~\citep{mcmurry2017package}.
    \item {\bfseries Initial convex sequence (Init-Conv)} The initial convex sequence estimator of~\citet{geyer1992practical}, using the implementation in the R package {\bfseries mcmc}~\citep{geyerMCMCpackage}
    
    \item {\bfseries Overlapping batch means (OBM)} the overlapping batch means~\citep{meketon1984overlapping} method, with batch size tuning method implemented in the R package~\citep{liu2021batch}.

\end{enumerate}
We note that the momentLS, Bartlett kernel, and infinite order kernel estimators estimate the autocovariance sequence $\gamma$. Therefore, both spectral density and asymptotic variance estimates are obtained based on the Fourier transform and sum of the estimated autocovariance sequence, respectively (in the case of MomentLS, we perform infinite sum, see Section \ref{subsec: 3_2_weighted_mLS} for details). In contrast, initial convex sequence and overlapping batch mean estimators directly estimate the asymptotic variance, and thus are used only for the comparison for asymptotic variance estimation.

\subsection{Results\label{subsec:results}}
\paragraph{Ground truth} For the Bayesian Lasso example, the true spectral density and asymptotic variance are unknown. To estimate these quantities, we independently sampled $B=100000$ chains, each of length $M=200000$ with $10000$ warm-up iterations, using the two-block Bayesian Lasso from \cite{rajaratnam2019uncertainty}. For each chain, we used the infinite order kernel with the empirical rule to estimate the spectral density at $\omega_j = 2\pi j/M_0$ for $j=0,\dots,M_0-1$, with $M_0=80000$. These estimates were then averaged across $B=100000$ chains to obtain the final ``ground truth" estimate of the spectral density function. The estimated spectral density at $0$ was used for the ``ground truth" for the asymptotic variance.

\paragraph{Comparison metrics}
For the Bayesian LASSO and AR(1) examples, we compare estimators based on mean squared error for asymptotic variance estimation, and integrated squared error for spectral density estimation. We describe these metrics in detail below. We consider chain lengths $M\in\{10000,40000\}$, and we simulate $B=500$ chains at each sample size.
For each asymptotic variance estimator $est$, sample size $M$, and parameter $par$ under consideration, we compute the average, over $B=500$ chains, of the squared error for asymptotic variance estimation $\{\widehat{\sigma^2_{est,par,M}}-\sigma^2_{ref,par}\}^2$, where $\sigma^2_{ref,par}$ denotes a reference asymptotic variance estimate for the parameter $par$. Similarly, for spectral density estimation, we compute an average integrated squared error (ISE) over $B=500$ chains, where $\text{ISE}=\frac{1}{2\pi}\int_{[-\pi,\pi]}\{\widehat{\phi_{est,par,M}}(\omega)-\phi_{par,ref}(\omega)\}^2\,d\omega$, for each considered spectral density estimator, parameter, and sample size. For computing the required integral, we use the approximation $\text{ISE}\approx M_0^{-1}\sum_{j=0}^{M_0-1}\{\widehat{\phi_{est,par,M}}(\omega_j)-\phi_{par,ref}(\omega_j)\}^2$, where $\omega_j=\frac{2\pi j}{M_0}$, $j=0,...,M_0-1$, and we use $M_0=80000$. 

For the Bayesian LASSO example, we compute (univariate) asymptotic variance and spectral density estimates for each of the model parameters $\sigma^2$ and $\beta_i,$ $i=1,...,p$.

\begin{table}[H]
\centering
\setlength\tabcolsep{3pt}
\caption{\textit{Average squared error (s.e.) for the asymptotic variance estimators and average integrated squared error (s.e.) for the spectral density estimators from AR1 examples with $\rho=-0.9$ and $0.9$ when $M=10000$. AR1(MLE) corresponds to the results from the parametric AR1 model. The method with the smallest average error (except AR1(MLE)) is highlighted in bold for each chain.}}
\label{tab:arTab}

\begin{subtable}[t]{\textwidth}
\small
\centering
\begin{threeparttable}
\caption{Asymptotic variance mean squared error\tnote{1}\label{tab:avarErrorAR}}
\begin{tabular}{l|c|ccccc}
  \hline
& AR1(MLE) & OBM & Init-Conv & IO & mLS-uw & mLS-w \\ 
  \hline
AR1($-0.9$) & 0.19 (0.01) & 13.61 (1.01) & 1346.68 (57.61) & 13.29 (0.89) & 16.11 (0.90) & \textbf{2.95 (0.31) }\\ 
  AR1($0.9$) & 66.43 (4.52) & 236.81 (11.12) & 178.93 (18.85) & 217.07 (24.48) & 165.16 (13.45) & \textbf{134.75 (10.42) }\\ 
   \hline
\end{tabular}
 \begin{tablenotes}
\footnotesize
\item[1] values for AR1($\rho=-.9$) are scaled by $10^4$
    \end{tablenotes}
\end{threeparttable}
\vspace{1em}

\begin{threeparttable}
\caption{spectral density mean integrated squared error\label{tab:iseAR}}
\begin{tabular}{l|c|cccc}
  \hline
 & AR1 (MLE) & Bart & IO & mLS-uw & mLS-w \\ 
  \hline
AR1($-0.9$) & 1.21 (0.07) & 3.53 (0.13) & 2.90 (0.13) & 2.10 (0.11) & \textbf{1.84 (0.13) }\\ 
AR1($0.9$)  & 1.09 (0.07) & 3.39 (0.11) & 3.05 (0.15) & 2.14 (0.12) & \textbf{1.89 (0.15) }\\ 
   \hline
\end{tabular}
\end{threeparttable}
\end{subtable}

\end{table}

\begin{table}[H]
\centering
\setlength\tabcolsep{3pt}
\caption{\textit{Estimated average mean squared error (s.e.) for the asymptotic variance estimators and mean integrated squared error (s.e.) for the spectral density estimators from Bayesian Lasso example when $M=10000$. The method with the smallest average error is highlighted in bold for each coefficient.}}

\label{tab:BLtab}

\begin{subtable}[t]{\textwidth}
\small
\centering

\begin{threeparttable}
\caption{Asymptotic variance mean squared error\tnote{1}\label{tab:avarErrorBL}}
\begin{tabular}{l|cccccc}
  \hline
 & Bartlett & OBM & Init-con & IO-kernel & mLS\_uw & mLS\_w \\ 
  \hline
$\beta_0$ & 333.45 (14.90) & 343.99 (15.19) & 162.57 (10.41) & 164.04 (9.82) & 163.10 (10.78) & \textbf{153.18 (10.37)} \\ 
  $\beta_1$ & 291.61 (13.49) & 300.97 (13.75) & 166.04 (11.06) & 160.04 (9.99) & 164.35 (11.04) & \textbf{148.67 (10.02)} \\ 
  $\beta_2$ & 368.42 (18.31) & 383.17 (18.83) & 227.44 (15.42) & 230.95 (14.68) & 229.72 (16.54) & \textbf{214.63 (15.28)} \\ 
  $\beta_3$ & 98.26 (4.89) & 101.43 (5.01) & 55.91 (3.46) & 58.10 (3.32) & 52.96 (3.27) & \textbf{49.24 (3.14)} \\ 
  $\beta_4$ & 63.69 (3.57) & 65.62 (3.64) & 41.21 (2.66) & 37.00 (2.32) & 38.75 (2.59) & \textbf{35.53 (2.39)} \\ 
  $\beta_5$ & 19.65 (1.00) & 20.10 (1.01) & 11.89 (0.84) & 12.27 (0.81) & 12.06 (0.89) & \textbf{11.62 (0.84)} \\ 
  $\beta_6$ & 141.89 (6.94) & 147.49 (7.17) & 87.85 (5.50) & 82.63 (4.68) & 86.93 (5.54) & \textbf{81.21 (5.27) }\\ 
  $\beta_7$ & 0.54 (0.03) & 0.55 (0.03) & 0.71 (0.05) & 0.53 (0.03) & 0.49 (0.03) & \textbf{0.45 (0.03)} \\ 
$\sigma^2$ & 351.56 (18.28) & 357.09 (18.47) & 474.50 (32.61) & 490.00 (17.44) & 337.27 (23.94) & \textbf{311.73 (21.89)} \\ 
   \hline
\end{tabular}
 \begin{tablenotes}
\footnotesize
\item[1] values for $\beta_0$-$\beta_7$ are scaled by $10^4$, and $\sigma^2$ is scaled by $10^2$
    \end{tablenotes}
\end{threeparttable}
\vspace{1em}

\begin{threeparttable}
\caption{spectral density mean integrated squared error\tnote{2}\label{tab:iseBL}}
\begin{tabular}{c|cccc}
  \hline
 & Bartlett & IO-kernel & mLS\_uw & mLS\_w \\ 
  \hline
$\beta_0$ & 177.58 (6.11) & 101.31 (3.95) & 80.85 (3.64) &\textbf{ 74.76 (3.45)} \\ 
  $\beta_1$ & 153.17 (5.42) & 92.70 (3.70) & 75.28 (3.49) & \textbf{68.34 (3.29)} \\ 
  $\beta_2$ & 198.89 (6.98) & 128.41 (5.34) & 103.72 (5.01) & \textbf{94.98 (4.75)} \\ 
  $\beta_3$ & 76.46 (2.53) & 50.57 (1.85) & 37.32 (1.66) & \textbf{34.89 (1.63)} \\ 
  $\beta_4$ & 61.09 (1.97) & 37.88 (1.46) & 29.41 (1.34) &\textbf{ 27.04 (1.27)} \\ 
  $\beta_5$ & 12.52 (0.50) & 9.07 (0.40) & 7.37 (0.38) & \textbf{7.10 (0.37)} \\ 
  $\beta_6$ & 73.69 (2.69) & 48.30 (1.85) & 40.49 (1.81) & \textbf{37.58 (1.77)} \\ 
  $\beta_7$ & 2.17 (0.06) & 2.17 (0.04) & 1.31 (0.05) & \textbf{1.27 (0.05)} \\ 
$\sigma^2$ & 113.55 (3.49) & 168.96 (2.39) & 84.01 (2.99) & \textbf{82.44 (2.91)} \\ 
   \hline
\end{tabular}
 \begin{tablenotes}
\footnotesize
\item[2] values for $\beta_0$-$\beta_7$ are scaled by $10^5$, and $\sigma^2$ is scaled by $10^2$
    \end{tablenotes}
\end{threeparttable}
\end{subtable}

\end{table}
\paragraph{Results} Tables~\ref{tab:arTab} and~\ref{tab:BLtab} show the results for the AR(1) and Bayesian LASSO example, respectively, for the chain length $M=10000$. We include tables for $M=40000$ in Supplementary Material Section\ifnum \pageoption>1~\ref{supp_sec: 5_tables}\else~S5\fi. For the AR(1) example, the weighted momentLS estimator performs the best for asymptotic variance estimation out of all of the considered estimators except for the AR(1) the maximum likelihood estimator, at both $\rho=0.9$ and $\rho=-0.9$ (Table~\ref{tab:avarErrorAR}). For spectral density estimation for the AR(1) example, the weighted momentLS again performs the best out of the considered estimators (excepting the AR(1) maximum likelihood estimator). 
For the Bayesian LASSO example, a similar pattern holds. For each of the coefficients $\sigma^2$ and $\beta_i$, $i=1,...,p$, the weighted momentLS estimator performs the best out of the considered estimators, for both asymptotic variance estimation as well as spectral density estimation.  Notably, we observe uniform improvement from the unweighted moment LS estimators.

\section{Conclusion}\label{sec: 6_conclusion}

In this work, we developed a novel weighted $\ell_2$ projection method for estimating autocovariance sequences and spectral density functions from reversible Markov chains. Building upon the least-squares estimation method for autocovariance sequences proposed in \citet{berg2023efficient}, we proposed an alternative approach based on the Fourier transform of the input autocovariance sequence, in order to take heteroscedasticity and dependence in the input autocovariance sequence into account. Our approach can be understood as estimating a spectral density function, while leveraging a mixture representation for the Fourier series of the spectral density function.  Under mild assumptions on the input sequence estimator and the weight function, we showed the consistency of our estimator for both spectral density function and asymptotic variance estimation. Finally, we empirically demonstrated improvement over other estimators, including our previous unweighted $\ell_2$ projection method of~\citet{berg2023efficient}, in both simulated and real data examples. 

We note some potential directions for future work. Multivariate extensions of our weighted least squares method are of interest for estimating asymptotic covariance matrices, cross-covariance sequences, and cross-spectral densities from reversible Markov chains, as in our multivariate extension~\citet{song2023multivariate} of the unweighted momentLS estimator. In terms of theoretical development, our analysis here shows consistency properties, but the statistical rates of convergence of momentLS estimators remain uncharacterized. We leave these questions as future research topics.

\section{Acknowledgements}
HS and SB gratefully acknowledge support from NSF DMS-2311141.

\ifnum\pageoption=1 
\bibliographystyle{plainnat}
\bibliography{bib}
\fi

\ifnum\pageoption=2 
\putbib 
\end{bibunit}
\fi 

\ifnum\pageoption>1
\newpage
\ifnum\pageoption=2\begin{bibunit}\fi
\setcounter{page}{1}
\renewcommand{\thepage}{S\arabic{page}} 
\renewcommand{\thesection}{S\arabic{section}}  
\renewcommand{\thetable}{S\arabic{table}}  
\renewcommand{\thefigure}{S\arabic{figure}}
\renewcommand{\theequation}{S-\arabic{equation}}
\renewcommand{\theremark}{S-\arabic{remark}}
\renewcommand{\thelem}{S-\arabic{lem}}
\renewcommand{\theprop}{S-\arabic{prop}}
\renewcommand{\thecor}{S-\arabic{cor}}
\setcounter{equation}{0}
\setcounter{section}{0}
\setcounter{remark}{0}
\setcounter{lem}{0}
\setcounter{prop}{0}
\setcounter{cor}{0}
\setcounter{table}{0}
\spacingset{1.5}
\begin{center}
{\Large\bf Supplement to ``\Title"}\\
\vspace{1em}
{\large Hyebin Song and Stephen Berg}\\
{\large Department of Statistics, Pennsylvania State University}
\end{center}

\section{Fourier Transform}\label{supp_sec: 1_Fourier}
In this section, we state definitions and present relevant facts regarding Fourier analysis, along with proofs for completeness.

\begin{definition}
For $f\in \ell_2(\Z,\C)$, we define the discrete-time Fourier transform $\hat{f}$ of $f$ as the $L^2$ limit of the partial sum sequence $S_N(\hat{f}) (\omega) = \sum_{|k|\le N} e^{-i\omega k} f(k)$ as $N\to\infty$, i.e., the function $\hat{f} \in L^2(\T)$ which satisfies
\begin{align*}
    \frac{1}{2\pi} \int_{[-\pi,\pi]} |\hat{f}(\omega) - S_N(\hat{f})(\omega)|^2 d\omega \to 0
\end{align*}
as $N\to\infty$.
\end{definition}
\begin{remark}
The existence of such an $\hat{f} \in L^2(\T)$ is due to the completeness of $L^2(\T)$ as $S_N(\hat{f})(\omega)= \sum_{|k|\le N} f(k)e^{-i\omega k}$ is a Cauchy sequence in $L^2(\T)$ (e.g., see Section 3 in \citealp{grafakos2008classical}).
\end{remark}
\begin{remark}\label{remark:FT_l1seq}
When $f \in \ell_1(\Z,\C)$, the Fourier transform $\hat{f}$ of $f$ is also the uniform limit of $S_N(\hat{f})$ on $[-\pi,\pi]$. For each $\omega \in [-\pi,\pi]$, define $\hat{f}(\omega)$ as the point-wise limit of $S_N(\hat{f})(\omega)$, so that $\hat{f}(\omega) = \lim_{N\to\infty} S_N(\hat{f})(\omega) = \lim_{N\to\infty} \sum_{|k|\le N} e^{-i\omega k} f(k)$. By the Weierstrass M-test,  \[\sup_{\omega \in [-\pi,\pi]} | S_N(\hat{f})(\omega)-\hat{f}(\omega)| \to 0.\] We note that this $\hat{f}$ is also the $L^2$ limit of $S_N(\hat{f})$ since \begin{align*}
  \|  S_N(\hat{f})- \hat{f}\|_{L^2(\T)}^2= \frac{1}{2\pi} \int_{[-\pi,\pi]}|S_N(\hat{f})(\omega)-\hat{f}(\omega) |^2  d\omega  \le \sup_{\omega \in [-\pi,\pi]} | S_N(\hat{f})(\omega)-\hat{f}(\omega)| ^2 \to 0
\end{align*}
\end{remark}

\begin{lem}\label{lem:L2_FourierCoefficients}
For any square summable $a=\{a(k)\} \in \ell_2(\Z,\C)$, let $\hat{a}\in L^2(\T)$ be the Fourier transform of $a$, i.e., the $L^2$ limit of $S_N(\hat{a}) (\omega)= \sum_{|l|\le N} a(l) e^{-i\omega l}$.  Define $u_k(\omega) = e^{-iwk}$ for $k\in \Z$. We have $ ( \hat{a}, u_k)=a(k)$ for each $k\in\Z$
\end{lem}
\begin{proof}
Let $k\in \Z$ be fixed and also let $\epsilon>0$ be given. By the $L^2$ convergence, we can find $N_0<\infty$ such that $\|S_N(\hat{a}) - \hat{a}\|_{L^2(\T)} \le \epsilon$ for $N \ge N_0$. 
Also note that for any $N \ge |k|$, $a(k) = ( S_N(\hat{a}), u_k) $ as $(u_l)_{|l|\le N}$ are orthogonal. Then for $N \ge \max\{N_0 ,|k|\}$
\begin{align*}
    |a(k) -  ( \hat{a},u_k ) | = | ( S_N(\hat{a}), u_k) - ( \hat{a},u_k )| \le \|S_N(\hat{a}) - \hat{a}\|_{L^2(\T)} \le \epsilon.
\end{align*}
Since $\epsilon$ was arbitrary, we have $(\hat{a}, u_k)=a(k)$.
\end{proof}
\begin{lem}\label{lem:partialSumL2limit}
For $\hat{a} \in L^2(\T)$, define $a(k) = ( \hat{a}, u_k) = \frac{1}{2\pi} \int_{[-\pi,\pi]}\hat{a}(\omega) e^{iwk} d\omega$ for each $k\in\Z$. Then $a=\{a(k)\}_{k\in\mathbb{Z}}\in\ell_2(\mathbb{Z},\mathbb{C})$, and $\hat{a}$ is the $L^2$ limit of $S_N(\hat{a})= \sum_{|k|\le N} a(k)u_k$. 
\end{lem}
\begin{proof}
We have $a\in\ell_2(\mathbb{Z},\mathbb{C})$ from Bessel's inequality and the fact $\hat{a}\in L^2(\T)$. Let $\tilde{a}$ be the Fourier transform of $\{a(k)\}_{k\in\Z}$, i.e., the $L^2$ limit of $\sum_{|k|\le N} a(k)e^{-i\omega k}$. We show that $\tilde{a} -\hat{a}=0$. We have that $\{u_k\}_{k\in\Z}$ is an orthogonal basis of $L^2(\T)$ (e.g., Theorem 8.20 in \citet{folland1999real}).  Therefore, it is sufficient to check $ ( \tilde{a}-\hat{a}, u_k) =0$ for all $k\in \Z$ to show $\tilde{a} -\hat{a}=0$. From Lemma \ref{lem:L2_FourierCoefficients}, we have $( \tilde{a}, u_k) = a(k)$, and $ ( \hat{a}, u_k)=a(k)$ by definition of $a(k)$, which completes the proof.
\end{proof}
\begin{prop}[Parseval's Identity]\label{prop:parseval-1}
For any square summable $a=\{a(k)\} \in \ell_2(\Z,\C)$, let $\hat{a}\in L^2(\T)$ be the Fourier transform of $a$, i.e., the $L^2$ limit of $S_N(\hat{a})= \sum_{|k|\le N} a(k)u_k$ where $u_k (\omega)= e^{-i\omega k}$. We have
\begin{equation}\label{eq:parseval-1}
    \frac{1}{2\pi} \int_{[-\pi,\pi]}|\hat{a}(\omega)|^2 d\omega  = \sum_{k\in\Z} |a(k)|^2\end{equation}
Conversely, for a given square integrable function $\hat{a}\in L^2(\T)$ , define $a(k) = (\hat{a}, u_k) $ for each $k\in \Z$. Then the equality \eqref{eq:parseval-1} holds. 
\end{prop}
\begin{proof}
We have
\begin{align*}
    \|\hat{a} - S_N(\hat{a})\|_{L^2(\T)}^2 &= \|\hat{a}\|_{L^2(\T)}^2 - 2 \Re( ( \hat{a}, \sum_{|k|\le N} a(k) u_k)) + \|\sum_{|k|\le N} a(k) u_k\|_{L^2(\T)}^2\\
    &=\|\hat{a}\|_{L^2(\T)}^2 -  \sum_{|k|\le N} |a(k) |^2 \to 0
\end{align*}
as $N\to \infty$, where we use the fact that $( \hat{a}, u_k ) = a(k)$ for each $k$.
\end{proof}
\begin{cor}\label{cor:parseval-2}
For square summable $a=\{a(k)\}, b=\{b(k)\} \in \ell_2(\Z,\C)$, let $\hat{a},\hat{b}$ be the Fourier transform of $a$ and $b$, respectively. Then
\begin{equation}\label{eq:parseval-2}
     \sum_{k\in \Z} a(k) \overline{b(k)}=\frac{1}{2\pi} \int_{[-\pi,\pi]} \hat{a}(\omega)\overline{\hat{b}(\omega)}d\omega 
\end{equation}    
Conversely, for given square integrable functions $\hat{a},\hat{b}\in L^2(\T)$, define sequences $a,b$  such that $a(k) = ( \hat{a}, u_k )$ and $b(k) = (\hat{b}, u_k )$ for each $k\in \Z$. Then the equality \eqref{eq:parseval-2} holds.
\end{cor}
\begin{proof}
Note for any $c_1,c_2\in \C$, the $L^2$ limit of $c_1 S_N(\hat{a})+c_2 S_N(\hat{b})$ is $c_1\hat{a} + c_2\hat{b}$. Moreover,  $c_1 S_N(\hat{a})+c_2 S_N(\hat{b}) = \sum_{|k|\le N} (c_1a(k) + c_2 b(k)) u_k$. Thus from \eqref{eq:parseval-1}, $\|c_1 \hat{a}+c_2\hat{b}\|_{L^2(\T)}^2 =\|c_1a + c_2 b\|_2^2$. Then the result follows from the polarization identity, as
\begin{align*}
   \langle a, b\rangle&=\frac{1}{4}(\|a+b\|_2^2-\|a-b\|_2^2+i\|a+i b\|_2^2-i\|a-i b\|_2^2)\\
   &=\frac{1}{4}(\|\hat{a}+\hat{b}\|_{L^2(\T)}^2-\|\hat{a}-\hat{b}\|_{L^2(\T)}^2+i\|\hat{a}+i \hat{b}\|_{L^2([-\pi,\pi])}^2-i\|\hat{a}-i\hat{ b}\|_{L^2(\T)}^2)\\
   &=( \hat{a},\hat{b}).
\end{align*}
\end{proof}

\begin{lem}\label{lem:a_ahat}
For $a,b \in \ell_2(\Z,\C)$, let $\hat{a}$ and $\hat{b}$ be the Fourier transform of $a$ and $b$ respectively. Then $a(k)=b(k)$ for all $k\in \Z$  iff $\hat{a}(\omega)=\hat{b}(\omega)$ for a.e. $\omega$
\end{lem}
\begin{proof}
We have $\hat{a},\hat{b} \in L^2([-\pi,\pi])$, and the result follows by Parseval's identity as $\|\hat{a} - \hat{b}\|_{L^2(\T)}^2=\|a-b\|_2^2$ .
\end{proof}

\begin{lem}\label{lem:FourierCoef}
Let $g \in C^2(\T)$ be a twice continuously differentiable function on the Torus. We have
\begin{equation}
(i k)^2\frac{1}{2\pi} \int_{[-\pi,\pi]} g(\omega) e^{ik \omega} d\omega=\frac{1}{2\pi} \int_{[-\pi,\pi]} g^{'' }(\omega) e^{i k \omega} d \omega 
\end{equation}
In particular, $| ( g, u_k ) u_k| \le \frac{1}{k^2} \frac{1}{2\pi} \int_{[-\pi,\pi]} |g^{'' }(\omega)| e^{i k \omega} d \omega  \le \|g''\|_\infty /k^2$ and the power series $\sum_{k\in\Z} ( g, u_k ) u_k$ is uniformly convergent on $[-\pi,\pi]$. 
\end{lem}
\begin{proof}
Using integration by parts, for any $k\ne 0$, for any continuously differentiable function $f$ on the circle, we have,
\begin{equation}\label{lem_FourierCoef:eq1}
-(i k)\frac{1}{2\pi} \int_{[-\pi,\pi]} f(\omega) e^{ik \omega} d\omega=\frac{1}{2\pi} \int_{[-\pi,\pi]} f^{\prime}(\omega) e^{i k \omega} d \omega 
\end{equation}
since
\begin{align*} 
\int_{[-\pi,\pi]} f(\omega) e^{i k \omega} d \omega
& =\left[f(\omega) \cdot \frac{e^{i k \omega}}{i k}\right]_{-\pi}^{ \pi}-\frac{1}{i k} \int_{[-\pi,\pi]} f'(\omega) e^{i k \omega} d \omega \\ 
& =-\frac{1}{i k} \int_{[-\pi,\pi]} f'(\omega) e^{i k \omega} d \omega 
\end{align*}
where we use the fact that $f(-\pi) = f(\pi)$ for the second equality. 
Applying  \ref{lem_FourierCoef:eq1} for $f=g$ and $f=g'$ gives the result.
 
\end{proof}

\section{A few technical Lemmas}\label{supp_sec: 2_lemmas}

The following Lemma \ref{lem: PoissonKernel} gives the Fourier transform $\hat{x}_\alpha$ of $x_\alpha$ for $\alpha\in(-1,1)$, and conversely, the inverse Fourier transform of $\hat{x}_\alpha$ is $x_\alpha$.

\begin{lem}\label{lem: PoissonKernel}
Let $x_\alpha:\Z\to \R$ be defined such that $x_\alpha(k) = \alpha^{|k|}$ for $k\in\Z$. For any $\alpha\in(-1,1)$, $x_\alpha \in \ell_1^{even}(\Z,\R)$, 
and its Fourier transformation $\hat{x}_\alpha$ is given as $\hat{x}_\alpha(\omega) = K(\alpha,\omega)$, where $K(\alpha,\omega) = \frac{1-\alpha^2}{1-2\alpha\cos(\omega) + \alpha^2}$. Moreover, $x_\alpha(k) = ( \hat{x}_\alpha, u_k)$ for any $k\in\Z$ and $u_k(\omega) = e^{-i\omega k}$.
\end{lem}
\begin{proof}
For any $\alpha\in(-1,1)$, $x_\alpha$ is absolutely summable, as $\sum_{k\in\Z} |\alpha|^{|k|} <\infty$. Therefore, by Remark \ref{remark:FT_l1seq}, $\hat{x}_\alpha$ is the pointwise limit of $S_N(\hat{x}_\alpha)(\omega)= \sum_{|k|\le N} \alpha^{|k|} e^{-i\omega k}$. For each $\omega \in [-\pi,\pi]$, we have
\begin{align*}
   \hat{x}_\alpha(\omega)= \sum_{k\in\Z}\alpha^{|k|}  e^{-i\omega k}   =  \sum_{k\in\Z}\alpha^{|k|} \cos(\omega k) =  2\sum_{k=0}^\infty \alpha^{k} \cos(\omega k)-1
\end{align*}
and
\begin{align*}
  2  \sum_{k=0}^\infty \alpha^k \cos(\omega k) -1 
  &= \sum_{k=0}^\infty \alpha^k (e^{i\omega k} + e^{-i\omega k}) -1 \\
  &= \frac{1}{1-\alpha e^{i\omega}} +  \frac{1}{1-\alpha e^{-i\omega}} -1 \\
  &=\frac{1-\alpha^2}{1-2\alpha\cos(\omega) + \alpha^2}.
\end{align*}
For the second statement, we first note that \begin{align*}
    (2\pi)^{-1} \int_{[-\pi,\pi]}  \sum_{l\in\Z} |\alpha^{|l|}  e^{-i\omega l}  e^{i\omega k}| d\omega \le (2\pi)^{-1} \int_{[-\pi,\pi]}  \sum_{l\in\Z}|\alpha|^{|l|}  d\omega = \sum_{l\in\Z}|\alpha|^{|l|}<\infty
\end{align*}
since $\alpha\in(-1,1)$. Then
\begin{align*}
    ( \hat{x}_\alpha, u_k) &= (2\pi)^{-1} \int_{[-\pi,\pi]} \hat{x}_\alpha(\omega) e^{i\omega k} d\omega\\
    &=(2\pi)^{-1} \int_{[-\pi,\pi]}  \sum_{l\in\Z}\alpha^{|l|}  e^{-i\omega l} e^{i\omega k} d\omega\\
    &=\sum_{l\in\Z} \frac{\alpha^{|l|}}{2\pi} \int_{[-\pi,\pi]}   e^{-i\omega (l-k)}d\omega\\
    &= \alpha^{|k|}
\end{align*}
\end{proof}

The following Lemma \ref{lem: fourier_xa} presents the Fourier transform of a moment sequence.

\begin{lem}\label{lem: fourier_xa}
Suppose $f \in  \mathscr{M}_{\infty}(0) \cap \ell_2(\mathbb{Z},\R)$ with the representing measure $F$, i.e., $f = \int x_\alpha \,F(d\alpha)$. 
The Fourier transformation $\hat{f}$ of $f$, which is the $L^2$ limit of the partial sum $S_N(\hat{f})(\omega) = \sum_{|k|\le N} f(k)e^{-i\omega k}$, is given by $\int K(\alpha,\omega) F(d\alpha).$
\end{lem}

\begin{proof}
Let $\tilde{f}(\omega) = \int K(\alpha,\omega) F(d\alpha).$
First, we show $\|\tilde{f}\|_{L^2(\T)}<\infty$. We have \begin{align*}
    \|\tilde{f}\|_{L^2(\T)}^2&=\frac{1}{2\pi}\int_{[-\pi,\pi]}|\tilde{f}(\omega)|^2d\omega\\
    &=\frac{1}{2\pi}\int_{[-\pi,\pi]}\int_{[-1,1]}K(\alpha,\omega) F(d\alpha)\int_{[-1,1]}K(\alpha',\omega)F(d\alpha)d\omega\\
    &=\frac{1}{2\pi}\int_{[-\pi,\pi]}\int_{(-1,1)}K(\alpha,\omega) F(d\alpha)\int_{(-1,1)}K(\alpha',\omega)F(d\alpha)d\omega\\
    &=\frac{1}{2\pi}\int_{[-\pi,\pi]}\int_{(-1,1)}\int_{(-1,1)}K(\alpha,\omega) K(\alpha',\omega)F(d\alpha')F(d\alpha)d\omega\\
    &=\int_{(-1,1)}\int_{(-1,1)}\frac{1}{2\pi}\int_{[-\pi,\pi]}K(\alpha,\omega) K(\alpha',\omega)d\omega F(d\alpha')F(d\alpha)\\
    &=\int_{[-1,1]}\int_{[-1,1]}\frac{1+\alpha\alpha'}{1-\alpha\alpha'} F(d\alpha')F(d\alpha) <\infty 
\end{align*} The third and second-to-last equalities follow from Lemma 2 of \citet{berg2023efficient}, which shows $F(\{-1,1\})=0$. The last equality follows from Corollary 2 of Lemma 4 of \citet{berg2023efficient}. The change of the order of integration is justified since $K(\alpha,\omega)>0$ for each $\alpha\in (-1,1)$. 

Now, we show that $\tilde{f}$ is the $L^2$ limit of the Cauchy sequence $S_N(\hat{f})$ in $L^2([-\pi,\pi])$.
By Lemma \ref{lem:partialSumL2limit}, it is sufficient to show that \begin{align} ( \tilde{f}, u_k ) = \frac{1}{2\pi}\int_{[-\pi,\pi]}\tilde{f}(\omega)e^{i\omega k}d\omega = f(k).\label{eq:fourierCoeff}\end{align}

We have
\begin{align*}
    \frac{1}{2\pi}\int_{[-\pi,\pi]}\hat{f}(\omega)e^{i\omega k} d\omega= \frac{1}{2\pi}\int_{[-\pi,\pi]} \int K(\alpha,\omega) e^{i\omega k} F(d\alpha) d\omega
\end{align*}
First we show
\begin{align*}
    \frac{1}{2\pi}\int_{[-\pi,\pi]} \int |K(\alpha,\omega) e^{i\omega k}| F(d\alpha') d\omega <\infty.
\end{align*}
We have
\begin{align*}
    \frac{1}{2\pi}\int_{[-\pi,\pi]} \int |K(\alpha,\omega) e^{i\omega k}| F(d\alpha) d\omega
    &=\frac{1}{2\pi}\int_{[-\pi,\pi]} \int K(\alpha,\omega)  F(d\alpha) d\omega\\
    &= \int \{\frac{1}{2\pi}\int_{[-\pi,\pi]} K(\alpha,\omega)  d\omega\} F(d\alpha) \\
    &= \int 1 F(d\alpha)  = f(0) < \infty
\end{align*}
where for the second equality we use Tornelli's Theorem and the third equality is due to $\int 1 F(d\alpha) = \int \alpha^0 F(d\alpha) = f(0)$.

Then by Fubini,
\begin{align*}
    \frac{1}{2\pi}\int_{[-\pi,\pi]} \int K(\alpha,\omega) e^{i\omega k} F(d\alpha) d\omega =  \int \{ \frac{1}{2\pi}\int_{[-\pi,\pi]} K(\alpha,\omega) e^{i\omega k} d\omega\}F(d\alpha)  = \int \alpha^{|k|} F(d\alpha)
\end{align*}
where we use $(2\pi)^{-1}\int_{[-\pi,\pi]} K(\alpha,\omega) e^{i\omega k} d\omega   =  ( \hat{x}_\alpha, u_k ) = \alpha^{|k|}$ from Lemma \ref{lem: PoissonKernel}.
\end{proof}

We now show when for a moment sequence $f = \int x_\alpha F(d\alpha)$ represented as a mixture of $x_\alpha$, a weighted inner product with $g$ can also be represented as a mixture of the weighted inner product of $x_\alpha$ with $g$:
\begin{lem}\label{lem:inner_product_rep}
Suppose $f \in \mathscr{M}_{\infty}(0) \cap \ell_2(\mathbb{Z},\R)$ with the representing measure $F$ and $g \in \ell_2(\mathbb{Z},\C)$. Suppose the weight function $\phi$ satisfies \ref{cond:phi}. Then
$$
\langle f, g\rangle_\phi =\int\left\langle x_\alpha, g\right\rangle_\phi  F(d \alpha).
$$
\end{lem}
\begin{proof}
Recall that by the definition of $\langle\cdot, \cdot \rangle_\phi$, we have,
$\langle f,g \rangle_\phi  = \frac{1}{2\pi}\int_{[-\pi,\pi]} \frac{\hat{f}(\omega)\overline{\hat{g}(\omega)}}{\phi(\omega)^2}d\omega$ and $\langle x_\alpha,g\rangle_\phi = \frac{1}{2\pi} \int \frac{K(\alpha,\omega) \overline{\hat{g}(\omega)}}{\phi(\omega)^2} d\omega $ where $\hat{f}$ and $\hat{g}$ are the $L^2$ limits of $S_N(\hat{f})=\sum_{|k|\le N} f(k)u_k$ and $S_N(\hat{g})=\sum_{|k|\le N} g(k)u_k$. We need to show
\begin{align}\label{lem_inner_product_rep:eq1}
    \frac{1}{2\pi}\int_{[-\pi,\pi]} \frac{\hat{f}(\omega)\overline{\hat{g}(\omega)}}{\phi(\omega)^2}d\omega = \int \left\lbrace  \frac{1}{2\pi} \int_{[-\pi,\pi]} \frac{K(\alpha,\omega) \overline{\hat{g}(\omega)}}{\phi(\omega)^2} d\omega \right\rbrace F(d\alpha).
\end{align}
Recall that $\hat{f}(\omega) = \int K(\alpha,\omega)F(d\alpha)$ is the $L^2$ limit of $S_N(\hat{f})$ by Lemma \ref{lem: fourier_xa}.
If we can exchange the order of integration in \eqref{lem_inner_product_rep:eq1}, 
\begin{align*}
    \frac{1}{2\pi}  \int_{[-\pi,\pi]}  \int K(\alpha,\omega) F(d\alpha) \frac{ \overline{\hat{g}(\omega)}}{\phi(\omega)^2}   d\omega  =\frac{1}{2\pi}  \int_{[-\pi,\pi]}  \hat{f}(\omega) \frac{ \overline{\hat{g}(\omega)}}{\phi(\omega)^2}   d\omega  
\end{align*}
which proves the result. Now we justify exchanging the order of integrations in the RHS of equation \eqref{lem_inner_product_rep:eq1}. We have for any $N\in \N$,
\begin{align*}
 &\frac{1}{2\pi}  \int_{[-\pi,\pi]} \int \frac{|K(\alpha, \omega) \overline{\hat{g}(\omega)}|}{\phi(\omega)^2}F(d\alpha) d\omega \\
 &\le  \frac{1}{2c_1^2\pi}  \int_{[-\pi,\pi]} \int |K(\alpha, \omega) \overline{\hat{g}(\omega)}|F(d\alpha) d\omega \\
  &= \frac{1}{2c_1^2\pi}  \int_{[-\pi,\pi]}  |\overline{\hat{g}(\omega)}| \int K(\alpha, \omega)F(d\alpha)  d\omega \\
  &= \frac{1}{2c_1^2\pi}  \int_{[-\pi,\pi]}  |\overline{\hat{g}(\omega)}| \hat{f}(\omega)  d\omega  \\
  &\le  \frac{1}{c_1^2}\|  \hat{g} \|_{L^2(\T)} \| \hat{f}\|_{L^2(\T)}
\end{align*}
where we use the fact that $K(\alpha,\omega) \ge 0$ for all $\alpha, \omega$ for the second equality and use Holder's inequality for the last inequality. By Parseval's identity \eqref{eq:parseval-1}, $\|\hat{g}\|_{L^2(\T)} = \|g\|_2 <\infty$ by assumption of $g$. Also $\| \hat{f}\|_{L^2(\T)}<\infty$ since $\hat{f} \in L^2([-\pi,\pi])$.

\end{proof}

We now show in Lemma \ref{lem: l1projection} that the projection of an $\ell_1$ sequence remains absolutely summable. We remark that since the unweighted projection $\Pi(r;C)$ satisfies $\Pi(r;C)=\Pi^{\phi}(r;C)$ with the choice of weight $\phi\equiv 1$, Lemma~\ref{lem: l1projection} implies that $\Pi(r,C)\in\ell_1(\mathbb{Z},\R)$ in the case $r\in\ell_1(\mathbb{Z},\R)$.

\begin{lem}\label{lem: l1projection}
Let $r\in \ell_1(\mathbb{Z},\R)$ be an even sequence. Suppose a weight function $\phi$ satisfies \ref{cond:phi}. For any closed $C \subseteq [-1,1]$, $\Pi^{\phi}(r;C)\in \ell_1(\mathbb{Z},\R)$.
\end{lem}
\begin{proof}
Let $\mu_C^{\phi}$ be the representing measure of $\Pi^{\phi}(r;C)$. By definition, we have $\Pi^{\phi}(r;C) = \int x_\alpha \,\mu_C^{\phi} (d\alpha)$ and we observe\begin{align*}
    \| \Pi^{\phi}(r;C) \|_1&=\|\int x_{\alpha}\,\mu_C^{\phi}(d\alpha)\|_1\\
    &=\|\int_{[-1,0)}x_{\alpha}\,\mu_C^{\phi}(d\alpha)+\int_{[0,1]}x_{\alpha}\,\mu_C^{\phi}(d\alpha)\|_1\\
		&\leq \|\int_{[-1,0)}x_{\alpha}\,\mu_C^{\phi}(d\alpha)\|_1+\|\int_{[0,1]}x_{\alpha}\,\mu_C^{\phi}(d\alpha)\|_1\\
		&\leq \int_{[-1,0)}\|x_{\alpha}\|_1\,\mu_C^{\phi}(d\alpha)+\int_{[0,1]}\|x_{\alpha}\|_1\,\mu_C^{\phi}(d\alpha).
	\end{align*}
Therefore, it is sufficient to show $\int_{[-1,0)}\|x_{\alpha}\|_1\,\mu_C^{\phi}(d\alpha)<\infty$ and $\int_{[0,1]}\|x_{\alpha}\|_1\,\mu_C^{\phi}(d\alpha)<\infty$

We have $\|x_\alpha\|_1 = \sum_{k\in\Z} |\alpha|^{|k|} = \frac{1+|\alpha|}{1-|\alpha|}$. We first consider the case where the endpoints $-1$ or $1$ are not in $\operatorname{Supp}(\mu_C^{\phi})$. If $-1$ is not in $\operatorname{Supp}(\mu_C^{\phi})$, we can find a $\delta>0$ open neighborhood of $-1$ with $0$ measure, i.e., there exists
$\delta >0$ such that $\mu_C^{\phi}(N_\delta(-1)) = 0$ where $N_\delta(-1) = \{\alpha; |\alpha+1| < \delta\}$. In particular, $\inf \operatorname{Supp}(\mu_C^{\phi}) \ge -1+\delta$. Then
\begin{align*}
    \int_{[-1,0)}\|x_{\alpha}\|_1\,\mu_C^{\phi}(d\alpha)\leq \int_{[-1+\min\{\delta,1\},0]}\frac{1+|\alpha|}{1-|\alpha|}\,\mu_C^{\phi}(d\alpha)\le   \frac{2+\min\{\delta,1\}}{\min\{\delta,1\} } \mu_C^{\phi}([-1,1])<\infty.
\end{align*} 
Similarly, if $1 \notin \operatorname{Supp}(\mu_C^{\phi})$, there exists $\delta'>0$ such that $\sup \operatorname{Supp}(\mu_C^{\phi}) \le 1-\delta'$, and 
\begin{align*}
    \int_{[0,1]}\|x_{\alpha}\|_1\,\mu_C^{\phi}(d\alpha)\leq \int_{[0,1-\min\{\delta',1\}]}\frac{1+|\alpha|}{1-|\alpha|} \mu_C^{\phi}(d\alpha)\le   \frac{2-\min\{\delta',1\}}{\min\{\delta',1\}} \mu_C^{\phi}([-1,1])<\infty.
\end{align*} 

Now we consider the case where $-1 \in \operatorname{Supp}(\mu_C^{\phi})$ and show $\int_{[-1,0)}\|x_{\alpha}\|_1\,\mu_C^{\phi}(d\alpha)<\infty$.
Since $\Pi^{\phi}(r;C)$ in in $\ell_2(\Z,\R)$, by  Lemma 2 of \citet{berg2023efficient}, $\mu_C^\phi(\{-1\})=0$. 
Thus there exists $(\beta_n)_{n=1}^\infty$ with $\beta_n\in \operatorname{Supp}(\mu_C^\phi)$, $\beta_n<0$, $\forall n$, \, and $\beta_n\downarrow -1$.
Since for $\alpha\in(-1,0), \|x_\alpha\|_1 = \sum_{k\in\Z} |\alpha^{|k|}| = \sum_{k\in\Z} (-1)^{|k|}\alpha^{|k|}= \sum_{k\in\Z} \lim_{n\to \infty} \beta_n^{|k|}  \alpha^{|k|}$, and
the term $\beta_n^{|k|}  \alpha^{|k|}$ is point-wise dominated by absolutely summable $|\alpha|^{|k|}$, by dominated convergence theorem $\|x_\alpha\|_1 =\lim_{n\to \infty}  \sum_{k\in\Z} \beta_n^{|k|}  \alpha^{|k|} = \lim_{n\to \infty} \langle x_{\beta_n}, x_\alpha \rangle$. Finally we note that $\langle x_{\beta_n}, x_\alpha \rangle =\sum_{k\in \Z} (\beta_n \alpha)^{|k|} = \frac{1+\alpha\beta_n}{1-\alpha\beta_n}$. Then 
\begin{align*}\int_{[-1,0)}\|x_{\alpha}\|_1\,\mu_C^{\phi}(d\alpha)&=\int_{[-1,0)}\,\underset{n\to\infty}{\lim}\,\braket{x_{\beta_n},x_{\alpha}}\mu_C^{\phi}(d\alpha)\\
 &= \int_{[-1,0)}\,\underset{n\to\infty}{\lim}\,\frac{1+\alpha\beta_n}{1-\alpha\beta_n}\,\mu_C^{\phi}(d\alpha)\\
  &= \underset{n\to\infty}{\lim}\int_{[-1,0)} \frac{1+\alpha\beta_n}{1-\alpha\beta_n}\,\mu_C^{\phi}(d\alpha)\\
        &\leq \underset{n\to\infty}{\lim}\int_{[-1,1]}\braket{x_{\beta_n},x_{\alpha}}\mu_C^{\phi}(d\alpha)
\end{align*} 
In the third equality, we used the monotone convergence  $\frac{1+\alpha\beta_n}{1-\alpha\beta_n}\uparrow \frac{1+|\alpha|}{1-|\alpha|}$ for $\alpha<0$, since $\beta_n\downarrow -1$, in order to interchange the limit and the integral. In the inequality, we used $\braket{x_{\alpha},x_{\alpha'}}=\frac{1+\alpha\alpha'}{1-\alpha\alpha'}>0$ for all $\alpha,\alpha'\in(-1,1)$. 

Note that for any even $f \in {\ell}_2^{even}(\Z,\R)$ with Fourier transform $\hat{f}$, $\langle f, x_\alpha \rangle = (2\pi)^{-1}\int \hat{f}(\omega) K(\alpha,\omega) d\omega$ by Parseval's equality and by assumption on $\phi$ we have $c_1^{-2} \le \phi(\omega)^{-2}\le c_0^{-2}$. Therefore, we have,
\begin{align*}
    \frac{1}{2\pi c_1^2}\int K(\alpha,\omega) \hat{f}(\omega) d\omega \le \frac{1}{2\pi }\int \frac{K(\alpha,\omega) \hat{f}(\omega)}{\phi(\omega)^2} d\omega \le \frac{1}{2\pi c_0^2}\int  K(\alpha,\omega) \hat{f}(\omega) d\omega,
\end{align*}
i.e.,
\begin{equation}\label{lem_l1projection:ineq1}
    c_1^{-2}  \langle x_\alpha,f\rangle  \le \langle x_\alpha,f\rangle_\phi \le c_0^{-2}  \langle x_\alpha,f\rangle.
\end{equation}
Now, using $f=x_{\beta_n}$ in \eqref{lem_l1projection:ineq1},
\begin{align*}
    \int_{[-1,0)}\|x_{\alpha}\|_1\,\mu_C^{\phi}(d\alpha) 
    &\le c_1^2  \lim_{n\to\infty} \int_{[-1,1]} \langle x_\alpha,x_{\beta_n}\rangle_\phi \mu_C^{\phi}(d\alpha)\\
    &= c_1^2 \lim_{n\to\infty} \langle \Pi^\phi(r; C), x_{\beta_n } \rangle_\phi 
\end{align*}
where for the last equality we use Lemma \ref{lem:inner_product_rep}. Also since $\beta_n \in \operatorname{Supp}(\mu_C^{\phi})$ for each $n$, by Proposition \ref{prop:weighted_opt_eq}, $\langle \Pi^\phi(r; C), x_{\beta_n } \rangle_\phi = \langle r, x_{\beta_n } \rangle_\phi$. By applying inequality \eqref{lem_l1projection:ineq1} with $f=r$ and $\alpha=\beta_n$, we have,
$\langle r, x_{\beta_n } \rangle_\phi \le c_0^{-2} \langle r, x_{\beta_n } \rangle $. Combining these results,
\begin{align*}
     \int_{[-1,0)}\|x_{\alpha}\|_1\,\mu_C^{\phi}(d\alpha) \le c_1^2\lim_{n\to\infty }\langle r, x_{\beta_n } \rangle_\phi \le (c_1/c_0)^2 \lim_{n\to\infty } \langle r,x_{\beta_n}\rangle 
\end{align*}
Finally, 
$\langle r,x_{\beta_n}\rangle = \sum_{k\in \Z} r(k) \beta_n^{|k|} \le \sum_{k\in \Z} |r(k)| |\beta_n|^{|k|}  \le\sum_{k\in \Z} |r(k)| = \|r\|_1$, and therefore $\int_{[-1,0)}\|x_{\alpha}\|_1\,\mu_C^{\phi}(d\alpha)< \infty$ by assumption on $r$.

The case when $1\in \operatorname{Supp} (\mu_C^{\phi})$ can be similarly handled, where we can show $\int_{[0,1]}\|x_{\alpha}\|_1\,\mu(d\alpha)<\infty$ by following the same steps, but with $(\beta_n)_{n \ge 1}$ such that $\beta_n \in \operatorname{Supp} (\mu_C^{\phi})$, $\beta_n >0$ for all $n$, and $\beta_n \uparrow 1$.
\end{proof}

The last Lemma \ref{lem: K_bound} concerns pointwise lower and upper bounds of the Poisson kernel, as well as its first and second derivatives.
\begin{lem}\label{lem: K_bound}
 Let $\alpha \in (-1,1)$ be given. For $K(\alpha,\omega) = \frac{1-\alpha^2}{1-2\alpha\cos\omega+\alpha^2}$, we have
\begin{align*}
 &\inf_{\omega\in[-1,1]} K(\alpha,\omega) \ge \frac{1-|\alpha|}{1+|\alpha|} , \mbox{ and} \\
 &\max\{\sup_{\omega\in[-1,1]} |K(\alpha,\omega)|,\,\sup_{\omega\in[-1,1]} |\frac{d}{d\omega}K(\alpha,\omega)|,\,\sup_{\omega\in[-1,1]} |\frac{d^2}{d\omega^2}K(\alpha,\omega)|\} \le \frac{10}{(1-|\alpha|)^6}.
\end{align*}
 \end{lem}
\begin{proof}
First we note that for any $\alpha \in (-1,1)$, $1-\alpha^2 \ge 0$.
\begin{align*}
    K(\alpha,\omega) &\ge \frac{1-\alpha^2}{1+2|\alpha|+\alpha^2} = \frac{1-|\alpha|}{1+|\alpha|}\\
    K(\alpha,\omega) &\le \frac{1-\alpha^2}{1-2|\alpha|+\alpha^2} = \frac{1+|\alpha|}{1-|\alpha|}.
\end{align*}
as $1-\alpha^2 \ge 0$ for $\alpha\in(-1,1)$.

By taking the derivative of $K(\alpha,\omega)$ with respect to $\omega,$
\begin{align*}
    \frac{d}{d\omega} K(\alpha,\omega) = \frac{-2\alpha(1-\alpha^2) \sin(\omega)}{(1+\alpha^2 - 2\alpha \cos\omega)^2}
\end{align*}
In particular,
\begin{align*}
    |\frac{d}{d\omega} K(\alpha,\omega)| = \frac{2|\alpha(1-\alpha^2) \sin(\omega)|}{(1+\alpha^2 - 2\alpha \cos\omega)^2} \le \frac{2|\alpha|(1-\alpha^2)}{(1+|\alpha|^2 -2|\alpha|)^2} \leq  \frac{4|\alpha|}{(1-|\alpha|)^3}
\end{align*}
where we use the fact that for any $\alpha \in (-1,1)$,  $1+\alpha^2 -2\alpha \cos(\omega) \ge 1+\alpha^2 -2|\alpha| \ge 0$.
Lastly,
\begin{align*}
\frac{d^2}{d\omega^2} K(\alpha,\omega)  = 
    (1 - \alpha^2) \{ \frac{-2 \alpha \cos(\omega)}{(1 + \alpha^2 - 2 \alpha \cos(\omega))^2} + \frac{8 \alpha^2 \sin^2(\omega)}{(1 + \alpha^2 - 2 \alpha \cos(\omega))^3} \}
\end{align*}
and
\begin{align*}
    |\frac{d^2}{d\omega^2} K(\alpha,\omega)|  \le
    (1 - \alpha^2) \{ \frac{2 |\alpha| }{(1 -|\alpha|)^4} + \frac{8 \alpha^2 }{(1 -|\alpha|)^6} \} \le \frac{2|\alpha|(1+4|\alpha|)}{(1 -|\alpha|)^6}.
\end{align*}
Therefore
\begin{align*}
    &\max\{\sup_{\omega\in[-1,1]} |K(\alpha,\omega)|,\sup_{\omega\in[-1,1]} |\frac{d}{d\omega}K(\alpha,\omega),|\sup_{\omega\in[-1,1]} |\frac{d^2}{d\omega^2}K(\alpha,\omega)|\} \\
    &\le  \max \{\frac{1+|\alpha|}{1-|\alpha|}, \frac{4|\alpha|}{(1-|\alpha|)^3}, \frac{2|\alpha|(1+4|\alpha|)}{(1 -|\alpha|)^6}\} \le \frac{10}{(1 -|\alpha|)^6}
\end{align*}
\end{proof}

\section{Proofs for Section \ref{sec: 3_weighted-momentLS}}

\subsection{Proof of Proposition \ref{prop: existence_and_uniqueness}}\label{supp_sec: 3_1_unique}
\begin{proof}
The first part of the Proposition is the consequence of Hilbert space projection theorem, noting that $\ell_2(\Z,\C)$ is a Hilbert space equipped with the weighted inner product norm $\langle \cdot, \cdot \rangle_\phi$ (ref. Remark \ref{rmk: weighted_inner_product}), and $\mathscr{M}_\infty(C) \cap \ell_2(\Z,\R) $ is a closed and convex subset of $\ell_2(\Z,\C)$ with respect to the weighted norm $\|\cdot \|_\phi$ due to the equivalence of the $\ell_2$ norm $\|\cdot\|_2$ and the weighted norm $\|\cdot\|_\phi$ under Assumption \ref{cond:phi}. For the second part, from the definition of $\Pi^{\phi}(r;C)$, there exists a measure $\mu \in \mathcal{M}_\R$ supported on $C$ such that $\Pi^{\phi}(r;C) = \int x_\alpha \mu (d\alpha)$. On the other hand, since $\Pi^{\phi}(r;C) \in \mathscr{M}_\infty(C) \subseteq \mathscr{M}_\infty([-1,1])$, from Proposition 2 in \citet{berg2023efficient} (with $a=-1$ and $b=1$), we know that there exists a unique $\mu'$ supported on $[-1,1]$ such that $\Pi^{\phi}(r;C) = \int x_\alpha \mu' (d\alpha)$. Since $\operatorname{Supp}(\mu) \subseteq C \subseteq [-1,1]$, $\mu = \mu'$, which concludes the proof.
\end{proof}

\subsection{Proofs on total positivity of Kernel $K(\alpha,\omega)$}\label{supp_sec: 3_2_tp}

\subsubsection{Proof of Lemma \ref{lem:tp_K}}
\begin{proof}
    First suppose that $\alpha_i \ne 0$ for all $i = 1,\dots,n$.
    Then, we have 
    \begin{align*}
        \tilde{K}(\alpha_i,x_j) = \frac{1-\alpha_i^2}{1+\alpha_i^2 - 2\alpha_i x_j} = \frac{\frac{1-\alpha_i^2}{2\alpha_i}}{\frac{1+\alpha_i^2}{2\alpha_i}-x_j} = \frac{f_1(\alpha_i)}{f_2(\alpha_i) -x_j}
    \end{align*}
       where we define $f_1(\alpha) = \frac{1-\alpha^2}{2\alpha}$ and $f_2(\alpha) = \frac{1+\alpha^2}{2\alpha}$ for $\alpha \ne 0$. Note that $f_1(\alpha_i)$ factor is shared for all elements in the $i$th row. Then
\begin{align*}
\mathbf{M} = 
    \begin{bmatrix}
          \frac{f_1(\alpha_1)}{f_2(\alpha_1)-x_1} &\cdots & \frac{f_1(\alpha_1)}{f_2(\alpha_1)-x_n}\\
            \vdots & \cdots &\vdots\\
            \frac{f_1(\alpha_n)}{f_2(\alpha_n)-x_1} &\cdots & \frac{f_1(\alpha_n)}{f_2(\alpha_n)-x_n}\\
       \end{bmatrix} =   \begin{bmatrix}
    f_{1}(\alpha_1) &0 &\cdots \\
    \vdots & \ddots & \vdots\\
   0 &\cdots  & f_1(\alpha_n)
  \end{bmatrix}  \begin{bmatrix}
          \frac{1}{f_2(\alpha_1)-x_1} &\cdots & \frac{1}{f_2(\alpha_1)-x_n}\\
            \vdots & \cdots &\vdots\\
            \frac{1}{f_2(\alpha_n)-x_1} &\cdots & \frac{1}{f_2(\alpha_n)-x_n}\\
       \end{bmatrix} = \mathbf{M}_1  \mathbf{M}_2
\end{align*}
where $ \mathbf{M}_1 =\mathrm{diag}(\{\frac{1-\alpha_i^2}{2\alpha_i}\}_{i=1}^n)$ and $ {\mathbf{M}_2} = [\frac{1}{f_2(\alpha_i)-x_j)}]_{i,j=1}^n$. By the property of determinant, $|\mathbf{M}| = |\mathbf{M_1}||\mathbf{M_2}|$. Let $m$ be the number of $\alpha_i$ such that $\alpha_i<0$.

Since $\mathbf{M}_1$ is a diagonal matrix, $|\mathbf{M}_1| = \prod_{i=1}^n \frac{1-\alpha_i^2}{2\alpha_i}$. 
In particular, since $1-\alpha_i^2 >0$, 
\begin{align}\label{lem_TP:Diag_M1_sgn}
\textrm{sgn}(|\mathbf{M}_1|) = \prod_{i=1}^n \textrm{sgn}(\alpha_i) = (-1)^{m}    
\end{align}

For $|\mathbf{M}_2|$, we use the determinant property of a Cauchy matrix. 
A $m \times n$ matrix  $\mathbf{A}$ is a Cauchy matrix if the elements $a_{ij}$ are in the form
$a_{i j}=\frac{1}{u_i-v_j}$ where $u_i-v_j \neq 0$, for $1 \leq i \leq m$, $1 \leq j \leq n$, and $\left(u_i\right)$ and $\left(v_j\right)$ contain distinct elements. We first note $f_2$ is strictly decreasing on $(-1,0)$ and $(0,1)$, since $f'_2(\alpha) =  \frac{\alpha^2 - 1}{2\alpha^2}<0$. In particular, $(f_2(\alpha_j))_{j=1}^n$ are distinct values as $\alpha_1<\alpha_2<\dots<\alpha_n$ are assumed to be distinct.  We have $\mathbf{M}_2$ is a Cauchy matrix, as $i,j$th element of $\mathbf{M}_2$ has the form of $\frac{1}{f_2(\alpha_i) - x_j}$. 

By being a Cauchy matrix, the determinant of $\mathbf{M}_2$ is non-zero, and moreover, the determinant of $\mathbf{M}_2$ is explicitly given by
\begin{align}\label{lem_TP:Cauchy_det}
   |\mathbf{M}_2| =  \frac{\prod_{i=2}^n \prod_{j=1}^{i-1}\left(f_2(\alpha_j)-f_2(\alpha_i)\right)\left(x_i-x_j\right)}{\prod_{i=1}^n \prod_{j=1}^n\left(f_2(\alpha_i)-x_j\right)}
\end{align}
First for the denominator of \eqref{lem_TP:Cauchy_det}, for any $i,j$, 
\begin{equation} \label{lem_TP:eq_1}
    f_2(\alpha_i) - x_j = \frac{1+\alpha_i^2}{2\alpha_i} - x_j = \frac{1+\alpha_i^2 - 2\alpha_i x_j }{2\alpha_i} = \frac{(\alpha_i -x_j )^2 + (1-x_j^2)}{2\alpha_i}
\end{equation}
and since $x_j \in (-1,1)$, the numerator in \eqref{lem_TP:eq_1} is always non-negative. Therefore, 
\begin{align*}
    \textrm{sgn}(\prod_{i=1}^n \prod_{j=1}^n\left(f_2(\alpha_i)-x_j\right)) = \prod_{i=1}^n\prod_{j=1}^n \textrm{sgn}(\alpha_i) = (-1)^{mn}
\end{align*}
where we recall that $m$ is the number if $\alpha_i$ such that $\textrm{sgn}(\alpha_i) <0$.
For the numerator of \eqref{lem_TP:Cauchy_det}, for any $i>j$, $\textrm{sgn}(x_i - x_j) =1$ as $x_i > x_j$ by the choice of $x_i$s. 
Then,
\begin{align*}
    \textrm{sgn}(\prod_{i=2}^n \prod_{j=1}^{i-1}\left(f_2(\alpha_j)-f_2(\alpha_i)\right)\left(x_i-x_j\right)) = \textrm{sgn}(\prod_{i=2}^n \prod_{j=1}^{i-1}\left(f_2(\alpha_j)-f_2(\alpha_i)\right))
\end{align*}
On the other hand, since $f_2$ is not monotone on $(a,b)\setminus \{0\}$ for $a<0<b$ as $f_2(\alpha) < 0 $ when $\alpha<0$ and $f_2(\alpha)>0$ for $\alpha>0$, the sign of $f_2(\alpha_j) - f_2(\alpha_i)$ are not always the same. 

In particular, when $m=0$ (all positive $\alpha$s) or $m=n$ (all negative $\alpha$s), $\textrm{sgn}(\prod_{i=2}^n \prod_{j=1}^{i-1}\left(f_2(\alpha_j)-f_2(\alpha_i)\right)) = 1$. When $0<m<n$, 
$\textrm{sgn}(\prod_{i=2}^n \prod_{j=1}^{i-1}\left(f_2(\alpha_j)-f_2(\alpha_i)\right)) =\prod_{i=(m+1)}^n  \prod_{j=1}^{i-1}  \textrm{sgn}(\alpha_j)=  \prod_{i=(m+1)}^n \prod_{j=1}^{m} (-1) = (-1)^{m(m-n)}$. Therefore,
\begin{align}\label{lem_TP:eq_2}
    \textrm{sgn}(|M_2|) = (-1)^{nm}  (-1)^{m(n-m)} = (-1)^{2nm-m^2}
\end{align}
and $\textrm{sgn}(\mathbf{M})  = \textrm{sgn}(\mathbf{M}_1) \textrm{sgn}(\mathbf{M}_2)  = (-1)^{m + 2nm -m^2}  = (-1)^{m(m-1) + 2nm} = 1$ as desired when no $\alpha_i \ne 0$.

Now consider the case that one of $\alpha$ is zero. We continue to let $m$ be the number of $\alpha$s such that $\alpha_i <0$. Since we assume each $\alpha$s to be distinct, we have $-1 <\alpha_1<\dots<\alpha_m<\alpha_{m+1} = 0<\alpha_{m+2}<\dots\alpha_n$. 

Since $\tilde{K}(0,x) = 1$ regardless of the value of $x$, $\mathbf{M}_{m+1,k}=1$ for all $k=1,\dots,n$. Similarly, we have $\mathbf{M} = \mathbf{M}_1' \mathbf{M}_2'$ where $\mathbf{M}_1' = \textrm{diag}(\{f_1(\alpha_1),\dots,f_1(\alpha_{m}), 1, f_1(\alpha_{m+2}),\dots,f_1(\alpha_n)\})$ and
\begin{align*}
\mathbf{M}_{2,ij}' =\begin{cases}
    \frac{1}{f_2(\alpha_i) - x_j} & i \ne m+1\\
    1 & i = m+1 
\end{cases}
\end{align*}
Since $\mathbf{M}_1'$ is a diagonal matrix, $|\mathbf{M}_1'| = \prod_{i\ne (m+1)} f_1(\alpha_i)=\prod_{i\ne (m+1)} \frac{1-\alpha_i^2}{2\alpha_i}$ and therefore $ \textrm{sgn}(|M_1'|) = (-1)^{m}$.

To compute the determinant of $\mathbf{M}_2'$, we consider the determinant of the matrix which we obtain subtracting the first column from each of the columns from 2 to $n$. Then
\begin{align*}
|\mathbf{M}_2'| &=\left\lvert 
    \begin{bmatrix}
          \frac{1}{f_2(\alpha_1)-x_1} &\frac{1}{f_2(\alpha_1)-x_2} &\cdots & \frac{1}{f_2(\alpha_1)-x_n}\\
            \vdots &\vdots &\cdots &\vdots\\
             1 & 1&\dots & 1 \\
             \vdots &\vdots &\cdots &\vdots\\
            \frac{1}{f_2(\alpha_n)-x_1} &\frac{1}{f_2(\alpha_n)-x_2}&\cdots & \frac{1}{f_2(\alpha_n)-x_n}\\
       \end{bmatrix} \right\rvert\\
       &=   \left\lvert 
       \begin{bmatrix}
          \frac{1}{f_2(\alpha_1)-x_1} &\frac{1}{f_2(\alpha_1)-x_2} -\frac{1}{f_2(\alpha_1)-x_1}&\cdots & \frac{1}{f_2(\alpha_1)-x_n} - \frac{1}{f_2(\alpha_1)-x_1}\\
            \vdots &\vdots &\cdots &\vdots\\
             1 & 0&\dots & 0 \\
             \vdots &\vdots &\cdots &\vdots\\
            \frac{1}{f_2(\alpha_n)-x_1} &\frac{1}{f_2(\alpha_n)-x_2}-\frac{1}{f_2(\alpha_n)-x_1}&\cdots & \frac{1}{f_2(\alpha_n)-x_n}-\frac{1}{f_2(\alpha_n)-x_1}\\
       \end{bmatrix} \right\rvert\\
       &= (-1)^{m+2}\left\lvert 
       \begin{bmatrix}
          \frac{1}{f_2(\alpha_1)-x_2} -\frac{1}{f_2(\alpha_1)-x_1}&\cdots & \frac{1}{f_2(\alpha_1)-x_n} - \frac{1}{f_2(\alpha_1)-x_1}\\
            \vdots &\cdots &\vdots\\
              \frac{1}{f_2(\alpha_{m})-x_2} -\frac{1}{f_2(\alpha_{m})-x_1}&\cdots & \frac{1}{f_2(\alpha_{m})-x_n} - \frac{1}{f_2(\alpha_{m})-x_1}\\
               \frac{1}{f_2(\alpha_{m+2})-x_2} -\frac{1}{f_2(\alpha_{m+2})-x_1}&\cdots & \frac{1}{f_2(\alpha_{m+2})-x_n} - \frac{1}{f_2(\alpha_{m+2})-x_1}\\
             \vdots &\vdots &\cdots \\
           \frac{1}{f_2(\alpha_n)-x_2}-\frac{1}{f_2(\alpha_n)-x_1}&\cdots & \frac{1}{f_2(\alpha_n)-x_n}-\frac{1}{f_2(\alpha_n)-x_1}\\
       \end{bmatrix} \right\rvert\\
\end{align*}
where we use co-factor expansion of the determinant along the $m+1$th row. 

For $i \ne (m+1)$ and $j \ge 2$,
\begin{align*}
    \frac{1}{f_2(\alpha_i)-x_j} -\frac{1}{f_2(\alpha_i)-x_1} = \frac{x_j - x_1}{(f_2(\alpha_i)-x_j)(f_2(\alpha_i)-x_1)}
\end{align*}
Therefore,
\begin{align}\label{lem_TP:det_M2'}
    |\mathbf{M}_2'| = (-1)^{m+2} \frac{\prod_{j=2}^n (x_j - x_1)}{\prod_{i\ne (m+1)} f_2(\alpha_i)-x_1}\left\lvert 
       \begin{bmatrix}
          \frac{1}{f_2(\alpha_1)-x_2} &\cdots & \frac{1}{f_2(\alpha_1)-x_n} \\
            \vdots &\cdots &\vdots\\
              \frac{1}{f_2(\alpha_{m})-x_2} &\cdots & \frac{1}{f_2(\alpha_{m})-x_n} \\
               \frac{1}{f_2(\alpha_{m+2})-x_2} &\cdots & \frac{1}{f_2(\alpha_{m+2})-x_n}  \\
             \vdots &\vdots &\cdots \\
           \frac{1}{f_2(\alpha_n)-x_2}&\cdots & \frac{1}{f_2(\alpha_n)-x_n} \\
       \end{bmatrix} \right\rvert
\end{align}
In particular, since $x_j > x_1$ by set-up, and for any $\alpha$, $\textrm{sgn}(f_2(\alpha)-x_1 )=\textrm{sgn}(\alpha) $ by \eqref{lem_TP:eq_1}, we have
\begin{align*}
    \textrm{sgn}(\frac{\prod_{j=2}^n (x_j - x_1)}{\prod_{i\ne (m+1)} f_2(\alpha_i)-x_1}) = (-1)^m
\end{align*}
The determinant in \eqref{lem_TP:det_M2'} is exactly Cauchy determinant of the matrix formed by $(\alpha_1,\dots,\alpha_m,\alpha_{m+2},\dots,\alpha_n)$ and $(x_2,\dots,x_n)$. Therefore, by \eqref{lem_TP:eq_2},
\begin{align*}
    \textrm{sgn}(|M_2'|) = (-1)^{m+2}(-1)^m (-1)^{2(n-1)m - m^2}
\end{align*}
Then $\textrm{sgn}(M)=\textrm{sgn}(|M_1'|)\textrm{sgn}(|M_2'|) = (-1)^{m+(m+2)+m + 2(n-1)m - m^2} = (-1)^{2(m+1)+m(m-1) + 2(n-1)m}=1$
\end{proof}


\subsubsection{Proof of Corollary \ref{cor:tp_K}}
\begin{proof}
We note $K(\alpha,\omega)=\tilde{K}(\alpha,\cos\omega)$. Now, let $-1<\alpha_1<\alpha_2<...<\alpha_n<1$ and $-\pi\leq \omega_1< \omega_2<...<\omega_n\leq 0$ given. Since $\cos(\omega)$ is a strictly increasing function of $\omega$ on $[-\pi,0]$, $-1\leq \cos(\omega_1)< \cos(\omega_2)<...<\cos(\omega_n)\leq 1$. Then \begin{align*}
    \bigg|[K(\alpha_i,\omega_j)]_{i,j=1}^{n}\bigg|=\bigg|[\tilde{K}(\alpha_i,\cos(\omega_j))]_{i,j=1}^{n}\bigg|>0
\end{align*} from Lemma~\ref{lem:tp_K}. Since $n$, $\alpha_i$, $i=1,...,n$, and $\omega_j,j=1,...,n$ were arbitrary, $K(\cdot,\cdot)$ is strictly totally positive.

\end{proof}

\subsection{Proofs on discrete finite support of representing measure for $\Pi^{\phi}(r;C)$}\label{supp_sec: 3_3_finiteSupport}

For a function $f:A\to\mathbb{R}$, we define, following~\citet{karlin1968total},
\begin{align*}
S^{-}(f)=\sup\,S^{-}(f(t_1),f(t_2),...,f(t_m))
\end{align*} where the supremum is taken over sequences $t_1<t_2<...<t_m$, $t_i\in A$, $i=1,...,m$ with arbitrary finite length $m$, and $S^{-}(x_1,...,x_m)$ denotes the number of sign changes of the sequence $x_1,x_2,...,x_m$, with any zero terms discarded.

We also define \begin{align*}
    S^{+}(f)=\sup S^{+}(f(t_1),f(t_2),...,f(t_m))
\end{align*}
 where the supremum is taken over sequences $t_1<t_2<...<t_m$, $t_i\in A$, $i=1,...,m$ with arbitrary finite length $m$,  and $S^{+}(x_1,...,x_m)$ denotes the maximum number of sign changes of $x_1,...,x_m$, with the zero terms being permitted to take on arbitrary signs.

We first present the following two Lemmas, which bound the number of roots of a continuous function $f$, based on $S^{-}(f)$ and $S^{+}(f)$.
\begin{lem}\label{supp_lem: S-f}
Let $f:A \to \R$ be a continuous function for a non-empty interval $A \subseteq \R$. Suppose $f$ has exactly $n$ roots in $A$, for some $n\in\mathbb{N}$. Then we have $S^{-}(f) \le n$.
\end{lem}
\begin{proof}
First, consider the trivial case where the length of $A$ is zero, i.e., in the case $A=\{t\}$ for some $t \in \R$. Then, $S^{-}(f) = S^{-}(f(t)) = 0$. Therefore the result holds.
Now suppose $A$ has nonzero length. Let $m\in \mathbb{N}$ be given, and choose $t_1<t_2<...<t_{m}$ with $t_i\in A$, $i=1,...,m$.  Without loss of generality, assume $f(t_i)\neq 0$, $i=1,...,m$ (as zero terms are discarded in the definition of $S^{-}$). For any $(t_{i},t_{i+1})$ with $1 \le i\leq m-1$ satisfying $\textrm{sgn}{(f(t_i))}\neq \textrm{sgn}(f(t_{i+1}))$, we have from the intermediate value theorem that there exists some $t_{i}^{*}$ with $t_i<t_{i}^*<t_{i+1}$ such that $f(t_i^{*})=0$. Therefore $S^{-}(f(t_1),f(t_2),...,f(t_m))\leq n$.
\end{proof}

\begin{lem}\label{supp_lem: S+f}
Let $f:A \to \R$ be a continuous function for a non-empty interval $A \subseteq \R$ such that the length of $A >0$ . Suppose there exist $n$ roots of $f$ in $A$, for some $n \in \mathbb{N}.$ We have $n \le S^{+}(f)$.    
\end{lem}
\begin{proof}

In the case $n=1$, let $t_1\in A$ satisfy $f(t_1)=0$, and choose $t_1^*\in A$ such that $t_1^*\neq t_1$. In case $t_1^*>t_1$, then $S^{+}(f)\geq S^{+}(f(t_1),f(t_1^*))=1=n$. In the case $t_1^*<t_1$, then $S^{+}(f)\geq S^{+}(f(t_1^*),f(t_1))=1=n$. This proves the result in the case $n=1$.

In the case $n>1$, let $t_1<t_2<...<t_n$ denote  roots of $f$ in $A$, i.e., $t_i\in A$, $i=1,...,n$ and $f(t_i) = 0$ for $i=1,\dots,n$. Choose $t_1^* \in (t_1,t_2)$ so that 
$t_1<t_1^*<t_2<\dots<t_n$. Consider $S^{+}(f(t_1),f(t_1^*),f(t_2),\dots,f(t_n))$. In the case $f(t_1^*)=0$, assign the sign $+$ to $f(t_1^*)$, and $-$ to each of  $f(t_1)=0$ and $f(t_2)=0$. In case $f(t_1^*)\neq 0$, assign to each of $f(t_1)$ and $f(t_2)$ the opposite sign from $f(t_1^*)$. Choose the sign of $f(t_i)$ for any remaining $i\ge 3$ so that the sign of $f(t_i)$ is opposite to the sign of $f(t_{i-1})$. Then $S^{+}(f(t_1),f(t_1^*),f(t_2),f(t_3),...,f(t_n))= n$. Therefore, $S^{+}(f) \ge n$.

\end{proof}

We now present the proof of Proposition~\ref{prop:finiteSupport}:

\begin{proof}
    Let $\hat{f}:\T \to \mathbb{R}$ defined by $\hat{f}(\omega)=\widehat{\Pi^\phi(r;C)}(\omega)=\sum_{k\in\mathbb{Z}}\Pi^\phi(r;C)\exp(-i\omega k)$ denote the Fourier transform of $\Pi^\phi(r;C)$. Let $\mu_C^{\phi}$ be the representing measure of $\Pi^\phi(r;C)$. 
    
    Note by Lemma~\ref{lem: fourier_xa} and Lemma~\ref{lem: l1projection}, $\hat{f}(\omega)=\int K(\alpha,\omega)\,\mu_C^{\phi}(d\alpha)$ for $\omega\in[-\pi,\pi]$ and $\hat{f}$ is continuous on $[-\pi,\pi]$. For $\alpha\in(-1,1)$, define 
    \begin{align*}
    g(\alpha)
    &=\braket{x_{\alpha},r-\Pi^\phi(r;C)}_{\phi}=(2\pi)^{-1}\int_{[-\pi,\pi]}K(\alpha,\omega)\frac{\{\hat{r}(\omega)-\hat{f}(\omega)\}}{\phi(\omega)^2}\,d\omega\\
    &=\pi^{-1}\int_{[-\pi,0]}K(\alpha,\omega)\frac{\{\hat{r}(\omega)-\hat{f}(\omega)\}}{\phi(\omega)^2}\,d\omega,
    \end{align*}
    where the last equality follows since $K$, $\hat{r}$, $\hat{f}$, and $\phi$ are even functions of $\omega$.
    Note that by Proposition \ref{prop:weighted_opt_eq}, if $\alpha \in \operatorname{Supp}(\hat{\mu}_C)$, then $g(\alpha)=0$.

    We consider separately the cases $\hat{r}(\omega)-\hat{f}(\omega)\equiv 0$ and $\hat{r}(\omega)-\hat{f}(\omega)\not\equiv 0 $.

    \paragraph{Case 1: $\hat{r}(\omega)-\hat{f}(\omega)\equiv 0$.} In this case, we show $|\operatorname{Supp}(\mu_C^{\phi})|\in\{0,1\}$. Since $\hat{r}(\omega)=\hat{f}(\omega)$ for each $\omega$, we have $r(k)=\Pi^\phi(r;C)(k)$ for each $k\in\mathbb{Z}$. Since $r(k)=0$ for $|k|>M-1$, we have $\Pi^\phi(r;C)(k)=0$ for each $k>M-1$. Thus, $\hat{\mu}_{C}=c\delta_0$ for some $c>0$, so that $|\operatorname{Supp}(\mu_C^{\phi})|=|\{0\}|=1$, or else $\mu_C^{\phi}$ is the null measure, so that $|\operatorname{Supp}(\mu_C^{\phi})|=0$. In each case, $-1,1\notin \operatorname{Supp}(\mu_C^{\phi})$. Further, recalling the definition of $n$ as the smallest even number such that $n>(M-1)$, we have $|\operatorname{Supp}(\mu_C^{\phi})|\leq 1\leq \frac{n}{2}\leq \frac{n}{2}+1\leq n$, where we used $M\geq 1$ implies $n\geq 2$. This proves the result in Case 1.
    
    \paragraph{Case 2: $\hat{r}(\omega)-\hat{f}(\omega)\not\equiv 0$.} First, we show $g(\alpha)=0$ for at most $n$ points $\alpha_1,...,\alpha_n\in (-1,1)$. We proceed by showing $h_0(\omega):=\frac{\hat{r}(\omega)-\hat{f}(\omega)}{\phi(\omega)^2}$ is continuous with at most $n$ zeroes in $[-\pi,0]$, and thus, $S^{-}(h_0)\leq n$ by Lemma \ref{supp_lem: S-f}. Since $K(\alpha,\omega)$ is a strictly totally positive kernel on $(-1,1)\times [-\pi,0]$, the variation diminishing property of strictly totally positive kernels implies $S^{+}(g)\leq  S^{-}(h_0)$. By Lemma \ref{supp_lem: S+f}, the number of roots of the function $g$ is less than equal to $S^{+}(g)$. Thus, combining the two bounds, we have $($the number of roots of $g) \le S^{+}(g) \le S^{-}(h_0)\le n$.
    
Now we show that the number of roots of $h_0$ is less than or equal to $n$.   
We have that $h_0$ is a continuous function of $\omega$, since the weight $\phi(\omega)$ is continuous and bounded away from 0 by assumption, $\hat{r}(\omega) = \sum_{k\in \Z} r(k) \cos(\omega k)$ is a finite degree cosine polynomial and thus continuous, and $f(\omega)$ is continuous from Lemma~\ref{lem: fourier_xa} and \ref{lem: l1projection}. Now, the number of sign changes of $h_0$ and $h_1(\omega):=\hat{r}(\omega)-\hat{f}(\omega)$ are identical, that is, $S^{-}(h_0)=S^{-}(h_1)$, since $\phi(\omega)^2>0$, $\forall \omega$. Additionally, define $\tilde{h}:[-1,1]\to \mathbb{R}$ by $\tilde{h}(x)=\sum_{k=-(M-1)}^{M-1}r(k)T_{|k|}(x)-\int_{[-1,1]}\frac{1-\alpha^2}{1-2\alpha x+\alpha^2}\,\hat{\mu}_{C}(d\alpha)$, where $T_{k}(x)$ denotes the $k$th Chebyshev cosine polynomial, $k=0,1,2,...$, and we recall $T_k(\cdot)$ is a polynomial uniquely defined by $T_k(\cos\omega)=\cos(k\omega)$ for $\omega\in\mathbb{R}$. Then $\tilde{h}(\cos\omega)=h_1(\omega)$ for $\omega\in[-\pi,0]$. Since $\omega\mapsto\cos(\omega)$ is a strictly increasing and surjective map from $[-\pi,0]\to [-1,1]$, we have $S^{-}(\tilde{h})=S^{-}(h_1)=S^{-}(h_0)$.

Now, we show $\tilde{h}$ has at most $n$ roots in $(-1,1)$. Let $g^{(m)}(x)=\frac{d^{m}g(x)}{dx^{m}}$ denote the $m$th derivative of $g$. Then for $x\in(-1,1)$, we have 
\begin{align}
    \frac{d^m}{dx^m}\tilde{h}^{(m)}(x)
    &=\sum_{k=-(M-1)}^{(M-1)}r(k)\frac{d^m}{dx^m}T_{|k|}(x)-\frac{d^m}{dx^m} \int_{[-1,1]} \frac{1-\alpha^2}{1-2\alpha x +\alpha^2}\hat{\mu}_{C}(d\alpha) \nonumber\\
    &=\sum_{k=-(M-1)}^{(M-1)}r(k)\frac{d^m}{dx^m}T_{|k|}(x)-\int_{[-1,1]}\frac{m!(1-\alpha^2)(2\alpha)^{m}}{(1-2\alpha x+\alpha^2)^{m+1}}\,\hat{\mu}_{C}(d\alpha)\label{eq:hDerivative}
\end{align} where the integrand in the second term of~\eqref{eq:hDerivative} results from differentiation under the integral sign, and the interchange of integration and differentiation can be justified similarly as in the proof of Proposition 6 in~\citet{berg2023efficient}. First, we consider $m=1$, and show $\frac{d}{dx} \int_{[-1,1]} \frac{1-\alpha^2}{1-2\alpha x +\alpha^2}\hat{\mu}_{C}(d\alpha)= \int_{[-1,1]} \frac{d}{dx}\left\{\frac{1-\alpha^2}{1-2\alpha x +\alpha^2}\right\}\hat{\mu}_{C}(d\alpha)$ for each $x\in (-1,1)$. Let $x_0\in(-1,1)$ given. We note that for $x\in(-1,1)$, minimization over $\alpha$ shows $(1-2\alpha x+\alpha^2)\geq 1-x^2>0$. Now, define $\beta=(1-|x_0|)/2$ and take $\tilde{x}=1-\beta$. Let $\tilde{g}(\alpha,x)=\frac{(1-\alpha^2)(2\alpha)}{(1-2\alpha x+\alpha^2)^2}$. Define $\bar{g}(\alpha)=\frac{(1-\alpha^2)(2\alpha)}{(1-\tilde{x}^2)^2}$. Let $N_{\beta}(x_0)=\{y:|y-x_0|<\beta\}$. Then for $\alpha\in (-1,1)$, $y\in N_{\beta}(x_0)$, $|\tilde{g}(\alpha,y)|\leq |\bar{g}(\alpha)|$. Since $|\bar{g}(\alpha)|$ is bounded for $\alpha\in (-1,1)$ and $\hat{\mu}_{C}((-1,1))$ is finite, we have $\int_{[-1,1]}|\bar{g}(\alpha)|\,\hat{\mu}_{C}(d\alpha)<\infty$. Thus $\frac{d}{dx} \int_{[-1,1]} \frac{1-\alpha^2}{1-2\alpha x +\alpha^2}\hat{\mu}_{C}(d\alpha)= \int_{[-1,1]} \frac{d}{dx}\left\{\frac{1-\alpha^2}{1-2\alpha x +\alpha^2}\right\}\hat{\mu}_{C}(d\alpha)$ for each $x\in (-1,1)$. Proceeding by induction from $m=1$, a similar approach yields $\frac{d^m}{dx^m}\left\{ \int_{[-1,1]} \frac{1-\alpha^2}{1-2\alpha x +\alpha^2}\hat{\mu}_{C}(d\alpha)\right\}= \int_{[-1,1]} \frac{d^m}{dx^m}\left\{\frac{1-\alpha^2}{1-2\alpha x +\alpha^2}\right\}\hat{\mu}_{C}(d\alpha)$
    
Now, for $m>(M-1)$, the summation in~\eqref{eq:hDerivative} vanishes, since the $T_{|k|}$ are degree $|k|$ polynomials and $\frac{d^m}{dx^m}T_{|k|}(x) = 0$. Further, since $n$ is the smallest even number such that $n>(M-1)$, we have $n=2k>(M-1)$ for some $k\in \mathbb{N}$, and 
\begin{align*}
    \tilde{h}^{(n)}(x)=-\int_{[-1,1]}\frac{n!(1-\alpha^2)(2\alpha)^{2k}}{(1-2\alpha x+\alpha^2)^{n+1}}\,\hat{\mu}_{C}(d\alpha) \leq 0,
\end{align*}
since $\mu_C^{\phi}$ is nonnegative, and the integrand is nonnegative for each $\alpha$, and strictly positive for $\alpha\neq 0$. In the case $\tilde{h}^{(n)}(x)=0$, we have 
$\operatorname{Supp}(\mu_C^{\phi})=\{0\}$ or else $\mu_C^{\phi}$ is the null measure with $\operatorname{Supp}(\mu_C^{\phi})=\emptyset$, and in both cases the Lemma holds. Otherwise, $\tilde{h}^{(n)}(x)<0$ for all $x\in(-1,1)$, and so  $\tilde{h}$ has at most $n$ roots in $(-1,1)$. Thus $S^{-}(h_0)=S^{-}(\tilde{h})\leq n$. Now, since $K(\alpha,\omega)$ is totally positive on $(-1,1)\times (-\pi,0)$ (ref. Corollary \ref{cor:tp_K}), we have from Theorem 3.1b) of~\citet{karlin1968total} that $g(\alpha)=\pi^{-1}\int_{(-\pi,0)}K(\alpha,\omega)h_0(\omega)\,d\omega$ satisfies $S^{+}(g)\leq S^{-}(h_0)\leq n$.

Since the number of roots of $g$ in $(-1,1)$ is smaller than or equal to $S^{+}(g)$ by Lemma \ref{supp_lem: S+f} and $S^{+}(g)\le n$ , we have from Proposition~\ref{prop:weighted_opt_eq} that $|\text{Supp}(\hat{\mu}_{C})\cap (-1,1)|\leq n$. Additionally, $\hat{\mu}_{C}(\{-1,1\})=0$ since $\|\Pi^\phi(r;C)\|_{\phi}<\infty$, $\|\cdot\|_{\phi}$ and $\|\cdot\|_2$ are equivalent norms, and from Lemma 2 of~\citet{berg2023efficient}, $\|\Pi^{\phi}(r;C)\|_2<\infty$ implies $\hat{\mu}_{C}(\{-1,1\})=0$. Therefore $-1,1\notin \text{Supp}(\hat{\mu}_{C})$, so $\text{Supp}(\hat{\mu}_{C})\subset (-1,1)$.

We now show the tighter bounds on $|\text{Supp}(\hat{\mu}_C)|$ in the case $C=[L,U]$ for $-1\leq L\leq U\leq 1$. We note $g(\alpha)\leq 0$ for each $\alpha\in C$ from Proposition~\ref{prop:weighted_opt_ineq}. 

\paragraph{Case 2a: $-1<L\leq U<1$} We have from Proposition~\ref{prop:weighted_opt_eq} and the fact $\text{Supp}(\hat{\mu}_{C})\subset (-1,1)$ that $g(\alpha)=0$ for each $\alpha\in \text{Supp}(\hat{\mu}_{C})$. In the case that $\hat{\mu}_{C}$ is the null measure or $\text{Supp}(\hat{\mu}_{C})$ contains only a single element, then the result holds. Otherwise, let $\alpha_1<\alpha_2<...<\alpha_m$, $\alpha_i\in \text{Supp}(\hat{\mu}_{C})$, $i=1,...,m$ given. Define $\alpha_i^*=(\alpha_i+\alpha_{i+1})/2$, $i=1,...,m-1$. Then $n\geq S^{+}(g)\geq S^{+}(g(\alpha_1),g(\alpha_1^*),g(\alpha_2),g(\alpha_2^*),...,g(\alpha_{m-1}^*),g(\alpha_m))=2(m-1)$. Thus $|\text{Supp}(\hat{\mu}_{C})|\leq \frac{n}{2}+1$. 

    \paragraph{Case 2b: at least one of $L=-1$ or $U=1$ holds} We note that in the case $L=U=-1$ or $L=U=1$, then $\hat{\mu}_{C}$ is the null measure since $-1,1\notin \text{Supp}(\hat{\mu}_{C})$. 

    Next, we consider the case $C=[-1,1]$. In the case $\hat{\mu}_C$ is the null measure or $\text{Supp}(\hat{\mu}_{C})$ contains only a single element, then the result holds. Otherwise, let $\alpha_1<\alpha_2<...<\alpha_m$, $\alpha_i\in \text{Supp}(\hat{\mu}_{C})$, $i=1,...,m$ given. Choose $\alpha_0^*$, $\alpha_m^*$ such that $-1<\alpha_0^*<\alpha_1$ and $\alpha_m<\alpha_m^*<1$ and define $\alpha_i^*=(\alpha_i+\alpha_{i+1})/2$, $i=1,...,m-1$.. Then 
    \begin{align*}n\geq S^{+}(g)\geq S^{+}(g(\alpha_0^*),g(\alpha_1),g(\alpha_1^*),g(\alpha_2),g(\alpha_2^*),...,g(\alpha_{m-1}^*),g(\alpha_m),g(\alpha_m^*))=2m.\end{align*} Thus $\text{Supp}(\hat{\mu}_C)\leq \frac{n}{2}$.

    Now we consider the case $-1=L<U<1$. In the case $\hat{\mu}_C$ is the null measure or $\text{Supp}(\hat{\mu}_{C})$ contains only a single element, then the result holds. Otherwise, let $\alpha_1<\alpha_2<...<\alpha_m$, $\alpha_i\in \text{Supp}(\hat{\mu}_{C})$, $i=1,...,m$ given. Choose $-1<\alpha_0^*<\alpha_1$, and define $\alpha_i^*=(\alpha_i+\alpha_{i+1})/2$, $i=1,...,m-1$. Then $$n\geq S^{+}(g)\geq S^{+}(g(\alpha_0^*),g(\alpha_1),g(\alpha_1^*),g(\alpha_2),g(\alpha_2^*),...,g(\alpha_{m-1}^*),g(\alpha_m))=2m-1.$$ Thus $\text{Supp}(\hat{\mu}_{C})\leq \lfloor \frac{n+1}{2}\rfloor =\frac{n}{2}$, where we use the fact that $n$ is even. The case $-1<L<U=1$ is similar.

\end{proof}

\section{Proofs for Section \ref{sec: 4_statistical_analysis}}
\label{supp_sec: 4_statistical_analysis}

\subsection{Proof of Proposition~\ref{prop:prop_c}} \label{supp_sec: 4_1_prop_c}
\begin{proof}
We first show that for any given $\alpha \in \mathcal{K}$, we have
\begin{equation}\label{prop_c:conv_pw}
   \lim_{M\to\infty} |\langle  x_\alpha, r_M - \gamma\rangle_{\phi_M} | = 0
\end{equation}
$P_x$-almost surely for any $x \in \mathsf{X}$.  Let the initial condition $x\in \mathsf{X}$ be given.
Let $\hat{r}_M$ and $\hat{\gamma}$ be the Fourier transform $r_M$ and $\gamma$, respectively. Also, let $\hat{g}_{M\alpha}(\omega) = K(\alpha,\omega) \phi_M^{-2}(\omega)$. 

We first note for each $M$, $\hat{g}_{M\alpha} \in C^2(\T)$ is a twice continuously differentiable function for each $M$, as $K(\alpha,\omega)$ and $\phi_M(\omega)$ are twice continuously differentiable and $\phi_M(\omega) \ge c_{0M}>0$ by Assumption \ref{cond:phi_M}. In particular, $\hat{g}_{M\alpha} \in L^2(\T)$. 
We define the inverse Fourier transform $g_{M\alpha}$ of $\hat{g}_{M\alpha}$ by $g_{M\alpha}(k) = ( \hat{g}_{M\alpha}, u_k ) =  \frac{1}{2\pi} \int_{[-\pi,\pi]} \hat{g}_{M\alpha} (\omega) e^{i\omega k} d\omega$ for each $k\in \Z$ where we recall the definition that $u_k(\omega)=e^{-i\omega k}$. By Lemma \ref{lem:partialSumL2limit} and \ref{lem:FourierCoef}, $g_{M\alpha}$ is square summable and the power series $\sum_{k\in\Z}  ( \hat{g}_{M\alpha}, u_k ) u_k(\omega) = \sum_{k\in\mathbb{Z}}g_{M\alpha}(k)e^{-i\omega k}$ is uniformly convergent for $\omega \in [-\pi,\pi]$.

We have
\begin{align*}
    \langle x_\alpha, r_M - \gamma\rangle_{\phi_M} = \frac{1}{2\pi}\int \hat{g}_{M\alpha}(\omega) (\hat{r}_M - \hat{\gamma})(\omega) d\omega = \sum_{k\in\Z} g_{M\alpha}(k) (r_M(k) - \gamma(k))
\end{align*}
where the last equality is due to Parseval's identity (ref. eq \eqref{eq:parseval-2}). Note for any $B>0$,
\begin{align}
|\langle  x_\alpha, r_M - r\rangle_{\phi_M} | 
    &=| \sum_{k\le \Z} g_{M\alpha}(k) (r_M(k) - \gamma(k))| \nonumber\\
    &\le \sum_{|k|< B} |g_{M\alpha}(k) (r_M(k) - \gamma(k)) | + \sum_{|k|\ge B} |g_{M\alpha}(k)| |r_M(k) - \gamma(k)|\nonumber\\
    &\le \sum_{|k|< B}  |g_{M\alpha}(k) (r_M(k) - \gamma(k)) | + (r_M(0) + \gamma(0))\sum_{|k|\ge B} |g_{M\alpha}(k)| \label{prop_c:eq_1}
\end{align}
where we use the fact that $|\gamma(k)| \le \gamma(0)$ and $|r_M(k)| \le r_M(0)$ $P_x$-a.s. for any $k\in \Z$.
By Lemma \ref{lem:FourierCoef}, we have 
\begin{align*}
    k^2 g_{M\alpha}(k) &= -\frac{1}{2\pi } \int_{[-\pi,\pi]}\frac{d^2\{\hat{g}_{M\alpha} (\omega) e^{i\omega k}\}}{d\omega^2}\, d\omega.
\end{align*}
In particular, for any $k \in \Z$,
\begin{align}\label{prop_c:gmk_bound}
    | g_{M\alpha}(k)| & \le \frac{1}{k^2}\left( \sup_{\omega \in [ -\pi, \pi]} |\frac{d^2}{d\omega^2}\hat{g}_{M\alpha} (\omega) | \right)\frac{1}{2\pi }  \int_{[-\pi,\pi]}|e^{i\omega k}| d\omega = \frac{C_{M\alpha}}{k^2}
\end{align}
where we define $C_{M\alpha}=  \sup_{\omega \in [ -\pi, \pi]} |\frac{d^2}{d\omega^2}\hat{g}_{M\alpha} (\omega)|$. Suppose for now that the limit of $ C_{M\alpha }$ is $P_x$-almost surely bounded by a constant $C_\alpha<\infty$, i.e., $\limsup_{M\to\infty}  C_{M\alpha } \le C_\alpha$, $P_x$-almost surely. From \eqref{prop_c:eq_1} and \eqref{prop_c:gmk_bound}, we have
\begin{align}\label{prop_c:eq_2}
    |\langle  x_\alpha, r_M - r\rangle_{\phi_M} |   \le \sum_{|k|< B} |g_{M\alpha}(k) (r_M(k) - \gamma(k))| + (r_M(0) + \gamma(0))C_{M\alpha}\sum_{|k|\ge B}  \frac{1}{k^2}
\end{align}
Let $\epsilon>0$ be given. Choose $B \ge  4\gamma(0)C_\alpha/\epsilon + 2$ so that
\begin{align}\label{prop_c:B_choice}
    2\gamma(0)C_\alpha \sum_{|k|\ge B } \frac{1}{k^2} \le 4\gamma(0)C_\alpha \int_{B-1}^\infty \frac{1}{x^2} dx  = \frac{4\gamma(0)C_\alpha}{B-1} \le \epsilon.
\end{align}
Then
\begin{align*}
    &\limsup_{M\to\infty}  |\langle  x_\alpha, r_M - r\rangle_{\phi_M} |\\
    &\le \limsup_{M\to\infty}  \sum_{|k|< B} |g_{M\alpha}(k) (r_M(k) - \gamma(k))| + \limsup_{M\to\infty}  \{(r_M(0) + \gamma(0))C_{M\alpha}\}\sum_{|k|\ge B}  \frac{1}{k^2}\\
    &\le 2\gamma(0) C_{\alpha}\sum_{|k|\ge B}  \frac{1}{k^2} \le \epsilon
\end{align*}
$P_x$-almost surely, where for the second last inequality we use condition \ref{cond:rM_pwConvergence}, and $B$ is a finite number which only depends on $\gamma(0), C_\alpha$, and $\epsilon$. For the last inequality we use \eqref{prop_c:B_choice}. Since $\epsilon$ is arbitrary, we have $|\langle x_\alpha, r_M - \gamma \rangle_{\phi_M}| \to 0$ $P_x$-almost surely assuming $\limsup_{M\to\infty}  C_{M\alpha} \le C_\alpha $.

To establish the pointwise convergence, it remains to show $\limsup_{M\to\infty}  C_{M\alpha} \le C_\alpha $ for a constant $C_\alpha< \infty$, $P_x$-almost surely. From the definition of $\hat{g}_{M\alpha}$ and product rule, we have
\begin{align*}
    &\frac{d^2}{d\omega^2}\hat{g}_{M\alpha} (\omega) = \frac{d^2}{d\omega^2} \left( \frac{K(\alpha,\omega)}{\phi_M^2(\omega)}\right) \\
    &= \frac{1}{\phi_M^2(\omega)}\frac{d^2K(\alpha,\omega)}{d\omega^2}  +2\frac{dK(\alpha,\omega)}{d\omega}\frac{d}{d\omega}\left\{ \frac{1}{\phi_M^2(\omega)}\right\}+ K(\alpha,\omega)\frac{d^2}{d\omega^2}\left\{\frac{1}{\phi_M^2(\omega)}\right\}\\
    & = \frac{1}{\phi_M^2(\omega)}\frac{d^2K(\alpha,\omega)}{d\omega^2}  -\frac{4}{\phi_M^3(\omega)}\frac{dK(\alpha,\omega)}{d\omega}\frac{d\phi_M(\omega)}{d\omega} + K(\alpha,\omega)\left\{\frac{6}{\phi_M^4(\omega)}\frac{d\phi_M(\omega)}{d\omega} -\frac{2}{\phi_M^3(\omega)}\frac{d^2\phi_M(\omega)}{d\omega^2} \right\}
\end{align*}
In particular,
\begin{align*}
    C_{M\alpha}= &\sup_{\omega\in[-\pi,\pi]} |\frac{d^2}{d\omega^2}\hat{g}_{M\alpha} (\omega)|\\
    &\le \frac{10}{(1-|\alpha|)^{6}} \left\{\frac{1}{\inf_{\omega\in[-\pi,\pi]} \phi^2_M(\omega)} + \frac{4\|\phi'\|_\infty}{\inf_{\omega\in[-\pi,\pi]} |\phi^3_M(\omega)|} +  \frac{6\|\phi'\|_\infty}{\inf_{\omega\in[-\pi,\pi]} \phi_M^4(\omega)} + \frac{2\|\phi_M''\|_\infty}{\inf_{\omega\in[-\pi,\pi]} |\phi_M^3(\omega)|} \right\}
\end{align*}
where we use Lemma \ref{lem: K_bound}. Then, using assumptions on $\phi_M$,
\begin{align}\label{prop_c:C_Ma_bound}
\limsup_{M\to\infty} C_{M\alpha}\le  \frac{10}{(1-|\alpha|)^{6}}\left\{\frac{1}{c_\phi^2} + \frac{4 c'_\phi}{c_\phi^3} +   \frac{6c'_\phi}{c_\phi^4} + \frac{2c'_\phi}{c_\phi^3} \right\}
\end{align}
and take $C_\alpha =  \frac{10}{(1-|\alpha|)^{6}}\{\frac{1}{c_\phi^2} + \frac{4 c'_\phi}{c_\phi^3} +   \frac{6c'_\phi}{c_\phi^4} + \frac{2c'_\phi}{c_\phi^3} \} <\infty$.

Now, we show that the convergence is uniform over $\mathcal{K}$. Similar to the proof of Proposition 8 in \citet{berg2023efficient}, we obtain a finite subcover of a compact set $\mathcal{K}$ with sufficiently small radius and use the pointwise convergence in \eqref{prop_c:conv_pw} to establish a uniform convergence. Let $\epsilon_1>0$ be given. 
Let $d_{\mathcal{K}}$ be the smallest distance between a point $x\in\mathcal{K}$ and $\{-1,1\}$, i.e., $d_{\mathcal{K}} =\inf_{x\in \mathcal{K}} \min\{|x-1|,|x+1|\}$. Since $\mathcal{K} \subset(-1,1)$,  $0<d_\mathcal{K}\le 1$. If $d_{\mathcal{K}} = 1$ then $\mathcal{K} = \{0\}$ and the result \eqref{prop_c:result} is true by \eqref{prop_c:conv_pw}. Let 
\begin{align*}
r_{\mathcal{K}} = \min\,\{(\epsilon_1d_\mathcal{K}^2)/\{8\tilde{c}_{\phi}b\},d_\mathcal{K}/2\}
\end{align*} where $\tilde{c}_{\phi}=\frac{1}{c_{\phi}^2}+2\left\{\frac{6c_{\phi}'}{c_{\phi}^4}+\frac{2c'_{\phi}}{c_{\phi}^3}\right\}$ and $b=\begin{cases}\gamma(0) & \gamma(0)\neq 0\\1 & \gamma(0)=0\end{cases}$. 
Define $N(\alpha; r_{\mathcal{K}}) = \{\beta \in \mathcal{K} ; |\beta-\alpha| < r_{\mathcal{K}}\}$.
Since $\mathcal{K}$ is compact, we can find $\alpha_1,\dots,\alpha_K$ with $K<\infty$ and $\alpha_j\in \mathcal{K}$, $j=1,...,K$, such that 
\begin{align*}
    \mathcal{K} \subseteq \cup_{j\in \{1,\dots,K\}} N(\alpha_j; r_{\mathcal{K}}).
\end{align*}
We have
\begin{align}
& \sup _{\alpha \in \mathcal{K}}|\langle g_{M\alpha}, r_M-\gamma\rangle| \nonumber\\
& \leq \max _{j\in \{1,\dots,K\}} \sup _{\alpha \in N(\alpha_j; r_{\mathcal{K}})}|\langle g_{M\alpha}, r_M-\gamma\rangle| \nonumber\\
& \leq  \max _{j\in \{1,\dots,K\}} \sup _{\alpha \in N(\alpha_j; r_{\mathcal{K}})} \{ |\langle g_{M\alpha}-g_{M\alpha_j}, r_M-\gamma\rangle|+|\langle g_{M\alpha_j}, r_M-\gamma\rangle| \}\nonumber\\
& \leq  \max _{j\in \{1,\dots,K\}} \sup _{\alpha \in N(\alpha_j; r_{\mathcal{K}})} |\langle g_{M\alpha}-g_{M\alpha_j}, r_M-\gamma\rangle|+\max _{j\in \{1,\dots,K\}} |\langle g_{M\alpha_j}, r_M-\gamma\rangle|.\label{prop_c:eq_3}
\end{align}
Recall that $g_{M\alpha}(k)$ is the inverse Fourier transform of $\hat{g}_{M\alpha}(\omega)= K(\alpha,\omega) \phi_M^{-2}(\omega)$. Let $h_M(k) = ( \phi_M^{-2}, u_k)$ be the inverse Fourier transform of $\phi_M(\omega)^{-2}$ so that $\phi_M^{-2}(\omega) = \sum_{k\in\Z} h_M(k) e^{-i\omega k}$ (note by Lemma~\ref{lem:FourierCoef}, this power series is uniformly convergent). 
We have $g_{M\alpha}(k) = (x_\alpha * h_M)(k)$ for $k \in \Z$, since
\begin{align*}
    g_{M\alpha} (k) 
    &= \frac{1}{2\pi} \int_{[-\pi,\pi]} K(\alpha,\omega) \phi_M^{-2}(\omega) e^{i\omega k }d \omega\\
    &=\frac{1}{2\pi} \int_{[-\pi,\pi]} \sum_{u\in \Z} x_\alpha(u) e^{-i\omega u} \phi_M^{-2}(\omega) e^{i\omega k }d \omega\\
    &= \sum_{u\in \Z} x_\alpha(u)  \frac{1}{2\pi} \int_{[-\pi,\pi]} \phi_M^{-2}(\omega) e^{i\omega (k-u) }d \omega\\
     &= \sum_{u\in \Z} x_\alpha(u)  h_M(k-u).
\end{align*}

Now we bound $\sum_{k\in\mathbb{Z}}|h_{M}(k)|$. We have \begin{align}
    \sum_{k\in\mathbb{Z}}|h_{M}(k)|&=\sum_{k\in\mathbb{Z}}|\frac{1}{2\pi}\int_{[-\pi,\pi]}\frac{1}{\phi_{M}^2(\omega)}e^{ik\omega}\,d\omega|\nonumber\\
    &=|\frac{1}{2\pi}\int_{[-\pi,\pi]}\frac{1}{\phi_{M}^2(\omega)}\,d\omega|+\sum_{k\in\mathbb{Z}\setminus{\{0\}}}\frac{1}{k^2}\frac{1}{2\pi}\bigg|\int_{[-\pi,\pi]}\frac{d^2}{d\omega^2}\left\{\frac{1}{\phi_{M}^2(\omega)}\right\}\,d\omega\bigg|\label{eq:l1start}
\end{align} where the second equality follows from Lemma~\ref{lem:FourierCoef}. Now \begin{align}
    \bigg|\frac{d^2}{d\omega^2}\left\{\frac{1}{\phi_{M}^2(\omega)}\right\}\bigg|&=\bigg|\frac{6}{\phi_{M}^4(\omega)}\frac{d\phi_{M}(\omega)}{d\omega}-\frac{2}{\phi_{M}^3(\omega)}\frac{d^2\phi_{M}(\omega)}{d\omega^2}\bigg|\nonumber
\end{align} and thus \begin{align}
    \underset{M\to\infty}{\lim\sup}\bigg|\frac{d^2}{d\omega^2}\left\{\frac{1}{\phi_{M}^2(\omega)}\right\}\bigg|\leq \frac{6c_{\phi}'}{c_{\phi}^4}+\frac{2c'_{\phi}}{c_{\phi}^3}\label{eq:phiBound}
\end{align} almost surely $P_x$, from the asymptotic bounds on $\phi_M$ in Assumption~\ref{cond:phi_M_asymp}. Combining~\eqref{eq:l1start} and~\eqref{eq:phiBound}, we obtain \begin{align*}
    \underset{M\to\infty}{\lim\sup}\sum_{k\in\mathbb{Z}}|h_{M}(k)|\leq \frac{1}{c_{\phi}^2}+2\left\{\frac{6c_{\phi}'}{c_{\phi}^4}+\frac{2c'_{\phi}}{c_{\phi}^3}\right\}=\tilde{c}_{\phi}
\end{align*} almost surely $P_x$, where the inequality uses $\sum_{k=1}^{\infty}\frac{1}{k^2}\leq 2$.

In addition, we bound $\|x_{\alpha}-x_{\beta}\|_1$ for $\alpha,\beta\in(-1,1)$. We have \begin{align*}
    \|x_{\alpha}-x_{\beta}\|_1&=\sum_{k\in\mathbb{Z}}|\alpha^{|k|}-\beta^{|k|}|\\
    &=2\sum_{k=1}^{\infty}|\alpha^{k}-\beta^{k}|\\
    &=2\sum_{k=1}^{\infty}|(\alpha-\beta)\sum_{j=1}^{k}\alpha^{k-j}\beta^{j-1}|\\
    &\leq 2|\alpha-\beta|\sum_{k=1}^{\infty}\sum_{j=1}^{k}|\alpha|^{k-j}|\beta^{j-1}|\\
    &=2|\alpha-\beta|\sum_{j=1}^{\infty}\sum_{k=j}^{\infty}|\alpha|^{k-j}|\beta^{j-1}|\\
    &=\frac{2|\alpha-\beta|}{(1-|\alpha|)(1-|\beta|)}.
\end{align*}

Thus, \begin{align*}
    |\langle x_{\alpha}-x_{\alpha_j}, r_M-\gamma\rangle_{\phi_{M}}|
    &=|\braket{(x_{\alpha}-x_{\alpha_j})*h_{M},r_M-\gamma}\\
    &\leq \|(x_{\alpha}-x_{\alpha_j})*h_{M}\|_{1}\,\|r_M-\gamma\|_{\infty}\\
    &\leq \|x_{\alpha}-x_{\alpha_j}\|_1\,\|h_{M}\|_1\, \|r_M-\gamma\|_{\infty}\\
    &\le \{r_M(0) + \gamma(0)\}\|h_M\|_1\frac{2|\alpha - \alpha_j|}{(1-|\alpha|)(1-|\alpha_j|)}
\end{align*}
where the first inequality uses Holder's inequality, and the second inequality uses Young's inequality for convolutions. The third inequality uses $|r_M(k)|\leq r_M(0)$ from Assumption~\ref{cond:rM:even} as well as $|\gamma(k)|\leq \gamma(0)$.
Now taking the $M$-limits in \eqref{prop_c:eq_3}, we have
\begin{align*}
    \limsup_{M\to\infty}\sup _{\alpha \in \mathcal{K}}| \langle g_{M\alpha}, r_M-\gamma\rangle| 
    &\le  \limsup_{M}  \max _{j\in \{1,\dots,K\}} \sup _{\alpha \in N(\alpha_j; r_{\mathcal{K}})} \{r_M(0) + \gamma(0)\}\|h_{M}\|_1\frac{2|\alpha - \alpha_j|}{(1-|\alpha|)(1-|\alpha_j|)}
\end{align*}
since $\limsup_{M}\max _{j\in \{1,\dots,K\}} |\langle g_{M\alpha_j}, r_M-\gamma\rangle| =0$ as $K<\infty$ from the pointwise convergence in \eqref{prop_c:conv_pw}. Since $1-|\alpha| \ge d_{\mathcal{K}}$ for any $\alpha \in \mathcal{K}$ and for $\alpha \in N(\alpha_j;r_{\mathcal{K}})$ we have $|\alpha-\alpha_j| < r_{\mathcal{K}}$,  we have
\begin{align*}
     \limsup_{M\to\infty}\sup _{\alpha \in \mathcal{K}}| \langle g_{M\alpha}, r_M-\gamma\rangle| 
     &\le
     \limsup_{M}  \max _{j\in \{1,\dots,K\}} \sup _{\alpha \in N(\alpha_j; r_{\mathcal{K}})}  \{r_M(0) + \gamma(0)\}\|h_{M}\|_1\frac{2r_{\mathcal{K}}}{(1-|\alpha|)(1-|\alpha_j|)}\\
    &\le \frac{4\gamma(0)}{\tilde{c}_{\phi}} \frac{(\epsilon_1d_\mathcal{K}^2\tilde{c}_\phi)/\{8b\}}{(d_\mathcal{K}^2/2)} \leq  \epsilon_1
\end{align*} where the last inequality uses $\hat{r}_{M}(0)\overset{a.s.}{\to}\gamma(0)$ by~\ref{cond:rM_pwConvergence}, as well as $(1-|\alpha_j|)\geq d_{\mathcal{K}}$ and $(1-|\alpha|)\geq d_{\mathcal{K}}/2$ from the definitions of $d_{\mathcal{K}}$ and $r_{\mathcal{K}}$. Since $\epsilon_1>0$ is arbitrary, the proof is complete.
\end{proof}

\begin{prop}\label{prop:measbdd}
Let $r_M$ be an initial autocovariance estimator satisfying \ref{cond:rM_pwConvergence}--\ref{cond:rM:even}. 
Let $\hat{\mu}_{\delta}^{\phi_M}$ be the representing measure of the weighted moment LS estimator $\Pi^{\phi_{ M}}(r_M;\delta)$ where $\delta>0$ and the weight function $\phi_{M}$ satisfies Assumption \ref{cond:phi_M} and \ref{cond:phi_M_asymp}.  We have $\limsup_{M\to\infty} \hat{\mu}_{\delta}^{\phi_M} ([-1, 1]) \le C_\mu$ almost surely for a constant $C_{\mu}<\infty $.
\end{prop}
\begin{proof}

By Lemma \ref{lem:inner_product_rep} and Proposition \ref{prop:weighted_opt_eq}, we have
\begin{align*}
\langle\Pi^{\phi_M}(r_M;\delta), \Pi^{\phi_M}(r_M;\delta)\rangle_{\phi_M} &=\int_{[-1,1]}\langle  x_\alpha, \Pi^{\phi_M}(r_M;\delta) \rangle_{\phi_M} \,\hat{\mu}_{\delta}^{\phi_M}(d \alpha)\\
&=\int_{[-1,1]}\langle x_\alpha, r_M \rangle_{\phi_M}\, \hat{\mu}_{\delta}^{\phi_M}(d \alpha).
\end{align*}
Since $\operatorname{Supp}(\hat{\mu}_{\delta}^{\phi_M}) \subseteq [-1+\delta,1-\delta]$, we have
\begin{align}\label{prop_measbdd:eq1}
    \langle\Pi^{\phi_M}(r_M;\delta), \Pi^{\phi_M}(r_M;\delta)\rangle_{\phi_M} \le \sup_{\alpha \in [-1+\delta,1-\delta]} |\langle x_\alpha, r_M  \rangle_{\phi_M}| \,\hat{\mu}_{\delta}^{\phi_M} ([-1, 1]).
\end{align}

For $\alpha \in [-1+\delta, 1-\delta]$, define $g_{M\alpha}(k)$ to be the inverse Fourier transform of $\hat{g}_{M\alpha}(\omega) = K(\alpha,\omega) \phi_M^{-2}(\omega)$ such that
$\langle x_\alpha, r_M  \rangle_{\phi_M} = \sum_{k\in\Z} g_{M\alpha}(k) r_M(k)$. Moreover let $h_M(k)$ be the inverse Fourier transform of $\phi_M^{-2}(\omega)$. We have
\begin{align*}
|\langle x_\alpha, r_M  \rangle_{\phi_M}|&\le  \sum_{k \in \Z} |g_{M\alpha} (k) | | r_M(k)| \\
& = |g_{M\alpha}(0) r_M(0)| + 2\sum_{k=1}^\infty |g_{M\alpha} (k) | | r_M(k)| \\
&\le |g_{M\alpha}(0) r_M(0)|+2r_M(0) \sum_{k=1}^\infty \frac{C_{M\alpha}}{k^2} 
\end{align*}
where for the last inequality we use the bound in \eqref{prop_c:gmk_bound} in the proof of Proposition \ref{prop:prop_c} and $|r_M(k)| \le r_M(0)$ by \ref{cond:rM:even}. Since $g_{M\alpha}(k) = (x_\alpha * h_M)(k)$, we have $|g_{M\alpha}(0)| = |\sum_{y \in \Z} h_M(y) \alpha^{|y|} |\le \|h_M\|_\infty \frac{1+|\alpha|}{1-|\alpha|}$, and therefore
\begin{align*}
    |\langle x_\alpha, r_M  \rangle_{\phi_M}| 
    &\le \|h_M\|_\infty \frac{1+|\alpha|}{1-|\alpha|}r_M(0) + 4r_M(0) C_{M\alpha}\\
    &\le \frac{2}{\delta \inf_{\omega \in [-\pi,\pi]} \phi_M^{2}(\omega)} r_M(0)+4 r_M(0) C_{M\alpha}
\end{align*}for $\alpha\in[-1+\delta,1-\delta]$, where we use $|h_M(k)|=\bigg|(2\pi)^{-1}\int_{[-\pi,\pi]}\frac{1}{\phi_M^2(\omega)}\exp(ik\omega)\,d\omega\bigg|\leq (2\pi)^{-1}\int_{[-\pi,\pi]}\frac{1}{\phi_M(\omega)^2}\,d\omega\leq \frac{1}{\underset{\omega\in[-\pi,\pi]}{\inf}\,\phi_{M}^2(\omega)}$, $\forall k\in\mathbb{Z}$ to bound $\|h_{M}\|_{\infty}$.

As in Equation \eqref{prop_c:C_Ma_bound} in Proposition~\ref{prop:prop_c}, we have
$\limsup_{M\to\infty} C_{M\alpha} \le C_\alpha =  \frac{10}{(1-|\alpha|)^{6}}\{\frac{1}{c_\phi^2} + \frac{4 c'_\phi}{c_\phi^3} +   \frac{6c'_\phi}{c_\phi^4} + \frac{2c'_\phi}{c_\phi^3} \} $. For $\alpha \in [-1+\delta,1-\delta]$, $C_\alpha$ is further bounded by $c_2 := 10\delta^{-6}\{\frac{1}{c_\phi^2} + \frac{4 c'_\phi}{c_\phi^3} +   \frac{6c'_\phi}{c_\phi^4} + \frac{2c'_\phi}{c_\phi^3} \} $, which is independent of $\alpha$.
In particular,
\begin{align}\label{prop_mu_bound:ineq1}
    \limsup_{M\to\infty}\sup_{\alpha\in [-1+\delta, 1-\delta]} |\langle x_\alpha, r_M  \rangle_{\phi_M}| \le \frac{2}{\delta c_\phi^2}\gamma(0) + 4\gamma(0) c_2.
\end{align} On the other hand,
\begin{align*}
     \langle\Pi^{\phi_M}(r_M;\delta), \Pi^{\phi_M}(r_M;\delta)\rangle_{\phi_M} 
     & = 
     \int \langle  x_\alpha,\Pi^{\phi_M}(r_M;\delta) \rangle_{\phi_M} \,\hat{\mu}_{\delta}^{\phi_M}(d\alpha)
     \\
     &=\int \int \langle x_\alpha, x_{\alpha'}\rangle_{\phi_M}\,\hat{\mu}_{\delta}^{\phi_M}(d \alpha')  \,\hat{\mu}_{\delta}^{\phi_M}(d \alpha).
\end{align*}
Thus
\begin{align}\label{prop_measbdd:eq2}
    \langle\Pi^{\phi_M}(r_M;\delta), \Pi^{\phi_M}(r_M;\delta)\rangle_{\phi_M} \ge\inf_{\alpha \in [-1+\delta,1-\delta]} \langle x_\alpha, x_{\alpha'}  \rangle_{\phi_M} \,\hat{\mu}_{\delta}^{\phi_M} ([-1, 1])^2.
\end{align} We have \begin{align}
    \underset{\alpha,\alpha'\in[-1+\delta,1-\delta]}{\inf}\,\langle x_\alpha, x_{\alpha'} \rangle_{\phi_M} &= \underset{\alpha,\alpha'\in[-1+\delta,1-\delta]}{\inf}\; \frac{1}{2\pi} \int_{[-\pi,\pi]} \frac{K(\alpha,\omega)K(\alpha',\omega)}{\phi_M^2(\omega)}\,d\omega\nonumber\\
    &\geq \underset{\alpha,\alpha'\in[-1+\delta,1-\delta]}{\inf}
    \;\left(\frac{1-|\alpha|}{1+|\alpha|}\right)\left(\frac{1-|\alpha'|}{1+|\alpha'|}\right)\frac{1}{2\pi}\int_{[-\pi,\pi]}\frac{1}{\phi_{M}^2(\omega)}\,d\omega\nonumber\\
    & \ge \frac{\delta^2}{4 \underset{\omega \in [-\pi, \pi]}{\inf}\phi_M^2(\omega) }\nonumber\\
    &>0 \label{eq:liminf}
\end{align} where we used $K(\alpha,\omega) \ge \frac{1-|\alpha|}{1+|\alpha|}$ uniformly over $\omega\in[-\pi,\pi]$ from Lemma~\ref{lem: K_bound}. Therefore, combining \eqref{prop_measbdd:eq1} and \eqref{prop_measbdd:eq2}, we have
\begin{align}\label{prop_mu_bound:muhat_bound}
     \hat{\mu}_{\delta}^{\phi_M} ([-1, 1]) \le \frac{\sup_{\alpha \in [-1+\delta,1-\delta]} |\langle x_\alpha, r_M  \rangle_{\phi_M}|}{\inf_{\alpha \in [-1+\delta,1-\delta]} \langle x_\alpha, x_\alpha' \rangle_{\phi_M}}
\end{align} 

Finally, from the second to last inequality in~\eqref{eq:liminf} we have \begin{align}\label{prop_mu_bound:ineq2}
     \liminf_{M\to\infty}  \inf_{\alpha\in [-1+\delta, 1-\delta]}  |\langle x_\alpha, x_\alpha' \rangle_{\phi_M} | \ge \frac{4}{\delta^2 c_\phi^2}.
\end{align} Combining \eqref{prop_mu_bound:ineq1}, \eqref{prop_mu_bound:muhat_bound},
and \eqref{prop_mu_bound:ineq2}, we obtain the result.

\end{proof}

\subsection{Proof of Lemma~\ref{lem:deterministicPhi_M}}

\begin{proof}
    First, we show $\tilde{\phi}_{\delta M}$ is infinitely differentiable. Note $\Pi(r_M;\delta)(0) = \int \alpha^{0} \hat{\mu}_{\delta M} = \hat{\mu}_{\delta M} ([-1,1])$. In the case $\hat{\mu}_{\delta M}([-1,1])=0$, we have $\tilde{\phi}_{\delta M}\equiv1$ for all $\omega \in[-\pi,\pi]$, and $\tilde{\phi}_{M}(\omega)$ is differentiable. 
    Otherwise, from~\ref{cond:rM_finiteSupport} and~\ref{cond:rM:even}, we have $|\text{Supp}(\hat{\mu}_{\delta M})|$ is finite and $\text{Supp}(\hat{\mu}_{\delta M})\subset (-1,1)$ from Proposition 6 in \citet{berg2023efficient} or Proposition~\ref{prop:finiteSupport} in this work with weight function $\phi\equiv 1$. Thus $\phi_{\delta M}(\omega) = \int K(\alpha,\omega) \hat{\mu}_\delta (d\alpha) =  \sum_{\alpha\in \text{Supp}(\hat{\mu}_{\delta M})} \hat{\mu}_{\delta}(\{\alpha\}) K(\alpha,\omega)$ where the number of terms in the summation is finite. Since $K(\alpha,\omega)$ is infinitely differentiable with respect to $\omega$ for each $\alpha\in(-1,1)$, $\phi_{\delta M}$ is also infinitely differentiable. Then in the case $\hat{\mu}_{\delta M}([-1,1])\ne 0$, $\tilde{\phi}_{\delta M}(\omega)=\phi_{\delta M}(\omega)/\hat{\mu}_{\delta}([-1,1])$ is also infinitely differentiable.

    Next, we bound $\tilde{\phi}_{M}(\omega)$. In the case $\hat{\mu}_{\delta M}([-1,1])=0$, we have $1= \underset{\omega\in[-\pi,\pi]}{\inf}\,\tilde{\phi}_{\delta M}(\omega)= \underset{\omega\in[-\pi,\pi]}{\sup}\,\tilde{\phi}_{\delta M}(\omega)=1$. Otherwise, we have $\hat{\mu}_{\delta M}([-1,1])>0$ and \begin{align*}
        \underset{\omega\in[-\pi,\pi]}{\inf}\,\tilde{\phi}_{\delta M}(\omega)&=\underset{\omega\in[-\pi,\pi]}{\inf}\,\frac{1}{\hat{\mu}_{\delta M}([-1+\delta,1-\delta])}\int_{[-1+\delta,1-\delta]}K(\alpha,\omega)\,\hat{\mu}_{\delta M}(d\alpha)\\
        &\geq  \frac{1}{\hat{\mu}_{\delta M}([-1+\delta,1-\delta])}\int_{[-1+\delta,1-\delta]}\underset{\omega\in[-\pi,\pi]}{\inf}\,\underset{\alpha'\in[-1+\delta,1-\delta]}{\inf}\,K(\alpha',\omega)\,\hat{\mu}_{\delta M}(d\alpha)\\
        &=\frac{\hat{\mu}_{\delta M}([-1+\delta,1-\delta]}{\hat{\mu}_{\delta M}([-1+\delta,1-\delta]}\frac{\delta}{2-\delta}=\frac{\delta}{2-\delta}.
    \end{align*} Similarly, still in the case $\hat{\mu}_{\delta M}([-1,1])>0$, we have $\underset{\omega\in[-\pi,\pi]}{\sup}\,\tilde{\phi}_{\delta M}(\omega)\leq \frac{2-\delta}{\delta}$. Since $0<\frac{\delta}{2-\delta}\leq 1\leq \frac{2-\delta}{\delta}<\infty$, $\tilde{\phi}_{\delta M}$ satisfies~\ref{cond:phi_M} with $c_{0M}=\frac{\delta}{2-\delta}$ and $c_{1M}=\frac{2-\delta}{\delta}$. This shows~\ref{cond:phi_M} holds for $\tilde{\phi}_{\delta M}$. We note that the bounding constants $c_{0M}$ and $c_{1M}$ are deterministic and do not depend on $M$.

    Now we show~\ref{cond:phi_M_asymp} also holds. From the above, taking $c_{\phi}=\frac{\delta }{2-\delta}$, we have \begin{align*}\liminf_{M\to\infty} \inf_{\omega \in [-\pi,\pi]} \tilde{\phi}_{\delta M}(\omega) \geq c_{\phi}.
    \end{align*} Next, we obtain a valid choice of $c_{\phi}'$ in~\ref{cond:phi_M_asymp}. In the case $\hat{\mu}_{\delta M}([-1+\delta,1-\delta])=0$, $\tilde{\phi}_{\delta M}(\omega)=1$ for all $\omega \in[-\pi,\pi]$, and so $\tilde{\phi}'(\omega)=\tilde{\phi}_{\delta M}''(\omega)=0$, for all $\omega\in [-\pi,\pi]$. Otherwise, in the case $\hat{\mu}_{\delta M}([-1+\delta,1-\delta])>0$, we have
\begin{align*}
    \phi'_{\delta M}(\omega) = \sum_{\alpha\in   \text{Supp}(\hat{\mu}_{\delta M})} \hat{\mu}_{\delta M}(\alpha) \frac{d}{d\omega}K(\alpha,\omega)\quad \mbox{and} \quad
    \phi''_{\delta M}(\omega) = \sum_{\alpha\in   \text{Supp}(\hat{\mu}_{\delta M})} \hat{\mu}_{\delta M}(\alpha)\frac{d^2}{d\omega^2}K(\alpha,\omega).
\end{align*} where the numbers of terms in the summations are finite.
Since $\text{Supp}(\hat{\mu}_{\delta M})\subset [-1+\delta,1-\delta]$, we have from Lemma \ref{lem: K_bound} that $\max\{\|\phi_{\delta M}\|_\infty, \|\phi'_{\delta M}\|_\infty,\|\phi''_{\delta M}\|_\infty\} \le 10\delta^{-6}  \hat{\mu}_{\delta}([-1+\delta,1-\delta])$. Then $\tilde{\phi}_{\delta M}(\omega)=\frac{\phi_{\delta M}(\omega)}{\hat{\mu}_{\delta}([-1+\delta,1-\delta])}$ satisfies $\max\{\|\tilde{\phi}_{\delta M}\|_\infty, \|\tilde{\phi}'_{\delta M}\|_\infty,\|\tilde{\phi}''_{\delta M}\|_\infty\} \le 10\delta^{-6}$. Now since $0<\delta\leq 1$, we have the deterministic bound $\max\{\|\tilde{\phi}_{\delta M}\|_\infty, \|\tilde{\phi}'_{\delta M}\|_\infty,\|\tilde{\phi}''_{\delta M}\|_\infty\}\leq 10\delta^{-6}$ in each of the cases $\hat{\mu}_{\delta M}([-1+\delta,1-\delta])=0$ and $\hat{\mu}_{\delta M}([-1+\delta,1-\delta])>0$. Thus~\ref{cond:phi_M_asymp} holds for $\tilde{\phi}_{\delta M}$ with $c_{\phi}=\frac{\delta}{2-\delta}$ and $c_{\phi}'=10\delta^{-6}$.
\end{proof}
\subsection{Proof of Theorem~\ref{thm:l2conv}}

\begin{proof}

We let $\hat{\mu}_\delta^{\phi_M}$ be the representing measure of $\Pi^{\phi_M}(r_M;\delta)$.
%
By taking $f=\gamma$ and $\phi = \phi_M$ in Lemma \ref{lem:weightedl2_ineq}, we have
\begin{align*}
0 \leq \|\Pi^{\phi_M} (r_M; \delta)-\gamma\|^2_{\phi_M} 
& \leq-\int \langle x_\alpha, r-\gamma \rangle_{\phi_M} \mu_\gamma(d \alpha)+\int\left\langle x_\alpha, r-\gamma\right\rangle_{\phi_M} \hat{\mu}_\delta^{\phi_M}(d \alpha)  \\
&\le \sup_{\alpha \in [-1+\delta, 1-\delta]} |\langle x_\alpha, r-\gamma\rangle_{\phi_M}| \{\mu_\gamma([-1,1]) + \hat{\mu}_\delta^{\phi_M}([-1,1])\}
\end{align*}
as $\operatorname{Supp}(\mu_\gamma) \subseteq [-1+\delta,1-\delta]$ by the assumption on the choice of $\delta$. 
Taking $M$-limits to both sides, we have
\begin{align*}
&\limsup_{M\to\infty} \|\Pi^{\phi_M} (r_M; \delta)-\gamma\|^2_{\phi_M} \\
&\le \limsup_{M\to\infty}\sup_{\alpha \in [-1+\delta, 1-\delta]} |\langle x_\alpha, r-\gamma\rangle_{\phi_M}| \{\mu_\gamma([-1,1]) + \hat{\mu}_\delta^{\phi_M}([-1,1])\}\\
&\leq \limsup_{M\to\infty}\sup_{\alpha \in [-1+\delta, 1-\delta]} |\langle x_\alpha, r-\gamma\rangle_{\phi_M}| \mu_\gamma([-1,1]) \\&+ \left\{\limsup_{M\to\infty}\sup_{\alpha \in [-1+\delta, 1-\delta]} |\langle x_\alpha, r-\gamma\rangle_{\phi_M}|\right\}\left\{\limsup_{M\to\infty}\hat{\mu}_\delta^{\phi_M}([-1,1])\right\}
\end{align*}
and the result follows by Propositions \ref{prop:prop_c} and Proposition \ref{prop:measbdd} as $\limsup_{M\to\infty}\sup_{\alpha \in [-1+\delta, 1-\delta]} |\langle x_\alpha, r-\gamma\rangle_{\phi_M}| = 0$ and $\limsup_{M\to\infty}\hat{\mu}_\delta^{\phi_M}([-1,1]) \le C_\mu <\infty$, $P_x$-almost surely for any initial condition $x \in \mathsf{X}$.

For the convergence in $\ell_2$ norm, we have
\begin{align*}
    \|\Pi^{\phi_M} (r_M; \delta)-\gamma\|^2_{\phi_M} 
    &= \frac{1}{2\pi} \int_{[-\pi,\pi]} \frac{(\phi^W_{\delta M}(\omega) - \phi_\gamma(\omega))^2}{\phi_M^2(\omega)} d\omega\\
    &\ge \frac{1}{\underset{\omega \in [-\pi,\pi]}{\sup}\, \phi_M^{2}(\omega)}\ \frac{1}{2\pi} \int_{[-\pi,\pi]}  (\phi^W_{\delta M}(\omega) - \phi_\gamma(\omega))^2 d\omega\\
    &= (\|\phi_M\|_\infty)^{-2}\|\Pi^{\phi_M} (r_M; \delta)-\gamma\|^2_2
\end{align*}
where we denote the Fourier transform of $\Pi^{\phi_M}(r_M;\delta)$ as $\phi^W_{\delta M}$ and we use Parseval's equality for the last equality. Taking limits to both sides, we obtain
\begin{align*}
   (c_{\phi}')^{2} \limsup_{M\to\infty} \|\Pi^{\phi_M} (r_M; \delta)-\gamma\|^2_{\phi_M}  \ge \limsup_{M\to\infty} \|\Pi^{\phi_M} (r_M; \delta)-\gamma\|^2_{2}\ge 0.
\end{align*}
In particular, $\limsup_{M\to\infty} \|\Pi^{\phi_M} (r_M; \delta)-\gamma\|^2_{2}= 0$ $P_x$ almost surely, as $0<(c_\phi')^{2}<\infty$.
\end{proof}

\subsection{Proof of Theorem~\ref{thm:weightedl1}}

\begin{proof}
	First, we note $\gamma\in\ell_1(\mathbb{Z})$ from Lemma 3 of~\citet{berg2023efficient}. Also, we have $\Pi^{\phi_{M}}(r_M;\delta)\in \ell_1(\mathbb{Z})$ for each $M$ from Assumptions~\ref{cond:rM_finiteSupport}--\ref{cond:rM:even} and~\ref{cond:phi_M}, and Lemma~\ref{lem: l1projection}. Thus $\phi_{\delta M}^W(\omega)$ and $\phi_{\gamma}(\omega)$ are well-defined for each $\omega\in[-\pi,\pi]$. Also, 1. implies 2., since
	\begin{align*}
		|\phi_{\delta M}^W(\omega) - \phi(\omega)| 
		&= |\sum_{k=-\infty}^\infty  \Pi^{\phi_{M}}(r_M;\delta )(k)e^{-i\omega k} - \gamma(k)e^{-i\omega k}|\\
		&\le \|\Pi^{\phi_{M}}(r_M; \delta) - \gamma\|_1 
	\end{align*} uniformly over $\omega\in[-\pi,\pi]$. Additionally, 2. implies 3., since \begin{align*}
	    |\sigma^2(\Pi^{\phi_M}(r_M;\delta))-\sigma^2(\gamma)|&=|\phi_{\delta M}^W(0)-\phi_{\gamma}(0)|\\&\leq \underset{\omega\in[-\pi,\pi]}{\sup}\,|\phi_{\delta M}^W(\omega)-\phi_{\gamma}(\omega)|
	\end{align*}
	Therefore, it is sufficient to show 1. $\lim_{M\to\infty} \|\Pi^{\phi_M}(r_M; \delta) - \gamma\|_1 = 0$, $P_x$-almost surely.	Let $\epsilon>0$ be given. From the given conditions, we have from 2. of Theorem~\ref{thm:l2conv} that $\|\Pi^{\phi_M}(r_M;\delta)-\gamma\|_2\to 0$ a.s. as $M\to\infty$. Therefore, $\underset{M\to\infty}{\lim}\,\hat{\mu}_{\delta}^{\phi_M}([-1,1])=\underset{M\to\infty}{\lim}\,\Pi^{\phi_M}(r_M;\delta)(0)=\gamma(0)$ almost surely. Choose $L>0$ such that
	\begin{align}
		\frac{(1-\delta)^L}{\delta}4\gamma(0)\le \epsilon.
	\end{align}
We have
	\begin{align*}
		 \|\Pi^{\phi_M}(r_M; \delta) - \gamma\|_1 \le \sum_{k; |k|<L} |\Pi^{\phi_M}(r_M; \delta)(k) - \gamma(k)|+\sum_{k; |k|\ge L} |\Pi^{\phi_M}(r_M; \delta)(k)| +  \sum_{k; |k|\ge L}|\gamma(k)|
	\end{align*}
We first bound the second and third terms. Since $\operatorname{Supp}(\hat{\mu}_\delta^{\phi_M}) \subseteq [-1+\delta,1-\delta]$, we have
 \begin{align*}
		\sum_{k; |k|\ge L} |\Pi^{\phi_M}(r_M; \delta)(k)|
		 =\sum_{k; |k|\ge L} |\int \alpha^k \hat{\mu}_{\delta}^{\phi_M}(d\alpha)| 
		 \le \sum_{k; |k|\ge L}\int  |\alpha|^k \hat{\mu}_{\delta}^{\phi_M} (d\alpha) 
		 \le 2 \sum_{k; k \ge L} (1-\delta)^k  \int  \hat{\mu}_{\delta}^{\phi_M} (d\alpha)	
	\end{align*}
therefore \begin{align*}
    \underset{M\to\infty}{\lim\sup}\,\sum_{k; |k|\ge L} |\Pi^{\phi_M}(r_M; \delta)(k)| \le \underset{M\to\infty}{\lim\sup}\;2 \frac{(1-\delta)^L}{\delta}\hat{\mu}_{\delta}^{\phi_M}([-1,1])\le 2\gamma(0)\frac{(1-\delta)^L}{\delta}
\end{align*}
		Similarly, since $\operatorname{Supp}(\mu_\gamma)\subseteq [-1+\delta, 1-\delta]$,
		\begin{align*}
			\sum_{k; |k|\ge L} |\gamma(k)| \le 2 \frac{(1-\delta)^L}{\delta} \mu_\gamma([-1,1])
		\end{align*}
	Therefore,
	\begin{align*}
		 \|\Pi^{\phi_M}(r_M; \delta) - \gamma\|_1 
		 &\le \sum_{k; |k|<L} |\Pi^{\phi_M}(r_M; \delta)(k) - \gamma(k)|+2 \frac{(1-\delta)^L}{\delta} \{\hat{\mu}_\delta^{\phi_M}([-1,1])	+\mu_\gamma([-1,1])\}\\
		 &\le \left[\sum_{k; |k|<L} \{\Pi^{\phi_M}(r_M; \delta)(k) - \gamma(k)\}^2\right]^{1/2} \left\{\sum_{k;|k|\le L} 1\right\}^{1/2}\\&+2 \frac{(1-\delta)^L}{\delta} \{\hat{\mu}_\delta^{\phi_M}([-1,1])	+\mu_\gamma([-1,1])\}\\
		 &\le \|\Pi^{\phi_M}(r_M; \delta) - \gamma\|_2 \sqrt{2L+1} +2 \frac{(1-\delta)^L}{\delta} \{\hat{\mu}_\delta^{\phi_M}([-1,1])	+\mu_\gamma([-1,1])\}
	\end{align*}	
 where we use Holder's inequality for the second inequality. Taking the limit as $M\to\infty$ to both sides and using part 2. of Theorem~\ref{thm:l2conv}, we have
\begin{align*}
	\limsup_M \|\Pi^{\phi_M}(r_M; \delta) - \gamma\|_1 \le \frac{(1-\delta)^L}{\delta}4\gamma(0) \le \epsilon, \,\, a.s.
\end{align*}
by the choice of $L$. Since $\epsilon$ is arbitrary, we have the desired result.
\end{proof}

\begin{lem}\label{lem:unweighted_spectr_bound}
    Suppose $r\in\ell_2(\mathbb{Z},\mathbb{R})$ and suppose $r(0)>0$. Let $\phi_{\delta}$ denote the Fourier transform of the unweighted projection $\Pi(r;\delta)$. Then $\|\Pi(r;\delta)\|_2>0$ and $0<\underset{\omega\in[-\pi,\pi]}{\inf}\phi_{\delta}(\omega)\leq \underset{\omega\in[-\pi,\pi]}{\sup}\phi_{\delta}(\omega)<\infty$.
\end{lem}

\begin{proof}

Since $\Pi(r;\delta)\in\mathscr{M}_{\infty}(\delta)\cap \ell_2(\mathbb{Z},\R)$, we have $\Pi(r;\delta)\in\ell_1(\mathbb{Z},\R)$ from Lemma 3 of~\citet{berg2023efficient}. Thus from Lemma \ref{lem: fourier_xa} and Remark \ref{remark:FT_l1seq}, $\phi_{\delta}(\omega)=\int_{[-1+\delta,1-\delta]}\,K(\alpha,\omega)\,\mu(d\alpha)$ for each $\omega\in[-\pi,\pi]$ where $\mu$ denotes the representing measure for $\Pi(r;\delta)$. We have $\Pi(r;\delta)(0)=\mu([-1,1])<\infty$ from $\Pi(r;\delta)\in\mathscr{M}_{\infty}(\delta)$. Also, by applying Proposition~\ref{prop:weighted_opt_ineq} with the weight function $\phi\equiv 1$, we have
\begin{align*}\mu([-1,1])=\int_{[-1,1]}1\,\mu(d\alpha)=\braket{x_0,\Pi(r;\delta)}=\braket{x_{0},\Pi(r;\delta)}_{\phi}\geq \braket{x_{0},r}_{\phi}=\braket{x_{0},r}=r(0)>0.
\end{align*}

Thus, \begin{align*}
    \underset{\omega\in[-\pi,\pi]}{\inf}\,\phi_{\delta}(\omega)&=\underset{\omega\in[-\pi,\pi]}{\inf}\int_{[-1+\delta,1-\delta]}K(\alpha,\omega)\,\mu(d\alpha)\\&\geq \mu([-1+\delta,1-\delta])\underset{\alpha\in [-1+\delta,1-\delta]}{\inf}\,\underset{\omega\in[-\pi,\pi]}{\inf}\,K(\alpha,\omega)\\&\geq\mu([-1+\delta,1-\delta])\frac{\delta}{2-\delta}>0.
\end{align*} from Lemma~\ref{lem: K_bound}.

Similarly, using Lemma~\ref{lem: K_bound} we obtain \begin{align*}
    \underset{\omega\in[-\pi,\pi]}{\sup}\,\phi_{\delta}(\omega)&\leq\mu([-1+\delta,1-\delta])\frac{2-\delta}{\delta}<\infty.
\end{align*} 

\end{proof}

\newpage
\section{Supplementary Tables}\label{supp_sec: 5_tables}

\begin{table}[H]
\centering
\setlength\tabcolsep{3pt}
\caption{\textit{Estimated average mean squared error (s.e.) for the asymptotic variance estimators and mean integrated squared error (s.e.) for the spectral density estimators from AR1 example when $M=40000$}}

\begin{subtable}[t]{\textwidth}
\small
\centering
\begin{threeparttable}
\caption{Asymptotic variance mean squared error\tnote{1}}
\begin{tabular}{l|c|ccccc}
  \hline
 & AR1(MLE) & OBM & Init-Conv & IO & mLS-uw & mLS-w \\ 
  \hline
AR1($-0.9$) & 0.05 (0.00) & 5.84 (0.38) & 324.49 (13.17) & 4.30 (0.29) & 4.17 (0.26) & 0.88 (0.11) \\ 
  AR1($0.9$) & 19.05 (1.23) & 90.19 (4.98) & 44.81 (3.07) & 51.22 (3.61) & 41.34 (3.20) & 31.11 (2.44) \\ 
   \hline
\end{tabular}
 \begin{tablenotes}
\footnotesize
\item[1] values for AR1($\rho=-0.9$) are scaled by $10^4$
    \end{tablenotes}
\end{threeparttable}
\vspace{1em}

\begin{threeparttable}
\caption{spectral density mean integrated squared error}
\begin{tabular}{l|c|cccc}
  \hline
 & AR1(MLE) & Bart & IO & mLS-uw & mLS-w \\ 
  \hline
AR1($-0.9$) & 0.33 (0.02) & 1.34 (0.05) & 0.86 (0.04) & 0.61 (0.03) & 0.56 (0.06) \\ 
  AR1($0.9$) & 0.31 (0.02) & 1.25 (0.04) & 0.84 (0.04) & 0.57 (0.03) & 0.44 (0.03) \\ 
   \hline
\end{tabular}
\end{threeparttable}
\end{subtable}

\end{table}

\begin{table}[H]
\centering
\setlength\tabcolsep{3pt}
\caption{\textit{Estimated average mean squared error (s.e.) for the asymptotic variance estimators and mean integrated squared error (s.e.) for the spectral density estimators from Bayesian Lasso example when $M=40000$}}

\begin{subtable}[t]{\textwidth}
\small
\centering
\begin{threeparttable}
\caption{Asymptotic variance mean squared error\tnote{1}}
\begin{tabular}{l|cccccc}
  \hline
& Bart & OBM & Init-Conv & IO & mLS-uw & mLS-w  \\ 
  \hline
$\beta_0$ & 121.65 (5.73) & 124.20 (5.80) & 45.26 (3.40) & 44.43 (3.13) & 42.49 (3.12) & 38.52 (2.78) \\ 
  $\beta_1$ & 104.97 (5.27) & 107.23 (5.33) & 41.73 (2.79) & 37.98 (2.35) & 39.43 (2.58) & 35.53 (2.40) \\ 
  $\beta_2$ & 140.46 (7.74) & 143.81 (7.89) & 66.44 (4.32) & 62.18 (3.86) & 63.77 (4.29) & 58.84 (3.93) \\ 
  $\beta_3$ & 35.85 (1.89) & 36.60 (1.92) & 14.43 (0.97) & 14.85 (0.88) & 14.17 (0.92) & 13.06 (0.83) \\ 
  $\beta_4$ & 24.86 (1.51) & 25.37 (1.54) & 12.08 (0.82) & 10.87 (0.74) & 11.37 (0.77) & 10.48 (0.70) \\ 
  $\beta_5$ & 6.55 (0.34) & 6.66 (0.35) & 3.07 (0.19) & 3.07 (0.19) & 3.19 (0.21) & 3.06 (0.20) \\ 
  $\beta_6$ & 52.63 (2.79) & 53.85 (2.83) & 23.18 (1.65) & 21.94 (1.55) & 22.22 (1.64) & 19.79 (1.42) \\ 
  $\beta_7$ & 0.17 (0.01) & 0.17 (0.01) & 0.25 (0.02) & 0.33 (0.01) & 0.15 (0.01) & 0.14 (0.01) \\ 
  $\sigma^2$ & 82.46 (4.84) & 83.46 (4.88) & 190.61 (12.53) & 248.73 (11.42) & 115.34 (8.15) & 95.81 (7.08) \\ 
   \hline
\end{tabular}
 \begin{tablenotes}
\footnotesize
\item[1] values for $\beta_0$-$\beta_7$ are scaled by $10^4$, and $\sigma^2$ is scaled by $10^2$
    \end{tablenotes}
\end{threeparttable}
\vspace{1em}

\begin{threeparttable}
\caption{spectral density mean integrated squared error\tnote{2}}
\begin{tabular}{l|cccc}
  \hline
 & Bart &IO & mLS-uw & mLS-w  \\ 
  \hline
$\beta_0$ & 64.30 (1.88) & 26.41 (1.08) & 20.43 (0.97) & 18.22 (0.89) \\ 
  $\beta_1$ & 54.78 (1.60) & 23.27 (0.86) & 18.45 (0.82) & 16.36 (0.78) \\ 
  $\beta_2$ & 72.54 (2.18) & 35.20 (1.35) & 28.50 (1.29) & 25.91 (1.24) \\ 
  $\beta_3$ & 27.77 (0.80) & 12.77 (0.44) & 9.67 (0.40) & 9.00 (0.38) \\ 
  $\beta_4$ & 22.97 (0.63) & 10.33 (0.41) & 8.29 (0.36) & 7.68 (0.35) \\ 
  $\beta_5$ & 4.15 (0.14) & 2.28 (0.09) & 1.89 (0.08) & 1.82 (0.08) \\ 
  $\beta_6$ & 26.12 (0.81) & 12.43 (0.52) & 10.01 (0.49) & 9.03 (0.45) \\ 
  $\beta_7$ & 0.81 (0.02) & 1.22 (0.02) & 0.37 (0.01) & 0.36 (0.01) \\ 
  $\sigma^2$ & 42.71 (1.08) & 77.27 (2.60) & 24.42 (0.83) & 22.35 (0.79) \\ 
   \hline
\end{tabular}
 \begin{tablenotes}
\footnotesize
\item[2] values for $\beta_0$-$\beta_7$ are scaled by $10^5$, and $\sigma^2$ is scaled by $10^2$
    \end{tablenotes}
\end{threeparttable}
\end{subtable}

\end{table}

\ifnum\pageoption=2
\putbib 
\end{bibunit}\fi

\ifnum\pageoption=3
\bibliographystyle{plainnat}  
\bibliography{bib}
\fi

\fi
\end{document}